\documentclass[11pt, letter]{article}
\usepackage{geometry}
\usepackage{cite}
\usepackage{listings}

\usepackage{amsmath, amsfonts, amsthm, amstext}
\makeatletter
\def\ExtendSymbol#1#2#3#4#5{\ext@arrow 0099{\arrowfill@#1#2#3}{#4}{#5}}
\def\RightExtendSymbol#1#2#3#4#5{\ext@arrow 0359{\arrowfill@#1#2#3}{#4}{#5}}
\def\LeftExtendSymbol#1#2#3#4#5{\ext@arrow 6095{\arrowfill@#1#2#3}{#4}{#5}}
\makeatother


\usepackage{tikz}
\usetikzlibrary{shapes,shadows}
\tikzstyle{abstractbox} = [
	draw=black,
	fill=white,
	rectangle,
	inner sep=10pt,
	style=rounded corners,
	drop shadow={
		fill=black,
		opacity=1
	}
]

\definecolor{mygreen}{rgb}{0,0.6,0}
\definecolor{mygray}{rgb}{0.5,0.5,0.5}
\definecolor{mymauve}{rgb}{0.58,0,0.82}

\tikzstyle{abstracttitle} =[fill=white]

\usetikzlibrary{arrows}

\usepackage{stmaryrd}
\usepackage{amsmath}
\usepackage{caption2}
\usepackage{epsfig}

\usepackage{mathrsfs}
\usepackage{amsfonts}
\usepackage{bbding}
\usepackage{pifont}
\usepackage{amssymb}

\usepackage{multirow}
\usepackage{epstopdf}
\usepackage{algorithmicx}
         \usepackage[ruled]{algorithm}
         \usepackage{algpseudocode}

\usepackage[colorlinks,linkcolor=red]{hyperref}

\usepackage{algorithm,algpseudocode,algorithmicx}

\usepackage{array}

\usepackage{url}

\usepackage{multirow}
\usepackage{color}

\newtheorem{fact}{Fact}[section]
\newtheorem{theorem}{Theorem}[section]
\newtheorem{definition}{Definition}[section]
\newtheorem{proposition}{Proposition}[section]
\newtheorem{lemma}{Lemma}[section]
\newtheorem{corollary}{Corollary}[section]
\newtheorem{claim}{Claim}[section]

\let\oldReturn\Return
\renewcommand{\Return}{\State\oldReturn}

\usetikzlibrary{arrows,positioning,shapes.multipart}
\tikzset{
    >=stealth',
    punkt/.style={
          rectangle,
           minimum height=2em,
           text centered},
    pil/.style={
           ->,
           thick,
           shorten <=2pt,
           shorten >=2pt,}
}

\begin{document}
%

\title{Optimal Key Consensus in Presence of Noise\thanks{This research was  supported in part  by NSFC Grant No. U1536205.}}


\date{}

\author{Zhengzhong Jin\footnote{School of Mathematical Sciences, Fudan University, Shanghai, China.
\texttt{zzjin13@fudan.edu.cn}. }
 \and Yunlei Zhao\footnote{School of Computer Science,   Fudan University, Shanghai,
 China.    \texttt{ylzhao@fudan.edu.cn}  }}

\date{}


\maketitle

\begin{abstract}

%

In this work, we abstract some key ingredients in previous LWE- and RLWE-based key exchange protocols, by introducing and formalizing the building tool, referred to as key consensus (KC) and its asymmetric variant AKC. KC and AKC allow  two communicating parties to reach consensus from close values obtained by some secure information exchange.
We then discover upper bounds on parameters for any KC and AKC, which bounds what could or couldn't be done  for the various parameters involved. KC and AKC are fundamental to lattice based cryptography,
in the sense that a list of cryptographic primitives based on LWR, LWE and RLWE
(including key exchange, public-key encryption, and more) can be modularly constructed from them.
As a conceptual contribution, this much simplifies the design and analysis of these cryptosystems in the future.

We then design and analyze both general and almost optimal  KC and AKC schemes, which are referred to as OKCN and AKCN respectively for presentation simplicity.
    Based on KC and AKC, we present generic constructions of key exchange  from LWR, LWE and RLWE. 
  The generic construction allows versatile instantiations with our OKCN and AKCN schemes,
  for which we elaborate on evaluating and choosing the concrete parameters in order to achieve an optimally-balanced performance among security, computational cost, bandwidth efficiency,  error rate, and operation simplicity.

%

\begin{itemize}
\item We propose the first  construction  of key exchange based on the learning with rounding (LWR) problem, to the best of our knowledge. The rounding in  LWR much  reduces the bandwidth.
    We provide a delicate approach to calculating  the error rate accurately, which  avoids introducing some cumbersome   inequalities  in the traditional ways of error rate estimation. The accuracy of error rate allows us to choose tight parameters.
    Our LWR-based key exchange costs 16.39 kilo-bytes (kB) with error rate $2^{-35}$ at the level of at least 128-bit quantum security.




\item When  applied to LWE-based cryptosystems, OKCN and AKCN can directly result in more practical schemes of key exchange and CPA-secure public-key encryption. 
To further save bandwidth, we make a thorough analysis of  the variant where  some least significant bits of LWE samples are cut off. For instance, on recommended parameters
  our LWE-based protocol (resp., Frodo) has 18.58kB (resp., 22.57kB)  bandwidth, 887.15kB (resp., 1060.32kB) matrix, at least 134-bit (resp., 130-bit) quantum security,  and  error rate $2^{-39}$ (resp., $2^{-38.9})$.

\item
When applied to RLWE-based cryptosystems,
we make a key observation   by proving that the errors in different positions in the shared-key are essentially  independent.  Based upon this observation, we present an extremely  simple and fast code, referred to as \emph{single-error correction} (SEC) code, to correct at least one bit error.  By equipping OKCN/AKCN with  the SEC code, we achieve the simplest RLWE-based key exchange (up to now)
\emph{with negligible error rate} for much longer shared-key size; for instance,  $837$-bit shared-key with bandwidth of about 0.4kB and  error rate  $2^{-69}$.





\end{itemize}

%
%

\end{abstract}

\newpage
{
\tableofcontents
}
\newpage

\section{Introduction}

Most public-key cryptosystems currently in use,
based on the hardness of solving (elliptic curve) discrete logarithm or factoring large integers, will be broken, if large-scale quantum computers are ever built.
 The arrival of such quantum computers is now believed by many scientists to be merely a significant engineering challenge, and is  estimated by engineers at IBM to be within the
next two  decades or so. Historically, it has taken almost two decades to deploy the modern public key cryptography infrastructure.  Therefore, regardless of whether we can estimate the exact time of the arrival of the quantum computing era, we must begin now to prepare our information security systems to be able to resist quantum computing \cite{NIST16}. In addition, for the content we want to protect over a period of 15 years or longer,  it becomes  necessary
to switch to post-quantum cryptography today. This has been recognized not only by the cryptography research community,  but also by standardization bodies and leading information companies, for example,  NSA \cite{NSA15},  NIST \cite{NIST16}, U.K. CESG \cite{CESG16},  the Tor project \cite{Tor16}, and  Google \cite{Google16}.

As noted in \cite{newhope15,cryptoeprint:2016:758}, in the majority of contexts the most critical asymmetric primitive to upgrade to post-quantum security is ephemeral key exchange (KE). KE plays a central role in modern cryptography, which bridges public-key cryptography and symmetric-key cryptography and can, in turn,  be used to build  CPA-secure public-key encryption (PKE) as well as  CCA-secure PKE in the random oracle (RO) model via the FO-transformation \cite{FO99b,peikert2014lattice,TU15}, and more.
U.K. CESG has also
expressed their preference for post-quantum algorithms (in particular, post-quantum KE schemes)  over
quantum technologies ``such as Quantum Key Distribution"
to counter the threat of quantum computing \cite{CESG16}.

Lattice-based cryptography is among the major mathematical  approaches to achieving security resistant to quantum attacks. 	For cryptographic usage,
	compared with the classic hard lattice problems such as SVP and CVP,
	the learning with errors (LWE) problem is proven to be much more versatile \cite{regev2009lattices}. Nevertheless,    LWE-based cryptosystems are usually less efficient, which   was then resolved by the introduction of the ring-LWE (RLWE) problem \cite{LYUBASHEVSKY2013ON}. In recent years,  large numbers of impressive works are developed from   LWE and RLWE, with (ephemeral) key exchange and public-key encryption being the study focus of this work  \cite{DXL14,peikert2014lattice,BCNS14,newhope15,bcd16,regev2009lattices,gentry2008trapdoors,lindner2010better,LYUBASHEVSKY2013ON,lyubashevsky2013toolkit,
poppelmann2013towards}. For an excellent survey of lattice-based cryptography, the reader is referred to \cite{adecade}.



%
%
%
%
%
%
%
%
%
%

Some celebrating progresses on  achieving practical LWE- and RLWE-based key exchange are made in recent years.
The performance of RLWE-based key exchange is significantly improved with NewHope \cite{newhope15}, which achieves 256-bit shared-key with error rate about $2^{-61}$.
The negligible error rate of  NewHope is achieved  by decoding the four-dimensional lattice $\tilde{D}_4$. Decoding the 24-dimensional Leech lattice is also recently considered in  \cite{leech}. But  decoding the four-dimensional  lattice $\tilde{D}_4$ has already been  relatively complicated and computationally less efficient.
Compared to LWE, the additional ring structure of RLWE helps to improve the efficiency of cryptosystems,
but the concrete hardness of RLWE remains less clear. The work
\cite{bcd16} proposes a key exchange protocol Frodo only based on LWE, and demonstrates that LWE-based key exchange can be practical as well.
Nevertheless, bandwidth of Frodo is relatively large, as Frodo uses about 22kB bandwidth for its recommended parameter set.
In addition, Frodo has relatively large error rates, and  cannot be directly used for PKE.
Whether further improvements on LWE- and RLWE-based key exchange, as well as CPA-secure PKE,  can be achieved remains an interesting  question of practical significance.

{One of the main technical  contributions in the works \cite{newhope15, bcd16, poppelmann2013towards}, among others, is the improvement and generalization of the key  reconciliation mechanisms \cite{peikert2014lattice,DXL14}.\footnote{{To our knowledge, the key reconciliation mechanism in \cite{peikert2014lattice} is the first that fits our KC definition (the mechanism in \cite{DXL14} requires  the distance be  of special types). The Lindner-Peikert mechanism implicitly presented for PKE \cite{lindner2010better} is the first that fits our AKC definition. The reader is referred to  \cite{newhope-simple} for a detailed survey on key exchange from LWE and RLWE.}} 
But the  key reconciliation mechanisms  were only previously used and analyzed, for both KE and PKE, in a \emph{non-black-box} way. This means,  for new  key reconciliation mechanisms developed in the future to be used for constructing lattice-based cryptosystems,  we need to make analysis  from scratch. Also, for the various parameters involved in key reconciliation,  the bounds on what could or couldn't be achieved are unclear. }

\subsection{Our Contributions}

In this work, we abstract some key ingredients in previous LWE- and RLWE-based key exchange protocols, by introducing and formalizing the building tool, referred to as key consensus (KC) and its asymmetric variant AKC. KC and AKC allow  two communicating parties to reach consensus from close values obtained by some secure information exchange,
such as exchanging their LWE/RLWE samples.
We then discover upper bounds on parameters for any KC and AKC, and  make comparisons between KC/AKC and fuzzy extractor \cite{fuzzyextractor}.
KC and AKC are fundamental to lattice based cryptography,
in the sense that a list of cryptographic primitives based on LWE or RLWE
(including key exchange, public-key encryption, and more) can be modularly constructed from them. As a conceptual contribution, this much simplifies the design and analysis of these cryptosystems in the future.

We then design and analyze both general and highly practical KC and AKC schemes, which are referred to as OKCN and AKCN respectively for presentation simplicity.
Our OKCN and AKCN schemes are optimal in a sense of achieving optimal balance among security, (computational and bandwidth) efficiency, and operation simplicity.  
    Firstly, the correctness constraints on parameters are almost the same as the upper bounds we discovered.
    Secondly, the generality of our schemes allows us to take optimal balance among parameters in order to choose parameters towards different goals. Thirdly, the operations involved are  simple.

Based on KC and AKC, we present generic constructions of key exchange from LWR,LWE and RLWE with delicate analysis of error rates. Then, for the instantiations of these generic constructions with our OKCN and AKCN schemes, we elaborate on   evaluating and choosing the concrete parameters in order to achieve an optimally-balanced performance among security, computational efficiency, bandwidth efficiency,  error rate, and operation simplicity.
\begin{itemize}

\item We propose the first  construction  of key exchange based on the learning with rounding (LWR) problem, to the best of our knowledge.  The rounding in  LWR much  reduces the bandwidth.
     We provide a delicate approach to calculating  the error rate accurately, which  avoids introducing some cumbersome   inequalities  in the traditional ways of error rate estimation. The accuracy of error rate allows us to choose tight parameters.
    Our LWR-based key exchange costs 16.39kB with error rate $2^{-35}$ at the level of at least 128-bit quantum security.





\item When  applied to LWE-based cryptosystems, OKCN and AKCN can directly result in more practical or well-balanced schemes of key exchange and CPA-secure public-key encryption.
To further save bandwidth, we make a thorough analysis of  the variant where  some least significant bits of LWE samples are cut off. 
We remark that cutting off  some least bits of LWE samples can only improve the actual security guarantee  in reality, but complicates the analysis of error rates. For instance, on recommended parameters
  our LWE-based protocol (resp., Frodo) has 18.58kB (resp., 22.57kB)  bandwidth, 887.15kB (resp., 1060.32kB) matrix, at least 134-bit (resp., 130-bit) quantum security,  and  error rate $2^{-39}$ (resp., $2^{-38.9})$.

\item
When applied to RLWE-based cryptosystems, to the best of our knowledge,
AKCN can lead to more (bandwidth) efficient schemes of CPA-secure PKE
and AKC-based KE with negligible error rate. In order to achieve 256-bit shared-key with  negligible error rate, we use the technique of NewHope by encoding and decoding  the four-dimensional lattice $\tilde{D}_4$, which results in a protocol named AKCN-4:1 that is (slightly) better than NewHope-simple \cite{newhope-simple} in bandwidth expansion.

However, encoding and decoding  $\tilde{D}_4$ is  relatively  complicated and computationally less efficient.  In this work, we make a key observation  on RLWE-based key exchange, by proving that the errors in different positions in the shared-key are almost independent.
Based upon this observation, we present an super simple and fast code, referred to as \emph{single-error correction} (SEC) code, to correct at least one bit error.  By equipping OKCN/AKCN with  the SEC code, we achieve the simplest (up to now) RLWE-based key exchange, from both OKCN and AKCN, with negligible error rate for much longer shared-key size; for instance,   $837$-bit shared-key with bandwidth of about 0.4kB and error rate  $2^{-69}$.

\end{itemize}

Finally, we briefly   discuss the applications of OKCN and AKCN  to public-key encryption, authenticated key exchange, key transport, and TLS. As a fundamental building tool for lattice-based cryptographic schemes, we suggest OKCN, AKCN and the various KE protocols based on them are of independent value. They   may possibly  find more applications in  more advanced cryptographic primitives (e.g., oblivious transfer,  IBE, ABE, FHE) from LWR, LWE and RLWE, by simplifying their design and analysis with versatile performance optimization or balance.


All the {main}  protocols developed in this work are implemented. The code and scripts, together with those  for evaluating concrete security and error rates,  are  available from  Github \url{http://github.com/OKCN}.

\subsection{Subsequent or Concurrent Works}
A PKE scheme, named Lizard, is recently presented  in \cite{CKLS16}. It is easy to see that the underlying key consensus mechanism of Lizard is just   instantiated from  our AKCN  scheme presented in Algorithm \ref{kcs:2} (see more details in Appendix \ref{sec-lizard}).\footnote{AKCN was introduced at  the second  Asian PQC Forum (APQC 2016) on November 28, 2016, in Seoul, Korea,  and was  filed and  publicly available even  earlier. The original version of Lizard was also presented at APQC 2016. But we note that, according to the video presentation  available from \url{http://www.pqcforum.org/}, the underlying key consensus mechanism of Lizard  presented as APQC 2016 was based on  the Lindner-Peikert mechanism \cite{lindner2010better}, not our AKCN mechanism as appeared  in the later ePrint report \cite{CKLS16}.}

 Recently, we notice that an AKC-based variant  of NewHope, named NewHope-simple, was presented in a  note \cite{newhope-simple}. In comparison,  NewHope-simple is still slightly inferior to  AKCN4:1-RLWE in bandwidth expansion (specifically, 256 vs. 1024 bits), and our SEC-equipped protocols are simpler, have lower error rates and  much longer shared-key sizes.

\section{Preliminaries}\label{sec-pre}
A string or value $\alpha$ means a binary one, and  $|\alpha|$ is its binary  length.
 For any real number $x$, $\lfloor x \rfloor$ denotes the largest integer that less than or equal to $x$, and
    $\lfloor x \rceil = \lfloor x + 1/2 \rfloor$.
For any positive integers $a$ and $b$, denote by $\mathsf{lcm}(a, b)$ the least common multiple of them. For any $i,j\in \mathbb{Z}$ such that $i<j$, denote by $[i,j]$ the set of integers $\{i,i+1,\cdots, j-1,j\}$. For any positive integer $t$,  we let $\mathbb{Z}_t$ denote $\mathbb{Z} / t\mathbb{Z}$. The  elements of   $\mathbb{Z}_t$ are represented, by default,  as  $[0, t-1]$.
Nevertheless, sometimes,  $\mathbb{Z}_t$ is explicitly specified to be represented as
$\left[-\lfloor (t-1)/2\rfloor, \lfloor t/2\rfloor \right]$.

If $S$ is a finite set then  $|S|$
is its cardinality,  and $x\leftarrow S$ is the operation of
 picking an element uniformly at random from $\mathcal{S}$.
For two sets $A,B\subseteq \mathbb{Z}_q$, define $A+B\triangleq \{a+b|a\in A, b\in B\}$.
For an addictive group $(G, +)$,  an element $x \in G$ and a subset $S \subseteq G$,
denote by $x + S$ the set containing $x + s$ for all $s \in S$. For a set $S$, denote by $\mathcal{U}(S)$ the uniform distribution over $S$.
For any discrete random variable $X$ over $\mathbb{R}$, denote $\mathsf{Supp}(X) = \{x \in \mathbb{R} \mid \Pr[X = x] > 0\}$.

We use standard notations and conventions below for writing
probabilistic algorithms, experiments and interactive protocols.
If $\mathcal{D}$ denotes a probability distribution, $x\leftarrow \mathcal{D}$ is the operation of
 picking an element according to $\mathcal{D}$.
  If $\alpha$ is neither an algorithm nor a set then $x\leftarrow \alpha$ is a
 simple assignment statement.
If \emph{A} is a probabilistic algorithm, then $A(x_1, x_2, \cdots; r)$ is the result of running \emph{A} on inputs
 $x_1, x_2, \cdots$ and coins $r$. We let $y\leftarrow A(x_1, x_2, \cdots)$ denote the experiment of picking $r$ at
 random and letting $y$ be $A(x_1, x_2, \cdots; r)$.
  By $\Pr[R_1; \cdots; R_n: E]$ we denote
 the probability of event $E$, after the ordered execution of
 random processes $R_1, \cdots, R_n$.

{We say that a function $f(\lambda)$ is
\emph{negligible}, if for every $c>0$ there exists an $\lambda_c$
such that $f(\lambda)<1/\lambda^c$ for all $\lambda>\lambda_c$. Two distribution ensembles $\{X(\lambda,z)\}_{\lambda\in N,z\in \{0,1\}^*}$ and  $\{Y(\lambda,z)\}_{\lambda\in N,z\in \{0,1\}^*}$ are computationally indistinguishable, if for any probabilistic polynomial-time (PPT) algorithm $D$, and  for sufficiently large $\lambda$ and any $z\in \{0,1\}^*$, it holds $|\Pr[D(\lambda,z,X)=1]-\Pr[D(\lambda,z,Y)=1]|$ is negligible in $\lambda$.
}

\subsection{The LWE, LWR, and RLWE problems}
Given positive \emph{continuous} $\alpha > 0$, define the real Gaussian function $\rho_{\alpha}(x) \triangleq \exp(-x^2 / 2\alpha^2) / \sqrt{2\pi \alpha^2}$ for $x \in \mathbb{R}$.
Let $D_{\mathbb{Z}, \alpha}$ denote the one-dimensional \emph{discrete} Gaussian distribution over $\mathbb{Z}$,
which is determined by its probability density function $D_{\mathbb{Z}, \alpha} (x) \triangleq \rho_{\alpha} (x) / \rho_{\alpha}(\mathbb{Z}), x\in \mathbb{Z}$.
Finally, let $D_{\mathbb{Z}^n, \alpha}$ denote the $n$-dimensional \emph{spherical} discrete Gaussian distribution over $\mathbb{Z}^n$,
where each coordinate is drawn \emph{independently} from $D_{\mathbb{Z}, \alpha}$.

Given positive integers $n$ and $q$ that are both polynomial in the security parameter $\lambda$,
an integer vector $\mathbf{s} \in \mathbb{Z}_q^n$,
and a probability distribution $\chi$ on $\mathbb{Z}_q$,
let $A_{q, \mathbf{s}, \chi}$ be the distribution over $\mathbb{Z}_q^n \times \mathbb{Z}_q$ obtained by choosing $\mathbf{a} \in \mathbb{Z}_q^n$ uniformly at random, and an error term $e\gets \chi$, and outputting the pair $(\mathbf{a}, b = \mathbf{a}^T \mathbf{s} + e) \in \mathbb{Z}_q^n \times \mathbb{Z}_q$.
The error distribution $\chi$ is typically taken to be the discrete Gaussian probability distribution $D_{\mathbb{Z}, \alpha}$ defined previously;
However, as suggested in \cite{bcd16} and as we shall see in Section \ref{subsection:failure-probability}, other alternative distributions of $\chi$ can be taken.
Briefly speaking, the  (decisional) \emph{learning with errors} (LWE) assumption \cite{regev2009lattices} says that, for sufficiently large security parameter $\lambda$,  no probabilistic polynomial-time (PPT) algorithm  can distinguish, with non-negligible probability, $A_{q, \mathbf{s}, \chi}$ from the uniform distribution over $\mathbb{Z}_q^n \times \mathbb{Z}_q$.
This holds even if $\mathcal{A}$ sees polynomially many  samples, and even if the secret vector $\mathbf{s}$ is drawn randomly from  $\chi^n$ \cite{applebaum2009fast}.

The LWR problem \cite{BPR12} is a ``derandomized'' variant of the LWE problem. Let  $\mathcal{D}$ be some distribution over $\mathbb{Z}_q^n$, and  $\mathbf{s} \gets \mathcal{D}$. For  integers $q \ge p \ge 2$ and any $x\in \mathbb{Z}_q$, denote  $\lfloor x \rceil_p=\lfloor \frac{p}{q}x \rceil$. Then, for positive integers $n$ and $q \ge p \ge 2$, the LWR distribution $A_{n, q, p}(\mathbf{s})$
 over $\mathbb{Z}_q^n \times \mathbb{Z}_p$ is obtained by  sampling
$\mathbf{a}$ from $\mathbb{Z}_q^n$ uniformly at random, and outputting
$\left(\mathbf{a}, \left\lfloor 
\mathbf{a}^T \mathbf{s} \right\rceil_p \right) \in \mathbb{Z}_q^n \times \mathbb{Z}_p$.
The search LWR problem is to recover the hidden secret $\mathbf{s}$ given
polynomially many  samples of $A_{n, q, p}(\mathbf{s})$.
Briefly speaking, the  (decisional) LWR assumption  says that, for sufficiently large security parameter,  no PPT algorithm $\mathcal{A}$ can distinguish, with non-negligible probability, the
distribution $A_{n, q, p}(\mathbf{s})$
from the distribution $(\mathbf{a} \gets \mathbb{Z}_q^n, \lfloor u \rceil_p)$  where $u \gets \mathbb{Z}_q$.
This holds even if $\mathcal{A}$ sees  polynomially many   samples.


 An efficient reduction from the LWE problem
to the LWR problem, for super-polynomial large  $q$, is provided in \cite{BPR12}.
{Let $B$ denote the bound for any  component in  the secret $\mathbf{s}$.} 
{It is recently shown   that, when  $q \ge 2mBp$ (equivalently, $m\le q/2Bp$),  the LWE problem can be reduced to the (decisional) LWR assumption with $m$ independently random samples \cite{BGMRR16}. Moreover, the reduction from LWE to LWR is actually  independent of the  distribution of the secret $\mathbf{s}$.}



For the positive integer $m$ that is polynomial in the security parameter $\lambda$,
let $n\triangleq \varphi(m)$ denote the totient of $m$, and $\mathcal{K}\triangleq \mathbb{Q}(\zeta_m)$ be the number field obtained by adjoining an abstract element $\zeta_m$ satisfying $\Phi_m(\zeta_m) = 0$, where $\Phi_m(x)\in \mathbb{Z}[x]$ is the $m$-th cyclotomic polynomial of degree $n$.
Moreover, let $\mathcal{R}\triangleq \mathcal{O}_{\mathcal{K}}$ be the ring of integers in $\mathcal{K}$.
Finally, given a positive prime $q=\mathrm{poly}(\lambda)$ such that $q\equiv 1\pmod{m}$,
define the quotient ring $\mathcal{R}_q\triangleq \mathcal{R}/q\mathcal{R}$.

We briefly review the RLWE problem, and its hardness result \cite{LYUBASHEVSKY2013ON,lyubashevsky2013toolkit,DD12}.
As we shall see, it suffices in this work to consider a \emph{special} case of the original ring-LWE problem defined in \cite{LYUBASHEVSKY2013ON}.
Let $n \geq 16$ be a power-of-two and $q =\mathrm{poly}(\lambda)$ be a positive prime such that $q\equiv 1\pmod{2n}$.
Given  $\mathbf{s} \gets \mathcal{R}_q$, a sample drawn from the RLWE distribution  $A_{n, q, \alpha, \mathbf{s}}$ over $\mathcal{R}_q \times \mathcal{R}_q$ is generated by first choosing $\mathbf{a} \gets \mathcal{R}_q, \mathbf{e} \gets D_{\mathbb{Z}^n, \alpha}$, and then outputting
$(\mathbf{a}, \mathbf{a} \cdot \mathbf{s} + \mathbf{e}) \in \mathcal{R}_q \times \mathcal{R}_q$.  Roughly speaking, the (decisional) RLWE assumption  says that,
for sufficiently large security parameter $\lambda$,  no PPT algorithm $\mathcal{A}$ can distinguish, with non-negligible probability, $A_{n, q, \alpha, \mathbf{s}}$ from the uniform distribution over $\mathcal{R}_q\times \mathcal{R}_q$. This holds even if $\mathcal{A}$ sees polynomially many samples, and even if the secret $\mathbf{s}$ is drawn randomly from the same distribution of the error polynomial $\mathbf{e}$ \cite{DD12,applebaum2009fast}. Moreover, as suggested in \cite{newhope15}, alternative distributions for the error polynomials can be taken for the sake of efficiency while without essentially reducing security.

\section{Key Consensus with Noise}

Before presenting the definition of key consensus (KC) scheme, we first introduce a new function  $|\cdot|_t$ relative to  arbitrary positive integer $t\geq 1$: 
$$|x|_{t} = \min\{x \bmod t, t - x \bmod t\}, \quad \forall x\in \mathbb{Z},$$
where the result of modular operation is represented in $\{0,...,(t-1)\}$. For instance,  $|-1|_t = \min\{-1 \mod t, (t+1) \mod t\}=\min\{t-1, 1\}=1$.
In the following description, we use $|\sigma_1-\sigma_2|_q$
{to measure  the distance between two elements $\sigma_1,\sigma_2\in \mathbb{Z}_q$}. {In this work, such a distance is caused by small noises, and is relatively small compared to $q$.}

\begin{figure}[H]
\centering
\includegraphics[scale=0.9]{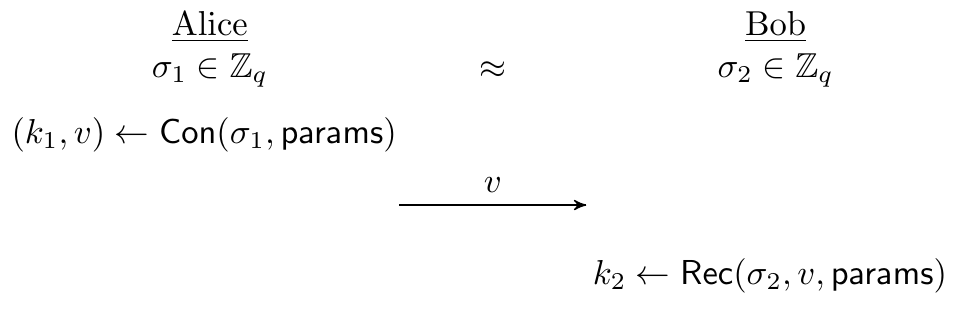}
\caption{Brief depiction of KC, where
$k_1,k_2\in \mathbb{Z}_m$, $v\in \mathbb{Z}_g$ and $|\sigma_1-\sigma_2|_q \le d$.}
\label{figure:kc}
\end{figure}

\begin{definition}
    A key consensus  scheme $KC=(\textsf{params}, \textsf{Con}, \textsf{Rec})$, briefly depicted in  \figurename~\ref{figure:kc},
    is  specified as follows.
\begin{itemize}

\item $\textsf{params}=(q,m,g,d,aux)$ denotes the system parameters,\footnote{{In practice, the system parameters are usually set by the higher-level protocol that calls KC.}}
    where   $q,m,g,d$ are positive integers satisfying $2\leq m,g\leq q, 0\leq d\leq \lfloor \frac{q}{2}\rfloor$ (which  dominate  security, correctness and bandwidth of the KC scheme), and
        $aux$ denotes some   auxiliary    values  that are usually determined by $(q,m,g,d)$ and could  be set to be a special symbol $\emptyset$ indicating ``empty".


    \item \textsf{$(k_1, v) \gets \textsf{Con}(\sigma_1, \textsf{\textsf{params}})$:} On input of $(\sigma_1\in \mathbb{Z}_q,\textsf{params})$, the \emph{probabilistic} polynomial-time  conciliation algorithm $\textsf{Con}$ outputs $(k_1,v)$, where $k_1 \in
        \mathbb{Z}_m$ is the shared-key, and $v \in \mathbb{Z}_g$ is a hint signal that will
        be publicly delivered to the communicating peer  to help the two
        parties reach consensus.
    \item $k_2 \gets \textsf{Rec}(\sigma_2, v, \textsf{params})$: On input of  $(\sigma_2\in\mathbb{Z}_q,v,\textsf{params})$, the \emph{deterministic} polynomial-time reconciliation  algorithm
        $\textsf{Rec}$ outputs $k_2 \in \mathbb{Z}_m$.
\end{itemize}
\begin{description}
\item [Correctness:] A KC scheme  is \emph{correct}, if   it holds  $k_1=k_2$ for
any $\sigma_1,\sigma_2\in \mathbb{Z}_q$ such that  $|\sigma_1-\sigma_2|_q \le d$.

\item [Security:]
A KC scheme is \emph{secure}, if
$k_1$ and $v$ are independent, and $k_1$ is uniformly distributed over $\mathbb{Z}_m$, whenever $\sigma_1 \gets \mathbb{Z}_q$ \emph{(}i.e.,  $\sigma_1$ is taken uniformly at random from $\mathbb{Z}_q$\emph{)}. The probability is taken over the sampling of $\sigma_1$ and  the random coins used by \textsf{Con}. 

\end{description}

\end{definition}

\subsection{Efficiency Upper Bound of KC}

For fixed $q, g, d$, we expect the two communicating parties to reach as more consensus bits as possible,
so the range of consensus key $m$ can be regarded as an indicator of efficiency.
The following theorem reveals an upper bound on the range of consensus key of a KC
with parameters $q$, $g$ (parameterize bandwidth), and $d$ (parameterize correctness).
Its proof also divulges some intrinsic properties of any \emph{correct} and \emph{secure} KC scheme. 

\begin{theorem}\label{th-optimal}
    If $KC=(\mathsf{params}, \mathsf{Con}, \mathsf{Rec})$ is a
    \emph{correct} and \emph{secure} key consensus scheme,
    and $\mathsf{params} = (q, m, g, d, aux)$,
    then
    \begin{equation*}
        2md \le q\left(1 - \frac{1}{g}\right).
    \end{equation*}
\end{theorem}

\textsf{Remark:} Some comments are in order.
Theorem~\ref{th-optimal} divulges an efficiency upper bound on the system parameters of KC schemes,
and allows us to take balance on these parameters according to different priorities among security, computational efficiency and bandwidth consumption.
When balancing these parameters, we are mainly concerned with the parameters $(q,d,m)$, with a focus on the parameter $q$ that dominates the security and  efficiency of the underlying KC scheme. The parameter $g$  is mainly related to bandwidth. But the bandwidth reduction  with a smaller $g$ can be  overtaken by the overall  efficiency gains with a smaller $q$.


Before proceeding to prove Theorem \ref{th-optimal}, we first prove the following propositions.

\begin{proposition} \label{lemma:determine-con}
    Given $\textsf{params}=(q,m,g,d,aux)$ for a \emph{correct} and \emph{secure} KC scheme.
    For any arbitrary fixed $\sigma_1 \in \mathbb{Z}_q$,
    if $\mathsf{Con}(\sigma_1, \textsf{params})$ outputs $(k_1, v)$ with positive prabability,
    then the value $k_1$ is fixed w.r.t. the $(v,\sigma_1)$. That is, for any random coins $(r,r^\prime)$, if $\textsf{Con}(\sigma_1, \textsf{params}, r)=(k_1,v)$ and $\textsf{Con}(\sigma_1, \textsf{params}, r^\prime)=(k^\prime_1,v)$, then $k_1=k^\prime_1$.
\end{proposition}

\begin{proof}
    Let $\sigma_2=\sigma_1$, then $|\sigma_1-\sigma_2|_q=0\leq d$.
    Then, according to the \emph{correctness} of KC, we have that $k_1=k_2=\textsf{Rec}(\sigma_2,v)=\textsf{Rec}(\sigma_1,v)$. However, as $\textsf{Rec}$ is a deterministic algorithm,  $k_2$ is fixed w.r.t. $(\sigma_1,v)$. As a consequence, $k_1$ is also fixed w.r.t. $(\sigma_1,v)$, {no matter what randomness is used by \textsf{Con}}.
  \end{proof}

\begin{proposition} \label{lemma:selectv0}
    Given $\textsf{params}=(q,m,g,d,aux)$ for a KC scheme, for any  $v \in \mathbb{Z}_g$, let  $S_{v}$ be the set containing all $\sigma_1$
    such that $\textsf{Con}(\sigma_1, \textsf{params})$ outputs $v$ with positive probability. Specifically,
     $$S_{v} = \left\{\sigma_1 \in \mathbb{Z}_q \mid
            \Pr\left[(k_1, v^\prime) \gets \textsf{Con}(\sigma_1, \textsf{params}): v^\prime= v\right] > 0\right\}.$$
      Then, there exists $v_0 \in \mathbb{Z}_g$ such that  $|S_{v_0}| \ge q/g$.
\end{proposition}

\begin{proof}
    For each $\sigma_1 \in \mathbb{Z}_q$,
    we run $\textsf{Con}(\sigma_1,\textsf{params})$ and get a pair $(k_1, v) \in \mathbb{Z}_{m} \times \mathbb{Z}_{g}$
    satisfying $\sigma_1 \in S_{v}$. Then,  the proposition is
    clear by the pigeonhole principle.
  \end{proof}

\begin{proof}[Proof of Theorem ~\ref{th-optimal}]
    From Proposition~\ref{lemma:selectv0},
    there exists a $v_0 \in \mathbb{Z}_g$ such that  $|S_{v_0}| \ge q / g$.
    Note that, for any $\sigma_1\in S_{v_0}$,
    $\mathsf{Con}(\sigma_1, \mathsf{params})$ outputs $v_0$ with positive probability.

    For each $i \in \mathbb{Z}_m$,
    let $K_i$ denote the set containing all $\sigma_1$
    such that $\textsf{Con}(\sigma_1, \textsf{params})$ outputs $(k_1 = i,v=v_0$) with positive probability.
    From Proposition~\ref{lemma:determine-con},
    $K_i$'s form a disjoint partition of $S_{v_0}$.
    From the independence between  $k_1$ and $v$, and the uniform distribution of $k_1$, (as we assume the underlying KC is \emph{secure}), we know $\Pr[k_1 = i \mid v = v_0] = \Pr[k_1 = i] > 0$, and so $K_i$ is  non-empty for each $i \in \mathbb{Z}_m$.
    Now, for each $i \in \mathbb{Z}_m$, denote by $K^\prime_i$ the set containing all $\sigma_2 \in \mathbb{Z}_q$
    such that $\textsf{Rec}(\sigma_2, v_0, \textsf{params})=i$.
    As $\textsf{Rec}$ is  deterministic, $K^\prime_i$'s are well-defined and are disjoint.

    From the \emph{correctness} of KC,
    for every $\sigma_1 \in K_i, |\sigma_2 - \sigma_1|_q \le d$, we have $\sigma_2 \in K_i'$. That is, $K_i + [-d, d] \subseteq K'_i$.

    We shall prove that $K_i + [-d, d]$ contains at least $|K_i| + 2d$ elements.
    If $K_i + [-d, d] = \mathbb{Z}_m$, then $m = 1$, which is a contradiction
    (we exclude the case of $m = 1$ in the definition of KC as it is a trivial case).
    If there exists an $x \in \mathbb{Z}_m$ such that $x \notin K_i + [-d, d]$,
    we can see $\mathbb{Z}_m$ as a segment starting from the point $x$ by
    arranging its elements as $x, (x + 1) \bmod m, (x + 2) \bmod m, \dots, (x + m - 1) \bmod m$.
    Let $l$ be the left most element in $K_i + [-d, d]$ on the segment,
    and $r$ be the right most such element.
    Then $K_i + [-d, d]$ contains at least $|K_i|$ elements between $l$ and $r$ inclusively on the segment.
    Since $l + [-d, 0]$ and $r + [0, d]$ are subset of $K_i + [-d, d]$, and are not overlap
    (because $x \notin K_i + [-d, d]$),
    the set $K_i + [-d, d]$ contains at least $|K_i| + 2d$ elements.

    Now we have $|K_i| + 2d \le |K_i'|$. When we add up on both sides for all $i \in \mathbb{Z}_m$,
    then we derive $|S_{v_0}| + 2md \le q$. By noticing that $|S_{v_0}| \ge q / g$, the theorem is established. 
  \end{proof}

\subsection{Construction and Analysis of OKCN}
The key consensus  scheme, named ``optimally-balanced  key consensus with noise (OKCN)", is presented  in Algorithm~\ref{kcs:1}, followed with  some explanations for implementation details.


Define $\sigma^\prime_A=\alpha \sigma_1+e$. Note that it always holds
$\sigma^\prime_A<q^\prime$. However, in some rare cases, $\sigma^\prime_A$ could
be a negative value; for example, for the case that $\sigma_1=0$ and $e\in
\left[-\lfloor(\alpha-1)/2\rfloor,-1\right]$. Setting $\sigma_A=\sigma^\prime_A
\mod q^\prime$, in line~\ref{op-sigmaA},  is to ensure that $\sigma_A$ is always a non-negative  value in $\mathbb{Z}_{q^\prime}$, which can be simply implemented as follows: if $\sigma^\prime_A<0$ then set $\sigma_A=\sigma^\prime_A+q^\prime$, otherwise set $\sigma_A=\sigma^\prime_A$. {Considering potential timing attacks, conditional
statement judging whether $\sigma^\prime_A$ is negative or not
can be avoided  by a bitwise operation extracting the sign bit of $\sigma^\prime_A$.
In specific, suppose $\sigma^\prime_A$ is a 16-bit signed or unsigned integer,
then one can code $\sigma_A = \sigma^\prime_A + ((\sigma^\prime_A >> 15) \&
1)*q'$ in C language. The same techniques can also  be applied to the calculation in line~\ref{op-k2}.} 

In lines \ref{op-div} and \ref{op-divmod},  $(k_1,v')$ can actually  be calculated simultaneously by a
single command $div$ in assembly language.
In line~\ref{op-k2}, the floating point arithmetic can be replaced by integer
arithmetic. If $m$ is small enough, such as $2$ or $3$, the slow
complex integer division operation can be replaced by relative faster conditional statements.

\noindent
\begin{algorithm}[H]

    \caption{OKCN: Optimally-balanced  KC with Noise}\label{kcs:1}
\begin{algorithmic}[1]
    \State{$\textsf{params}=(q,m,g,d,aux)$,  $aux =\{q'=\mathsf{lcm}(q,m),\alpha=q'/q,\beta=q'/m\}$}
    \Procedure {$\textsf{Con}$}{$(\sigma_1, \textsf{params})$} \Comment{{$\sigma_1\in [0, q-1]$}}
    \State{$e \gets \left[-\lfloor(\alpha-1)/2\rfloor, \lfloor\alpha/2\rfloor\right]$} \label{op-e}
    \State{$\sigma_A=(\alpha \sigma_1+e) \bmod q'$} \label{op-sigmaA}
    \State{$k_1=\lfloor\sigma_A/\beta\rfloor\in \mathbb{Z}_m$}\label{op-div}
    \State{$v'=\sigma_A \mod \beta$}
    \label{op-divmod}
    \State{$v = \lfloor v' g / \beta \rfloor$}\Comment{$v\in \mathbb{Z}_g$}
    \Return{$(k_1, v)$}
\EndProcedure

\Procedure{\textsf{Rec}}{$\sigma_2, v, \textsf{params}$}\Comment{$\sigma_2\in [0,q-1]$}
\State{$k_2 = \lfloor \alpha \sigma_2 / \beta - (v + 1/2)/g \rceil \bmod m$} \label{op-k2}
    \Return{$k_2$}
\EndProcedure
\end{algorithmic}
\end{algorithm}

The value $v+1/2$, in line~\ref{op-k2},  estimates  the exact value of $v'g/\beta$. Such an estimation
can be more accurate, if one chooses to use the average value of all $v'g/\beta$'s such that
$\lfloor v'g / \beta \rfloor = v$. Though such accuracy can improve the bound on
correctness slightly, the formula calculating $k_2$ becomes more complicated.

\subsubsection{Correctness and Security of OKCN}
Recall that, for  arbitrary positive integer $t\geq 1$ and any $x\in \mathbb{Z}$,  $|x|_{t} = \min\{x \bmod t, t - x \bmod t\}$. Then, the following fact is direct from the definition of $|\cdot|_t$.

\begin{fact} \label{lemma:cir}
   For any  $x,y,t,l\in \mathbb{Z}$ where $t\geq 1$ and $l\geq 0$,
   if $|x - y|_q \le l$,
   then there exists $\theta \in
    \mathbb{Z}$ and $\delta \in [-l, l]$ such that  $x=y+\theta t+\delta$.
\end{fact}

\begin{theorem} \label{th:kcs1-correct}
    Suppose that the system parameters satisfy
    $(2d + 1)m < q\left(1 - \frac{1}{g}\right)$ where $m\ge 2$ and $g \ge 2$. Then,   the \emph{OKCN} scheme is \emph{correct}.
\end{theorem}

\begin{proof}
    Suppose  $|\sigma_1 - \sigma_2|_q \le d$. By Fact~\ref{lemma:cir},
    there exist $\theta \in \mathbb{Z}$ and $\delta \in [-d,d]$  such that  $\sigma_2=\sigma_1 + \theta q + \delta$.
    From line~\ref{op-sigmaA} and~\ref{op-divmod} in Algorithm~\ref{kcs:1}, we know that
    there is a $\theta^\prime\in \mathbb{Z}$, such that
    $\alpha \sigma_1 + e + \theta^\prime q^\prime= \sigma_A = k_1 \beta + v'$.
    And from the definition of $\alpha, \beta$, we have $\alpha/\beta = m / q$.
    Taking these into the formula of $k_2$ in $\textsf{Rec}$ (line~\ref{op-k2} in Algorithm \ref{kcs:1}), we have
    \begin{align}
        k_2 &= \lfloor \alpha \sigma_2 / \beta - (v + 1/2)/g \rceil \bmod m \\
        &= \lfloor \alpha (\theta q + \sigma_1 + \delta) / \beta - (v + 1/2)/g \rceil \bmod m \\
        & = \left\lfloor m (\theta-\theta^\prime) + \frac{1}{\beta} (k_1\beta + v' - e) +
            \frac{\alpha\delta}{\beta}- \frac{1}{g}(v + 1/2)\right\rceil \bmod m \\
        &= \left\lfloor k_1 + \left(\frac{v'}{\beta} - \frac{v +
                1/2}{g}\right) - \frac{e}{\beta} + \frac{\alpha \delta}{\beta}
    \right\rceil \bmod m \label{formula:kcs1-correct}
    \end{align}

    Notice that $|v'/\beta - (v + 1/2)/g| = |v'g - \beta(v + 1/2)|/\beta g \le
    1/2g$.
    So
    \begin{equation*}
        \left|\left(\frac{v'}{\beta} - \frac{v + 1/2}{g}\right)
            - \frac{e}{\beta} + \frac{\alpha \delta}{\beta} \right|
        \le \frac{1}{2g} + \frac{\alpha}{\beta}(d + 1/2).
    \end{equation*}
    From the assumed condition $(2d+1)m<q(1-\frac{1}{g})$, we get that the right-hand side is
    strictly smaller than $1/2$; Consequently,  after the rounding, $k_2=k_1$.
  \end{proof}



\begin{theorem}\label{th:kcs1-secure} OKCN is secure. Specifically, when $\sigma_1\gets \mathbb{Z}_q$, $k_1$ and $v$ are independent, and $k_1$ is uniform over $\mathbb{Z}_m$, where  the probability is taken over the sampling of $\sigma_1$ and the random
    coins used by $\textsf{Con}$.
\end{theorem}

\begin{proof} Recall that $q'=\mathsf{lcm}(q,m),\alpha=q'/q,\beta=q'/m$.
    We first demonstrate that $\sigma_A$ is subject to uniform distribution over
    $\mathbb{Z}_{q'}$.
    Consider the  map $f: \mathbb{Z}_{q} \times \mathbb{Z}_{\alpha} \rightarrow
    \mathbb{Z}_{q'}$; $f(\sigma, e) = (\alpha\sigma + e) \bmod q'$, where the elements in
    $\mathbb{Z}_q$ and $\mathbb{Z}_{\alpha}$ are represented in the same way as specified
    in Algorithm~\ref{kcs:1}. It is easy to check that $f$ is an one-to-one map.
    Since $\sigma_1\gets \mathbb{Z}_q$ and $e\gets \mathbb{Z}_\alpha$ are subject to uniform distributions, and they are
    independent, $\sigma_A = (\alpha\sigma_1 + e) \bmod q' = f(\sigma_1,e)$ is also subject to uniform
    distribution over $\mathbb{Z}_{q'}$.

    In the similar way, defining $f':\mathbb{Z}_{m} \times \mathbb{Z}_{\beta} \rightarrow
    \mathbb{Z}_{q'}$ such that $f'(k_1, v') = \beta k_1 + v'$, then $f'$ is obviously a one-to-one map.
    From line~\ref{op-divmod} of Algorithm~\ref{kcs:1}, $f'(k_1, v') = \sigma_A$.
    As  $\sigma_A$ is distributed uniformly over $\mathbb{Z}_{q'}$,
    $(k_1, v')$ is uniformly distributed over $\mathbb{Z}_m \times \mathbb{Z}_{\beta}$,
    and so $k_1$ and $v'$ are independent.
    As $v$ only depends on $v'$, $k_1$ and $v$ are independent.
  \end{proof}

\subsubsection{Special Parameters, and Performance Speeding-Up}\label{OKCN-special}

The first and the second line of \textsf{Con} (line~\ref{op-e} and~\ref{op-sigmaA} in Algorithm \ref{kcs:1}) play the  role in transforming a
uniform distribution over $\mathbb{Z}_q$ to a uniform distribution over $\mathbb{Z}_{q'}$.
If one chooses $q, g, m$ to be power of $2$, i.e.,
$q = 2^{\bar{q}},g = 2^{\bar{g}},m = 2^{\bar{m}}$ where $\bar{q},\bar{g},\bar{m}\in \mathbb{Z}$, then such transformation is
not necessary. In this case $\textsf{Con}$ and $\textsf{Rec}$ can be simplified as follows:

\noindent
\begin{algorithm}[H]
    \caption{OKCN power 2}\label{kcs:1-power2-g}
\begin{algorithmic}[1]
\State{$\textsf{params}:$ $q=2^{\bar{q}},g=2^{\bar{g}},m=2^{\bar{m}},d,$ $aux =\{(\beta=q/m=2^{\bar{q}-\bar{m}}, \gamma=\beta/g=2^{\bar{q}-\bar{m}-\bar{g}})\}$}
\Procedure{\textsf{Con}}{$\sigma_1, \textsf{params}$}
    \State{$k_1 = \left\lfloor \sigma_1/ \beta \right\rfloor$}
    \State{$v = \lfloor (\sigma_1 \bmod \beta) / \gamma \rfloor$}
    \Return{$(k_1, v)$}
\EndProcedure

\Procedure{\textsf{Rec}}{$\sigma_2, v, \textsf{params}$}
\State{$k_2 = \left\lfloor \sigma_2/\beta - (v + 1/2) / g \right\rceil \bmod m$}
    \Return{$k_2$}
\EndProcedure{}
\end{algorithmic}
\end{algorithm}

Since the random noise $e$ used in calculating $\sigma_A$ in Algorithm \ref{kcs:1} is avoided,
the correctness constraint on parameters
can be relaxed, and we have the following corollary.
\begin{corollary}\label{claim:1-power2-g-correct}
   If $q$ and $m$ are power of $2$, and $d, g, m$ satisfy $2md < q\left(1 -
           \frac{1}{g}\right)$,
           then the KC scheme
       described in Algorithm \ref{kcs:1-power2-g} is both \emph{correct} and \emph{secure}.
\end{corollary}

\begin{proof}
For correctness, as no additional noise $e$ is added, one can take $e = 0$ into
Formula~\ref{formula:kcs1-correct}, and then the correctness of Algorithm~\ref{kcs:1-power2-g} directly follows from the  proof of
Theorem~\ref{th:kcs1-correct}.
For security, as a variation of the generic structure of  Algorithm~\ref{kcs:1}, the security of Algorithm~\ref{kcs:1-power2-g} inherits from that of  Algorithm~\ref{kcs:1}.
  \end{proof}

If we take $\bar{g} + \bar{m} = \bar{q}$,
Algorithm~\ref{kcs:1-power2-g} can be further simplified into  the variant depicted in Algorithm~\ref{kcs:1-power2},  with the constraint on parameters is further relaxed.

\begin{corollary} \label{claim:1-power2-correct}
If $ m, g$ are power of $2$,  $q=m \cdot g$, and  $2md < q$,
then the KC scheme described in Algorithm \ref{kcs:1-power2} is \emph{correct} and \emph{secure}.  Notice that the constraint on parameters is further simplified to $2md < q$  in this case.

\end{corollary}
\begin{proof}
For correctness, supposing $|\sigma_1 - \sigma_2|_q \le d$, by Fact~\ref{lemma:cir},
there exist $\theta \in \mathbb{Z}$ and $\delta \in [-d, d]$
such that $\sigma_2 = \sigma_1 + \theta q + \delta$.
Taking  this into line~\ref{op-optimal-k2} of Algorithm~\ref{kcs:1-power2}, i.e.,
the formula computing $k_2$,  we have
\begin{align*}
k_2 &= \lfloor (\sigma_1 - v + \theta q + \delta) / g \rceil \bmod m \\
&= (k_1 + \theta m + \lfloor \delta/g \rceil) \bmod m.
\end{align*}
If $2md < q$, then $|\delta / g| \le d/g < 1/2$, so that $k_2 = k_1 \bmod m = k_1$.

For security, as a special case of generic scheme described in Algorithm~\ref{kcs:1},
the security of Algorithm~\ref{kcs:1-power2} follows  directly from that of Algorithm \ref{kcs:1}.    \end{proof}

\begin{algorithm}[H]
    \caption{OKCN simple}\label{kcs:1-power2}
\begin{algorithmic}[1]
\State{$\textsf{params}:$ $q=2^{\bar{q}},g=2^{\bar{g}},m=2^{\bar{m}},d,$ where $\bar{g} + \bar{m} = \bar{q}$}
    \Procedure{\textsf{Con}}{$\sigma_1, \textsf{params}$}
    \State{$k_1 = \left\lfloor \frac{\sigma_1}{g}\right\rfloor$}
    \State{$v = \sigma_1 \bmod g $}
    \Return{$(k_1, v)$}
\EndProcedure

\Procedure{\textsf{Rec}}{$\sigma_2, v, \textsf{params}$}
\State{$k_2 = \left\lfloor \frac{\sigma_2 - v}{g}\right\rceil \bmod m$} \label{op-optimal-k2}
    \Return{$k_2$}
\EndProcedure
\end{algorithmic}
\end{algorithm}

\section{Asymmetric Key Consensus with Noise}

\begin{figure}[H]
\centering
\includegraphics[scale=.9]{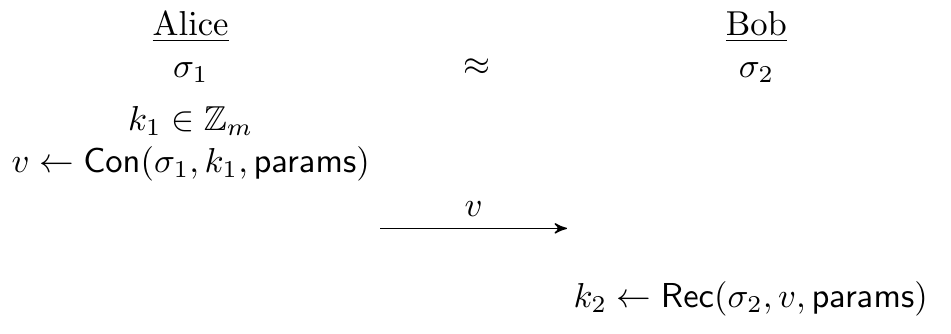}
\caption{Brief depiction of AKC}
\label{figure:akc}
\end{figure}

  When OKCN is used as the building tool in constructing  key exchange (KE) protocols,  the party    who sends the hint signal   $v$ actually plays the role of the responder. When cast into secure transport protocols in the client/server setting, e.g., the next generation of TLS 1.3 \cite{TLS},    the client (corresponding to  Bob) plays the role of the  initiator, and the server (corresponding to Alice) plays the role of the responder. {For OKCN-based key exchange,  both the initiator and the  responder play a \emph{symmetric} role in outputting the shared-key, in the sense that no one can pre-determine the session-key before the KE protocol run. Specifically, the responder can calculate  the session-key only after receiving the ephemeral message from the initiator (unless the message from the initiator is public or static).}
However, in  some application scenarios  particularly in the client/server setting like those based upon TLS1.3,
 it is desirable to render the server  \emph{asymmetric} power in predetermining  session-keys before the protocol run, e.g., in order to balance workloads and security or to  resist the more and more vicious  DDoS attacks.
 {For example, this can much reduce the online burden faced by the server, by offline or parallel computing the following ciphertexts or MACs  with the pre-determined session-keys.}

  Another motivation is that OKCN-based key exchange, with negligible failure probability,  cannot be directly transformed into a  CPA-secure public-key encryption (PKE) scheme without additionally employing a CPA-secure symmetric-key encryption  (SKE) scheme. 
These motivate us to introduce \emph{asymmetric key consensus with noise} (AKCN), as depicted in Figure \ref{figure:akc} and specified below. 


\begin{definition}
    An asymmetric key consensus scheme $AKC=(\textsf{params},\textsf{Con},
    \textsf{Rec})$ 
    is specified as follows:

\begin{itemize}
    \item $\textsf{params}=(q, m, g, d, aux)$ denotes the system
        parameters, where q, $2\leq m,g\leq q, 1\leq d\leq \lfloor \frac{q}{2}\rfloor$ are positive integers, and
        $aux$ denotes some   auxiliary    values  that are usually determined by $(q,m,g,d)$ and could  be set to be empty.

    \item $v \gets \textsf{Con}(\sigma_1, k_1, \textsf{params})$: On input of
        $(\sigma_1 \in \mathbb{Z}_q, k_1 \in \mathbb{Z}_m, \textsf{params})$,
        the \emph{probabilistic} polynomial-time conciliation algorithm
        $\textsf{Con}$ outputs the public hint signal  $v\in \mathbb{Z}_g$.

    \item $k_2 \gets \textsf{Rec}(\sigma_2, v, \textsf{params})$: On input of $(\sigma_2, v, \textsf{params})$,
        the \emph{deterministic} polynomial-time algorithm $\textsf{Rec}$ outputs $k_2 \in \mathbb{Z}_m$.
\end{itemize}
\begin{description}
\item [Correctness:] An  AKC scheme  is \emph{correct}, if   it holds  $k_1=k_2$ for
 any $\sigma_1,\sigma_2\in \mathbb{Z}_q$ such that $|\sigma_1-\sigma_2|_q \leq d$.

\item [Security:]
An  AKC scheme is \emph{secure}, if $v$ is independent of $k_1$
whenever $\sigma_1$ is uniformly distributed over $\mathbb{Z}_q$. Specifically, for arbitrary $\tilde{v}\in \mathbb{Z}_g$ and arbitrary $\tilde{k}_1,\tilde{k}^\prime_1\in \mathbb{Z}_m$, it holds that $\Pr[v=\tilde{v}|k_1=\tilde{k}_1]=\Pr[v=\tilde{v}|k_1=\tilde{k}^\prime_1]$, where the probability is taken over  $\sigma_1\gets \mathbb{Z}_q$ and the random coins used by \textsf{Con}.

\end{description}

\end{definition}

When AKC is used as a building tool for key exchange, $k_1$ is taken uniformly at random from $\mathbb{Z}_q$. However, when AKC is used for public-key encryption, $k_1$ can be arbitrary value from the space of  plaintext messages.

For the efficiency upper bound of AKC, the following theorem divulges bounds on bandwidth (parameterized by $g$),
consensus range (parameterized by $m$), and correctness (parameterized by $d$)
for any AKC scheme. 

\begin{theorem} \label{ookc-bnd}
    Let $AKC = (\textsf{params}, \textsf{Con}, \textsf{Rec})$ be an asymmetric
    key consensus scheme with $\mathsf{params} = (q, m, d, g, aux)$.
    If $AKC$ is \emph{correct} and \emph{secure}, then
    \begin{equation*}
        2md \le q\left(1 - \frac{m}{g}\right).
    \end{equation*}
\end{theorem}


Comparing the formula $2md \le q(1 - m/g)$ in Theorem~\ref{ookc-bnd} with the
formula $2md \le q(1 - 1/g)$ in Theorem~\ref{th-optimal}, we see that the only
difference is a factor $m$ in $g$. This indicates that, on the same values of $(q,m,d)$, an AKC scheme has to use a bigger bandwidth parameter  $g$ compared to KC.
 Before proving  Theorem~\ref{ookc-bnd},
we first adjust Proposition~\ref{lemma:selectv0} to  the AKC setting, {as following}.

\begin{proposition} \label{lemma:pigeonhole-ookc}
    Given $\mathsf{params} = (q, m, g, d, aux)$ for an \emph{correct} and \emph{secure} AKC scheme,
    then there exists $v_0 \in \mathbb{Z}_g$ such that  $|S_{v_0}| \ge mq /g$.
\end{proposition}

\begin{proof}
    If $k_1$ is taken uniformly at random from $\mathbb{Z}_m$, 
    AKC can be considered as a special  KC scheme by treating
    $k_1 \gets \mathbb{Z}_m; v\gets \textsf{Con}(\sigma_1,k_1,\textsf{params})$ as $(k_1, v) \gets \textsf{Con}(\sigma_1,  \textsf{params})$. Consequently,   Proposition~\ref{lemma:determine-con} holds for this case.
    Denote $S^\prime_{v} \overset{\triangle}{=} \left\{(\sigma_1,k_1) \in \mathbb{Z}_q \times \mathbb{Z}_m  \mid \Pr\left[v' \gets \textsf{Con}(\sigma_1, k_1, \textsf{params}): v' = v \right] > 0\right\}$.
    Then,  $S_v$ defined in Proposition~\ref{lemma:selectv0} equals to the set {containing all the values of $\sigma_1$ appeared in  
     $(\sigma_1, \cdot) \in S^\prime_v$.}
    We run $\textsf{Con}(\sigma_1, k_1, \textsf{params})$
    for each pair of $(\sigma_1, k_1)\in   \mathbb{Z}_q\times \mathbb{Z}_m$.
    By  the pigeonhole principle, there must exist a $v_0 \in \mathbb{Z}_g$ such that $|S^\prime_{v_0}| \ge qm / g$.
   For any two pairs $(\sigma_1, k_1)$ and $(\sigma_1', k_1')$ in $S^\prime_{v_0}$,
        if $\sigma_1 = \sigma_1'$, from Proposition~\ref{lemma:determine-con}
        we derive that $k_1 = k_1'$, and then $(\sigma_1, k_1) = (\sigma_1', k_1')$.
        Hence, if $(\sigma_1, k_1)$ and $(\sigma_1', k_1')$ are different, then
        $\sigma_1 \neq \sigma_1'$, and so $|S_{v_0}| = |S_{v_0}'| \ge mq / g$.
  \end{proof}

\begin{proof} [Proof of Theorem \ref{ookc-bnd}]
    By viewing  AKC, with $k_1\gets \mathbb{Z}_q$, as a special KC scheme,
    all the  reasoning in the proof of Theorem~\ref{th-optimal} holds true now.
    At the end of the proof of Theorem~\ref{th-optimal}, we derive $|S_{v_0}| + 2md
    \le q$. By taking $|S_{v_0}|\geq mq/g$ according to  Proposition~\ref{lemma:pigeonhole-ookc}, the proof is finished.
 \end{proof}



\subsection{Construction and Analysis of AKCN}

\begin{algorithm}[H]
\caption{AKCN: Asymmetric KC with Noise}\label{kcs:2}
\begin{algorithmic}[1]
    \State{$\textsf{params} = (q, m, g, d, aux)$, where $aux =\emptyset$.}
\Procedure{Con}{$\sigma_1, k_1, \textsf{params}$} \Comment{$\sigma_1\in [0, q-1]$}
    \State{$v = \left\lfloor g \left(\sigma_1 + \left\lfloor k_1 q/m\right\rceil\right) / q \right\rceil \bmod g$} \label{op-akcn-v}
    \Return{$v$}
\EndProcedure
\Procedure{Rec}{$\sigma_2, v, \textsf{params}$}\Comment{$\sigma_2\in [0, q-1]$}
\State{$k_2 = \left\lfloor m (v / g - \sigma_2 / q)\right\rceil \bmod m$}
    \Return{$k_2$}
\EndProcedure
\end{algorithmic}
\end{algorithm}

The AKC scheme, referred to as asymmetric KC with noise (AKCN),  is depicted in Algorithm~\ref{kcs:2}. {We note that, in some sense, AKCN could be viewed as the generalization and optimization of the consensus mechanism proposed in \cite{LYUBASHEVSKY2013ON} for CPA-secure public-key encryption.}
For AKCN, we can offline compute and store 
$k_1$ and $g \lfloor k_1 q/m \rceil$ in order to accelerate online performance.


%

\subsubsection{Correctness and Security of AKCN}
\begin{theorem}
    Suppose  the  parameters of AKCN satisfy $(2d + 1)m < q\left(1 - \frac{m}{g}\right)$. Then,  the AKCN scheme described in Algorithm \ref{kcs:2}  is \emph{correct}.
\end{theorem}

\begin{proof}
    From the formula generating $v$, we know that there exist
    $\epsilon_1, \epsilon_2 \in \mathbb{R}$ and $\theta \in \mathbb{Z}$, where $|\epsilon_1| \le 1/2$
    and $|\epsilon_2| \le 1/2$, such that
    \begin{align*}
        v = \frac{g}{q}\left(\sigma_1 + \left(\frac{k_1q}{m} + \epsilon_1
            \right)\right) + \epsilon_2 + \theta g
     \end{align*}

     Taking this into the formula computing $k_2$ in \textsf{Rec}, we have
    \begin{align*}
        k_2 &= \left\lfloor m(v / g - \sigma_2 / q) \right\rceil \bmod m \\
        &= \left\lfloor m \left(\frac{1}{q}(\sigma_1 + k_1 q / m + \epsilon_1) +
           \frac{\epsilon_2}{g} + \theta - \frac{\sigma_2}{q}\right)\right\rceil \bmod m\\
    &= \left\lfloor k_1 + \frac{m}{q}\left(\sigma_1 - \sigma_2 \right)+
        \frac{m}{q} \epsilon_1 + \frac{m}{g} \epsilon_2 \right\rceil \bmod m
    \end{align*}
    By Fact~\ref{lemma:cir} (page \pageref{lemma:cir}), there exist  $\theta' \in \mathbb{Z}$ and $\delta \in [-d, d]$ such that
    $\sigma_1 = \sigma_2 + \theta' q + \delta$.  Hence,
    \begin{align*}
        k_2 = \left\lfloor k_1 + \frac{m}{q} \delta +
        \frac{m}{q} \epsilon_1 + \frac{m}{g} \epsilon_2 \right\rceil \bmod m
    \end{align*}
    Since $|m\delta / q + m \epsilon_1/q + m \epsilon_2 / g| \le md/q + m/2q +
    m/2g < 1/2$, $k_1 = k_2$.
  \end{proof}

\begin{theorem}\label{th-akcn-secure}
    The AKCN scheme is \emph{secure}. Specifically,
    $v$ is independent of $k_1$ when $\sigma_1\gets \mathbb{Z}_q$.
\end{theorem}


\begin{proof}
For arbitrary $\tilde{v}\in \mathbb{Z}_g$ and arbitrary  $\tilde{k}_1,\tilde{k}^\prime_1\in \mathbb{Z}_m$, we prove that $\Pr[v=\tilde{v}|k_1=\tilde{k}_1]=\Pr[v=\tilde{v}|k_1=\tilde{k}^\prime_1]$ when   $\sigma_1\gets \mathbb{Z}_q$.

    For any $(\tilde{k}, \tilde{v})$ in $\mathbb{Z}_m \times \mathbb{Z}_g$,
    the event $(v =\tilde{v} \mid k_1 = \tilde{k})$ is equivalent to the event that there
    exists $\sigma_1 \in \mathbb{Z}_q$ such that
    $\tilde{v} = \lfloor g(\sigma_1 + \lfloor \tilde{k} q / m \rceil)/q \rceil \bmod g$.
    Note that $\sigma_1 \in \mathbb{Z}_q$ satisfies $\tilde{v} = \lfloor g(\sigma_1 + \lfloor \tilde{k} q / m \rceil)/q \rceil \bmod g$,
    if and only if  there exist $\epsilon \in (-1/2, 1/2]$ and $\theta \in \mathbb{Z}$ such that
    $\tilde{v} = g(\sigma_1 + \lfloor \tilde{k} q / m \rceil) / q + \epsilon - \theta g$.
    That is, $\sigma_1 = (q (\tilde{v} - \epsilon) / g - \lfloor \tilde{k} q / m\rceil) \bmod q$,
    for some $\epsilon \in (-1/2, 1/2]$.
    Let $\Sigma(\tilde{v}, \tilde{k}) =
    \{\sigma_1 \in \mathbb{Z}_q \mid \exists \epsilon \in (-1/2, 1/2] \  s.t. \
        \sigma_1 = (q (\tilde{v} - \epsilon) / g - \lfloor \tilde{k} q / m\rceil)   \bmod q \}$.
      Defining   the map $\phi: \Sigma(\tilde{v}, 0) \rightarrow \Sigma(\tilde{v}, \tilde{k})$,  by setting $\phi(x) = \left(x -  \lfloor \tilde{k} q / m \rceil \right) \bmod q$. Then $\phi$ is obviously a one-to-one map.
    Hence, the cardinality of $\Sigma(\tilde{v}, \tilde{k})$ is irrelevant to $\tilde{k}$. Specifically, for arbitrary $\tilde{v}\in \mathbb{Z}_g$ and arbitrary  $\tilde{k}_1,\tilde{k}^\prime_1\in \mathbb{Z}_m$, it holds that  $\left|\Sigma(\tilde{v}, \tilde{k}_1)\right|=\left|\Sigma(\tilde{v}, \tilde{k}^\prime_1)\right|=\left|\Sigma(\tilde{v}, 0)\right|$

    Now, for arbitrary $\tilde{v}\in \mathbb{Z}_g$ and arbitrary $\tilde{k}\in \mathbb{Z}_m$,  when $\sigma_1\gets \mathbb{Z}_q$  we have that $\Pr[v = \tilde{v} \mid k_1 = \tilde{k}] =
    \Pr\left[\sigma_1 \in \Sigma(\tilde{v}, \tilde{k}) \mid k_1
        = \tilde{k}\right]=|\Sigma(\tilde{v}, \tilde{k})| / q = |\Sigma(\tilde{v}, 0)|/q$.  The right-hand side only depends on $\tilde{v}$, and so $v$ is independent of  $k_1$. \end{proof}

\subsubsection{A Simplified Variant of AKCN for Special Parameters} \label{section-special}
We consider the parameters $q = g = 2^{\bar{q}}, m = 2^{\bar{m}}$ for positive integers $\bar{q}, \bar{m}$.
Then the two rounding operations in line~\ref{op-akcn-v} of \textsf{Con} (in
Algorithm~\ref{kcs:2}) can be directly eliminated, since only integers are involved in the computation.
We have the following variant described in Algorithm \ref{kcs:2-power2}. Note that, in Algorithm~\ref{kcs:2-power2}, the modular and multiplication/division operations can be implemented by simple bitwise operations.


\begin{algorithm}[H]
\caption{AKCN power 2 }\label{kcs:2-power2}
\begin{algorithmic}[1]
\State{$\textsf{params}: q = g = 2^{\bar{q}}, m = 2^{\bar{m}}, aux = \{G = q/m\}$}
\Procedure{Con}{$\sigma_1, k_1, \textsf{params}$}
\State{$v = \left(\sigma_1 + k_1 \cdot G\right) \bmod q$, where $k_1\cdot G$ can be offline computed}
\Return{$v$}
\EndProcedure
\Procedure{Rec}{$\sigma_2, v, \textsf{params}$}
\State{$k_2 = \left\lfloor (v - \sigma_2) / G \right\rceil \bmod m$}
\label{op-akcn-power2-k2}
\Return{$k_2$}
\EndProcedure
\end{algorithmic}
\end{algorithm}

For the protocol variant presented in Algorithm \ref{kcs:2-power2}, its correctness and security can be proved with a relaxed constraint on the parameters of $(q,d,m)$, as shown in the following corollary.

\begin{corollary} \label{claim:2-power2-correct}
    If $q$ and $m$ are power of $2$, and $d$, $m$ and $q$ satisfy $2md < q$, then the AKCN scheme described in
    Algorithm~\ref{kcs:2-power2} is \emph{correct} and \emph{secure}.
\end{corollary}

\begin{proof}
    For correctness, suppose $|\sigma_1 - \sigma_2|_q \le d$, then there exit
    $\delta \in [-d, d]$ and $\theta \in \mathbb{Z}$
    such that $\sigma_2 = \sigma_1 +  \theta q + \delta$.
    From the formula calculating $v$, there exists $\theta' \in \mathbb{Z}$ such
    that $v = \sigma_1 + k_1 2^{\bar{q} - \bar{m}} + \theta' q$.
    Taking these into the formula computing $k_2$, line~\ref{op-akcn-power2-k2}
    of \textsf{Rec} in Algorithm~\ref{kcs:2-power2},
    we have
    \begin{align*}
        k_2 &= \lfloor \left(v - \sigma_1 - \delta - \theta q\right)/2^{\bar{q} - \bar{m}} \rceil \bmod m \\
        &= \lfloor (k_1 2^{\bar{q} - \bar{m}} - \delta) / 2^{\bar{q} - \bar{m}} \rceil \bmod m \\
        &= \left(k_1 - \lfloor \delta / 2^{\bar{q} - \bar{m}} \rceil \right) \bmod m
    \end{align*}
    If $2md < q$, then $|\delta / 2^{\bar{q} - \bar{m}}| < 1/2$, so that $k_1 = k_2$.

    For security, as a special case of the generic AKCN scheme in  Algorithm~\ref{kcs:2}, the
    security of the AKCN scheme in  Algorithm~\ref{kcs:2-power2} directly follows from that of
    Algorithm~\ref{kcs:2}.   \end{proof}

\subsection{On KC/AKC vs.  Fuzzy Extractor}
Our formulations of KC and AKC are abstractions of the core ingredients of previous constructions of KE and PKE from LWE/RLWE. As we shall see in the subsequent sections, the design and analysis of KE and PKE from LWE, LWR and RLWE can be reduced to KC and AKC. We also note that KC and AKC are similar to fuzzy extractor proposed in \cite{fuzzyextractor},
which extracts shared-keys from biometrics and noisy data.
In this section, we make some discussions on the relationship between KC/AKC and fuzzy extractor.

The differences between the  definitions of KC/AKC and that of fuzzy extractor lie mainly in the following  ways. Firstly, AKC was not considered within the definitional framework of fuzzy extractor.  Secondly, the metric $|\cdot|_q$ we use in defining  KC and AKC was not considered for fuzzy extractor. 
Thirdly,  in the definitions of KC and AKC, the algorithm $\textsf{Rec}$ (corresponding $\textsf{Rep}$ for fuzzy extractor) is mandated to be \emph{deterministic}, while in the formulation of fuzzy extractor it is probabilistic.  Fourthly, in the formulation of fuzzy extractor \cite{fuzzyextractor},  $w$, $R$ and $P$ (corresponding $\sigma_1$, $k$ and $v$ in KC/AKC) are binary strings; while in the  definitions of KC/AKC, the corresponding values $\sigma_1\in \mathbb{Z}_q$,  $k\in \mathbb{Z}_m$ and  $v\in \mathbb{Z}_g$  have  more structured ranges, which are helpful in deriving the exact upper bound.
  Finally, for the security of KC and AKC, we require that the signal value $v$ be independent of the shared-key $k_1$ (that can be subject to arbitrary distribution for AKC);
   roughly speaking, in the definition of fuzzy extractor \cite{fuzzyextractor},  it is required that the joint distribution $(R,P)$ be statistically close to $(U_l,P)$ where $R\in \{0,1\}^l$ and $U_l$ is the uniform distribution over $\{0,1\}^l$. 



A generic upper bound on the length of key  extracted by  fuzzy extractor is proposed in \cite[Appendix C]{fuzzyextractor}. In comparison, the  upper bounds for KC and AKC proved in this work are  more versatile and precise w.r.t. the  metric $|\cdot|_q$. For example, the effect of the length of the signal $v$, i.e., the bandwidth parameter $g$,   is not considered in the upper bound for fuzzy extractor, but is taken into account  in the  upper bounds for KC and AKC.

%

A  generic construction of fuzzy extractor from \emph{secure sketch}, together with a generic construction of \emph{secure sketch} for \emph{transitive metric spaces}, is proposed in \cite{fuzzyextractor}. We note  that $(\mathbb{Z}_q, |\cdot|_q)$ can be naturally seen  as a \emph{transitive matric space}.
Compared to the secure sketch based  generic constructions of fuzzy extractor, our constructions of KC and AKC are direct and more efficient.

{In spite of some similarities between KC/AKC and fuzzy extractors, we remark that before our this work the relation between fuzzy extractor and  KE from LWR and its variants is actually opaque. 
 Explicitly identifying and formalizing KC/AKC and reducing lattice-based  KE  to KC/AKC in a \emph{black-box} modular way,  with inherent bounds  on  what could or couldn't be done, cut the complexity of future design and analysis of these  cryptosystems.
}


%
%

\section{LWR-Based Key Exchange from KC and AKC}\label{sec-lwrke}
In this section, we present the applications of OKCN and AKCN to key exchange protocols based on LWR.
The LWR-based key exchange (KE) is depicted in \figurename~\ref{kex:lwr}.
Denote by $(n, l_A, l_B, q, p, KC, \chi)$ the system parameters,
where $p | q$, and  $p$ and $q$ are chosen to be power of $2$.
Let $KC = (\mathsf{params} = (p, m, g, d, aux),
\mathsf{Con}, \mathsf{Rec})$ be a \emph{correct} and \emph{secure}
key consensus scheme, $\chi$ be a small noise distribution over $\mathbb{Z}_q$, and
$\mathsf{Gen}$ be a pseudo-random generator (PRG) generating the matrix $\mathbf{A}$ from a small seed.
For presentation simplicity, we assume $\mathbf{A}\in \mathbb{Z}_q^{n\times n}$ to be square matrix.
The length of the random seed, i.e., $\kappa$, is typically set to be $256$.

The actual session-key is derived from $\mathbf{K}_1$ and $\mathbf{K}_2$ via some key derivation function $KDF$. 
For presentation simplicity, the functions \textsf{Con} and \textsf{Rec} are applied to
matrices, meaning that they are applied to each of the coordinates respectively.
We shall see that the corresponding elements of $\mathbf{\Sigma}_1$ and $\mathbf{\Sigma}_2$ are close to each other.
Then $\mathsf{Con}$ and $\mathsf{Rec}$ are applied to them to reach consensus bits $\mathbf{K}_1$ and $\mathbf{K}_2$.

For presentation simplicity, we describe the LWR-based key exchange protocol from a KC scheme. But it can be trivially adapted to work on any \emph{correct} and \emph{secure} AKC scheme. In this case, the responder user Bob simply chooses $\mathbf{K}_2\gets \mathbb{Z}_m^{l_A\times l_B}$, and the output of $\textsf{Con}(\boldsymbol{\Sigma}_2, \mathbf{K}_2, \textsf{params})$ is simply {defined to be $\mathbf{V}$}.  For presentation simplicity, in the following security definition and analysis we also simply assume that the output of the PRG $\textsf{Gen}$ is truly random (which is simply assumed to be a random oracle in \cite{newhope15}).

\noindent
\begin{figure}[t]
\centering
\includegraphics[scale=1]{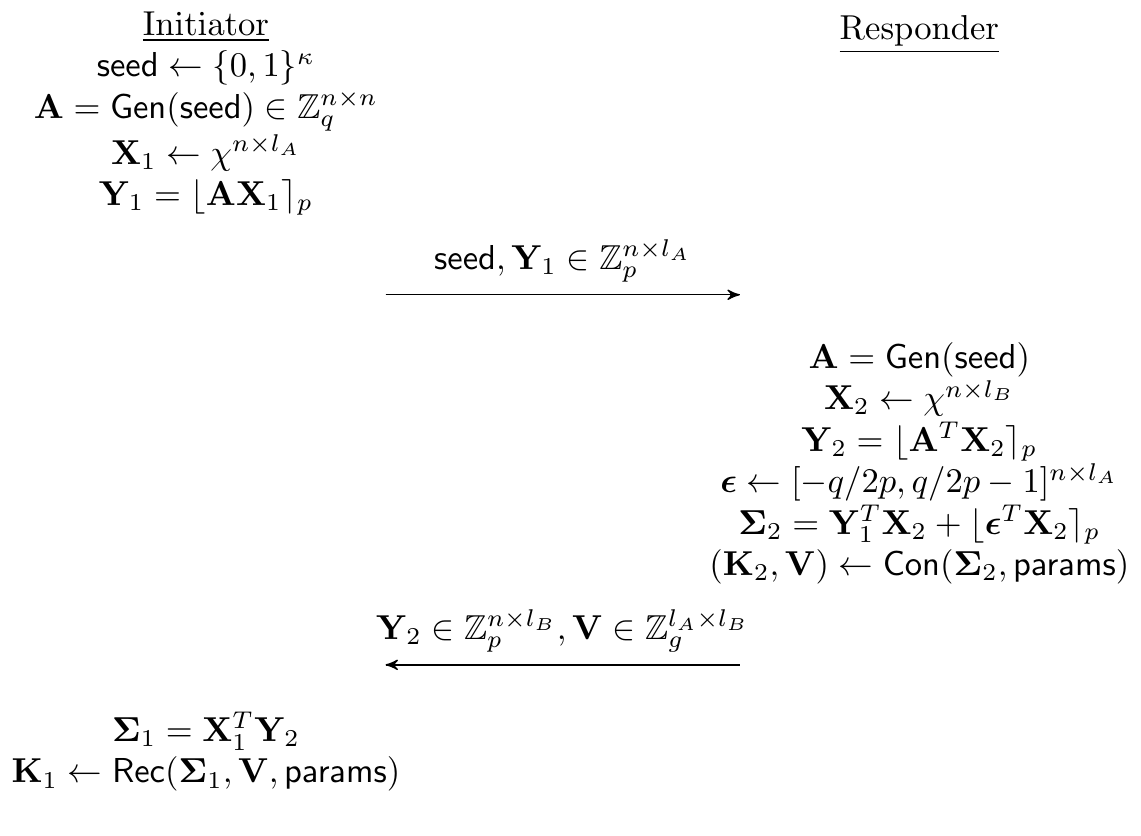}\label{ke-lwr}
\caption{LWR-based key exchange from KC and AKC,
where $\mathbf{K}_1,\mathbf{K}_2\in \mathbb{Z}_m^{l_A\times l_B}$ and $\left|\mathbf{K}_1\right|=\left|\mathbf{K}_2\right|=l_A\l_B|m|$.}
\label{kex:lwr}
\end{figure}

\subsection{Correctness Analysis}\label{SEC-LWR-1}
In this subsection, we analyze the correctness (specifically, the error probability) of the LWR-based KE protocol.
For any integer $x$, let $\{x\}_p$ denote $x - \frac{q}{p} \lfloor x \rceil_p$, where  $\lfloor x \rceil_p=\lfloor \frac{p}{q}x \rceil$.
Then, for any integer $x$, $\{x\}_p \in [-q/2p, q/2p-1]$,
hence $\{x\}_p$ can be naturally regarded as an element in $\mathbb{Z}_{q/p}$.
In fact, $\{x\}_p$ is equal  to  $x \bmod q/p$, where the result is represented in $[-q/2p, q/2p-1]$.
When the notation $\{\cdot\}_p$ is applied to a matrix, it means $\{\cdot\}_p$ applies to every element of the matrix respectively.

We have $\boldsymbol{\Sigma}_2 = \mathbf{Y}_1^T \mathbf{X}_2 + \lfloor \boldsymbol{\epsilon}^T \mathbf{X}_2 \rceil_p = \lfloor \mathbf{A} \mathbf{X}_1 \rceil_p^T \mathbf{X}_2 + \lfloor \boldsymbol{\epsilon}^T \mathbf{X}_2 \rceil_p
= \frac{p}{q} (\mathbf{A} \mathbf{X}_1 - \{\mathbf{A} \mathbf{X}_1\}_p)^T \mathbf{X}_2 + \lfloor \boldsymbol{\epsilon}^T \mathbf{X}_2 \rceil_p$. And $\boldsymbol{\Sigma}_1 = \mathbf{X}_1^T \mathbf{Y}_2 = \mathbf{X}_1^T \lfloor \mathbf{A}^T \mathbf{X}_2 \rceil_p
= \frac{p}{q} (\mathbf{X}_1^T \mathbf{A}^T \mathbf{X}_2 - \mathbf{X}_1^T \{\mathbf{A}^T \mathbf{X}_2\}_p)$.
Hence,
\begin{align*} \boldsymbol{\Sigma}_2 - \boldsymbol{\Sigma}_1 &=
\frac{p}{q}(\mathbf{X}_1^T \{\mathbf{A}^T \mathbf{X}_2\}_p - \{\mathbf{A} \mathbf{X}_1\}_p^T \mathbf{X}_2) +
\lfloor \mathbf{\epsilon}^T \mathbf{X}_2 \rceil_p \mod p \\
&= \left\lfloor\frac{p}{q}(\mathbf{X}_1^T \{\mathbf{A}^T \mathbf{X}_2\}_p - \{\mathbf{A} \mathbf{X}_1\}_p^T \mathbf{X}_2 +
 \mathbf{\epsilon}^T \mathbf{X}_2 )\right\rceil \mod p
\end{align*}

The general idea is that $\mathbf{X}_1, \mathbf{X}_2, \boldsymbol{\epsilon}, \{\mathbf{A}^T \mathbf{X}_2\}_p$ and
$\{\mathbf{A} \mathbf{X}_1\}_p$ are small enough,
so that $\boldsymbol{\Sigma}_1$ and $\boldsymbol{\Sigma}_2$ are close.
If $|\boldsymbol{\Sigma}_1 - \boldsymbol{\Sigma}_2|_p \le d$,
the \emph{correctness} of the underlying $KC$ guarantees $\mathbf{K}_1 = \mathbf{K}_2$.
For given concrete parameters,
we numerically derive the probability of $|\boldsymbol{\Sigma}_2 - \boldsymbol{\Sigma}_1|_p > d$
by numerically calculating the distribution of
$\mathbf{X}_1^T \{\mathbf{A}^T \mathbf{X}_2\}_p - (\{\mathbf{A} \mathbf{X}_1\}_p^T \mathbf{X}_2 - \epsilon^T\mathbf{X}_2)$
 for the case of $l_A = l_B = 1$, then applying the \emph{union bound}.
The independency between variables indicated by the following  Theorem~\ref{th:independent} can greatly simplify the calculation.

Let $\mathsf{Inv}(\mathbf{X}_1, \mathbf{X}_2)$ denote the event that there exist invertible elements of ring $\mathbb{Z}_{q/p}$
in both vectors $\mathbf{X}_1$ and $\mathbf{X}_2$.
$\mathsf{Inv}(\mathbf{X}_1, \mathbf{X}_2)$ happens with \emph{overwhelming} probability in our application.

\begin{lemma} \label{prop-s} Consider the  case of $l_A = l_B = 1$.
For any $a \in \mathbb{Z}_{q/p}, \mathbf{x} \in \mathbb{Z}^n_{q/p}$,
denote $S_{\mathbf{x}, a} = \{\mathbf{y} \in \mathbb{Z}_{q/p}^n \mid \mathbf{x}^T \mathbf{y} \bmod (q/p) = a\}$.
For any fixed $a \in \mathbb{Z}_{q/p}$, conditioned on $\mathsf{Inv}(\mathbf{X}_1, \mathbf{X}_2)$ and
$\mathbf{X}_1^T \mathbf{A}^T \mathbf{X}_2 \bmod (q/p) = a$,
the random vectors $\{\mathbf{A}^T \mathbf{X}_2\}_p$ and $\{\mathbf{A} \mathbf{X}_1\}_p$ are independent,
and are subjected to uniform distribution over $S_{\mathbf{X}_1, a}, S_{\mathbf{X}_2, a}$ respectively.
\end{lemma}

\begin{proof}
Under the condition of $\mathsf{Inv}(\mathbf{X}_1, \mathbf{X}_2)$,
for any fixed $\mathbf{X}_1$ and $\mathbf{X}_2$,
define the  map $\phi_{\mathbf{X}_1, \mathbf{X}_2}$:
$\mathbb{Z}_q^{n \times n} \rightarrow \mathbb{Z}_{q/p}^n \times \mathbb{Z}_{q/p}^n$, such that $\mathbf{A} \mapsto (\{\mathbf{A} \mathbf{X}_1\}_p, \{\mathbf{A}^T \mathbf{X}_2\}_p)$.

We shall prove that the image of $\phi_{\mathbf{X}_1, \mathbf{X}_2}$ is
$S = \{(\mathbf{y}_1, \mathbf{y}_2) \in \mathbb{Z}_{q/p}^n \times \mathbb{Z}_{q/p}^n
\mid \mathbf{X}_2^T \mathbf{y}_1 = \mathbf{X}_1^T \mathbf{y}_2 \mod (q/p) \}$.
Denote $\mathbf{X}_1 = (x_1, \mathbf{X}_1'^T)^T$
and $\mathbf{y}_2 = (y_2, \mathbf{y}_2'^T)^T$.
Without loss of generality,
we assume $x_1$ is invertible in the ring $\mathbb{Z}_{q/p}$.
For any $(\mathbf{y}_1, \mathbf{y}_2) \in S$, we need to find an $\mathbf{A}$
such that $\phi_{\mathbf{X}_1, \mathbf{X}_2}(\mathbf{A}) = (\mathbf{y}_1, \mathbf{y}_2)$.

From the condition $\mathsf{Inv}(\mathbf{X}_1, \mathbf{X}_2)$,
we know that there exists an  $\mathbf{A}' \in \mathbb{Z}^{(n-1) \times n}$ such that
$\{\mathbf{A}' \mathbf{X}_2\}_p = \mathbf{y}_2'$. Then, we let $\mathbf{a}_1 = x_1^{-1} (\mathbf{y}_1 - \mathbf{A}'^T \mathbf{X}_1') \bmod (q/p)$, and $\mathbf{A} = (\mathbf{a}_1, \mathbf{A}'^T)$.
Now we check that $\phi_{\mathbf{X}_1, \mathbf{X}_2}(\mathbf{A}) = (\mathbf{y}_1, \mathbf{y}_2)$.
\begin{align*}
    \{\mathbf{A} \mathbf{X}_1\}_p &=
    \left\{\begin{pmatrix}
        \mathbf{a}_1 & \mathbf{A}'^T
    \end{pmatrix}
    \begin{pmatrix}
        x_1 \\
        \mathbf{x}_1'
    \end{pmatrix}\right\}_p
    = \{ x_1 \mathbf{a}_1 + \mathbf{A}'^T \mathbf{X}_1' \}_p = \mathbf{y}_1  \\
    \{\mathbf{A}^T \mathbf{X}_2\}_p &=
    \left\{\begin{pmatrix}
        \mathbf{a}_1^T \\ \mathbf{A}'
    \end{pmatrix}
    \mathbf{X}_2 \right\}_p
    = \left\{ \begin{pmatrix}
     \mathbf{a}_1^T \mathbf{X}_2 \\ \mathbf{A}' \mathbf{X}_2
    \end{pmatrix} \right\}_p
    = \left\{ \begin{pmatrix}
     x_1^{-1}(\mathbf{y}_1^T - \mathbf{X}_1'^T \mathbf{A}) \mathbf{X}_2 \\
     \mathbf{A}' \mathbf{X}_2
    \end{pmatrix} \right\}_p \\
    &= \left\{ \begin{pmatrix}
     x_1^{-1}(\mathbf{X}_1^T \mathbf{y}_2 - \mathbf{X}_1'^T \mathbf{y}_2') \\
     \mathbf{y}_2'
    \end{pmatrix} \right\}_p
    = \left\{ \begin{pmatrix}
     y_2 \\
     \mathbf{y}_2'
    \end{pmatrix} \right\}_p
    = \mathbf{y}_2
\end{align*}

Hence, if we treat $\mathbb{Z}_q^{n \times n}$ and $S$ as $\mathbb{Z}$-modules,
then $\phi_{\mathbf{X}_1, \mathbf{X}_2}: \mathbb{Z}_q^{n \times n} \rightarrow S$ is a surjective homomorphism.
Then, for any fixed $(\mathbf{X}_1, \mathbf{X}_2)$,
$(\{\mathbf{A} \mathbf{X}_1\}_p, \{\mathbf{A}^T \mathbf{X}_2\}_p)$ is uniformly distributed over $S$. This completes the proof. 
  \end{proof}

\begin{theorem} \label{th:independent}
Under the condition $\mathsf{Inv}(\mathbf{X}_1, \mathbf{X}_2)$, the following two distributions are identical:

\begin{itemize}
\item $(a, \mathbf{X}_1, \mathbf{X}_2, \{\mathbf{A} \mathbf{X}_1\}_p, \{\mathbf{A}^T \mathbf{X}_2\}_p)$, where
 $\mathbf{A} \gets \mathbb{Z}_q^{n \times n}$, $\mathbf{X}_1 \gets \chi^{n}$,
$\mathbf{X}_2 \gets \chi^{n}$,
and  $a = \mathbf{X}_1^T \mathbf{A}^T \mathbf{X}_2 \bmod (q / p)$.

\item  $(a, \mathbf{X}_1, \mathbf{X}_2, \mathbf{y}_1, \mathbf{y}_2)$,  where
$a \gets \mathbb{Z}_{q/p}, \mathbf{X}_1 \gets \chi^{n}$, $
\mathbf{X}_2 \gets \chi^{n}$, $\mathbf{y}_1 \gets S_{\mathbf{X}_2, a}$, and
$\mathbf{y}_2 \gets S_{\mathbf{X}_1, a}$.
\end{itemize}

\end{theorem}

\begin{proof}
For any $\tilde{a}\in \mathbb{Z}_{q/p}$, $\tilde{\mathbf{X}}_1, \tilde{\mathbf{X}}_2 \in \mathsf{Supp}(\chi^{n})$,
$\tilde{\mathbf{y}}_1, \tilde{\mathbf{y}}_2\in \mathbb{Z}_{q/p}^n$,
we have
\begin{align*}
&\Pr[a = \tilde{a}, \mathbf{X}_1 = \tilde{\mathbf{X}}_1,
\mathbf{X}_2 = \tilde{\mathbf{X}}_2,
\{\mathbf{A} \mathbf{X}_1\}_p = \tilde{\mathbf{y}}_1,
\{\mathbf{A}^T \mathbf{X}_2\}_p = \tilde{\mathbf{y}}_2 \mid \mathsf{Inv}(\mathbf{X}_1, \mathbf{X}_2)] \\
=& \Pr[ \{\mathbf{A} \mathbf{X}_1\}_p = \tilde{\mathbf{y}}_1,
\{\mathbf{A}^T \mathbf{X}_2\}_p = \tilde{\mathbf{y}}_2
\mid a = \tilde{a}, \mathbf{X}_1 = \tilde{\mathbf{X}}_1,
\mathbf{X}_2 = \tilde{\mathbf{X}}_2, \mathsf{Inv}(\mathbf{X}_1, \mathbf{X}_2) ] \\
&\Pr[a = \tilde{a}, \mathbf{X}_1 = \tilde{\mathbf{X}}_1,
\mathbf{X}_2 = \tilde{\mathbf{X}}_2 \mid \mathsf{Inv}(\mathbf{X}_1, \mathbf{X}_2) ]
\end{align*}

From Lemma~\ref{prop-s}, the first term equals to $\Pr[\mathbf{y}_1 \gets S_{\tilde{\mathbf{X}}_2, \tilde{a}};
\mathbf{y}_2 \gets S_{\tilde{\mathbf{X}}_1, \tilde{a}}: \mathbf{y}_1 = \tilde{\mathbf{y}}_1,
\mathbf{y}_2 = \tilde{\mathbf{y}}_2 \mid a = \tilde{a}, \mathbf{X}_1 = \tilde{\mathbf{X}}_1,
\mathbf{X}_2 = \tilde{\mathbf{X}}_2, \mathsf{Inv}(\mathbf{X}_1, \mathbf{X}_2)]$.

For the second term, we shall prove that $a$ is independent of $(\mathbf{X}_1, \mathbf{X}_2)$,
and is uniformly distributed over $\mathbb{Z}_{q/p}$.
Under the condition of  $\mathsf{Inv}(\mathbf{X}_1, \mathbf{X}_2)$, the
map $\mathbb{Z}_q^{n \times n} \rightarrow \mathbb{Z}_{q/p}$, such that $\mathbf{A} \mapsto \mathbf{X}_1^T \mathbf{A}^T \mathbf{X}_2 \bmod (q/p)$,  is a surjective homomorphism
between the two $\mathbb{Z}$-modules.
Then, $\Pr[a = \tilde{a} \mid \mathbf{X}_1 =
\tilde{\mathbf{X}}_1, \mathbf{X}_2 = \tilde{\mathbf{X}}_2, \mathsf{Inv}(\mathbf{X}_1, \mathbf{X}_2)] = p/q$.
Hence, under the condition of  $\mathsf{Inv}(\mathbf{X}_1, \mathbf{X}_2)$,
$a$ is independent of $(\mathbf{X}_1, \mathbf{X}_2)$, and is distributed uniformly at random.
So the two ways of sampling result in  the same distribution.
  \end{proof}

We design and implement the following algorithm to numerically calculate the distribution of
$\boldsymbol{\Sigma}_2 - \boldsymbol{\Sigma}_1$ efficiently.
For any $c_1, c_2 \in \mathbb{Z}_q, a \in \mathbb{Z}_{q/p}$,
we numerically calculate $\Pr[ \mathbf{X}_1^T\{\mathbf{A}^T \mathbf{X}_2\}_p = c_1]$ and
$\Pr[ \{\mathbf{A} \mathbf{X}_1\}_p^T \mathbf{X}_2 - \boldsymbol{\epsilon}^T \mathbf{X}_2 = c_2,
\mathbf{X}_1^T \mathbf{A}^T \mathbf{X}_2 \bmod (q/p) = a]$,
then derive the distribution of $\mathbf{\Sigma}_2 - \mathbf{\Sigma}_1$.

As $\mathsf{Inv}(\mathbf{X}_1, \mathbf{X}_2)$ occurs with \emph{overwhelming} probability,
for any event $E$, we have
$|\Pr[E] - \Pr[E | \mathsf{Inv}(\mathbf{X}_1, \mathbf{X}_2)]| < negl$.
For simplicity, we ignore the effect of $\mathsf{Inv}(\mathbf{X}_1, \mathbf{X}_2)$
in the following calculations.
{
By Theorem~\ref{th:independent},
$\Pr[\mathbf{X}_1^T\{\mathbf{A}^T \mathbf{X}_2\}_p = c_1] =
\Pr[\mathbf{X}_1 \gets \chi^n, \mathbf{y}_2 \gets \mathbb{Z}_{q/p}^n; \mathbf{X}_1^T \mathbf{y}_2 = c_1]$.
This probability can be numerically calculated by computer programs.
The probability $\Pr[ \{\mathbf{A} \mathbf{X}_1\}_p^T \mathbf{X}_2 - \boldsymbol{\epsilon}^T \mathbf{X}_2 = c_2,
\mathbf{X}_1^T \mathbf{A}^T \mathbf{X}_2 \bmod (q/p) = a]$ can also be calculated by the similar way.
}Then, for arbitrary $c\in \mathbb{Z}_q$, 
\begin{align*}
&\Pr[\mathbf{\Sigma}_1 - \mathbf{\Sigma}_2 = c] =
\Pr[\mathbf{X}_1^T\{\mathbf{A}^T \mathbf{X}_2\}_p -
\{\mathbf{A} \mathbf{X}_1\}_p^T \mathbf{X}_2 + \boldsymbol{\epsilon}^T \mathbf{X}_2 = c] \\
=& \sum_{\substack{ c_1 - c_2 = c \\ a \in \mathbb{Z}_{q/p}}} \substack{
\Pr[\mathbf{X}_1^T\{\mathbf{A}^T \mathbf{X}_2\}_p = c_1,
\{\mathbf{A} \mathbf{X}_1\}_p^T \mathbf{X}_2 - \boldsymbol{\epsilon}^T \mathbf{X}_2 = c_2 \mid
\mathbf{X}_1^T \mathbf{A}^T \mathbf{X}_2 \bmod (q/p) = a] \cdot \\
\Pr[\mathbf{X}_1^T \mathbf{A}^T \mathbf{X}_2 \bmod (q/p) = a] } \\
=& \sum_{\substack{ c_1 - c_2 = c \\ a \in \mathbb{Z}_{q/p}}} \substack{
\Pr[\mathbf{X}_1^T\{\mathbf{A}^T \mathbf{X}_2\}_p = c_1 \mid \mathbf{X}_1^T \mathbf{A}^T \mathbf{X}_2 \bmod (q/p) = a ] \cdot \\
\Pr[\{\mathbf{A} \mathbf{X}_1\}_p^T \mathbf{X}_2 - \boldsymbol{\epsilon}^T \mathbf{X}_2 = c_2 \mid
\mathbf{X}_1^T \mathbf{A}^T \mathbf{X}_2 \bmod (q/p) = a]
\Pr[\mathbf{X}_1^T \mathbf{A}^T \mathbf{X}_2 \bmod (q/p) = a]} \\
=& \sum_{\substack{a \in \mathbb{Z}_{q/p} \\ c_1 - c_2 = c}}
\frac{\Pr[\mathbf{X}_1^T\{\mathbf{A}^T \mathbf{X}_2\}_p = c_1, c_1 \bmod (q/p) = a]
\Pr[\{\mathbf{A} \mathbf{X}_1\}_p^T \mathbf{X}_2 - \boldsymbol{\epsilon}^T \mathbf{X}_2 = c_2,
\mathbf{X}_1^T \mathbf{A}^T \mathbf{X}_2 \bmod (q/p) = a]}
{\Pr[\mathbf{X}_1^T \mathbf{A}^T \mathbf{X}_2 \bmod (q/p) = a]} \\
=& \sum_{\substack{a \in \mathbb{Z}_{q/p} \\ c_1 - c_2 = c \\ c_1 \bmod (q/p) = a}}
\frac{\Pr[\mathbf{X}_1^T\{\mathbf{A}^T \mathbf{X}_2\}_p = c_1]
\Pr[\{\mathbf{A} \mathbf{X}_1\}_p^T \mathbf{X}_2 - \boldsymbol{\epsilon}^T \mathbf{X}_2 = c_2,
\mathbf{X}_1^T \mathbf{A}^T \mathbf{X}_2 \bmod (q/p) = a]}
{\Pr[\mathbf{X}_1^T \mathbf{A}^T \mathbf{X}_2 \bmod (q/p) = a]}
\end{align*}


By Theorem~\ref{th:independent},
conditioned on $\mathsf{Inv(\mathbf{X}_1, \mathbf{X}_2)}$ and
$\mathbf{X}_1^T \mathbf{A}^T \mathbf{X}_2 \bmod (q/p) = a$,
$\mathbf{X}_1^T\{\mathbf{A}^T \mathbf{X}_2\}_p$ is independent of $\{\mathbf{A} \mathbf{X}_1\}_p^T \mathbf{X}_2 - \boldsymbol{\epsilon}^T \mathbf{X}_2$, which implies the second equality.
Our code and scripts are  available from  Github \url{http://github.com/OKCN}.


\subsection{Security Proof}\label{Sec-LWRproof}
\begin{definition} A KC or AKC based key exchange protocol from LWR is \emph{secure},
    if for any sufficiently large security parameter $\lambda$ and any PPT adversary $\mathcal{A}$,
    $\left|\Pr[b^\prime=b]-\frac{1}{2}\right|$ is negligible,  as defined w.r.t. game  $G_0$ specified in Algorithm \ref{G0}.\footnote{For presentation simplicity, we simply assume $\mathbf{K}_2^0 \gets \mathbb{Z}_m^{l_A \times l_B}$ when the key exchange protocol is implemented with AKC. However, when the AKC-based protocol is interpreted as a public-key encryption scheme,  $\mathbf{K}_2^0$ and $\mathbf{K}_2^1$  correspond to the plaintexts, which are taken independently at random from the same (arbitrary) distribution over $\mathbb{Z}_m^{l_A \times l_B}$.}

\begin{algorithm}[H]
\caption{Game $G_0$}\label{G0}
\begin{algorithmic}[1]
\State{$\mathbf{A} \gets \mathbb{Z}^{n \times n}_q$}
\State{$\mathbf{X}_1 \gets \chi^{n \times l_A}$}
\State{$\mathbf{Y}_1 = \lfloor \mathbf{A} \mathbf{X}_1 \rceil_p$}
\State{$\mathbf{X}_2 \gets \chi^{n \times l_B}$}
\State{$\boldsymbol{\epsilon} \gets \{-q/2p \dots q/2p-1\}^{n \times l_A}$}
\State{$\mathbf{Y}_2 = \lfloor \mathbf{A}^T \mathbf{X}_2 \rceil_p$}
\State{$\boldsymbol{\Sigma}_2 = \lfloor (\frac{q}{p}\mathbf{Y}_1 + \boldsymbol{\epsilon})^T \mathbf{X}_2 \rceil_p$}\Comment{$\boldsymbol{\Sigma}_2 = \mathbf{Y}_1^T \mathbf{X}_2 + \lfloor \boldsymbol{\epsilon}^T \mathbf{X}_2 \rceil_{p}=\lfloor (\frac{q}{p}\mathbf{Y}_1 + \boldsymbol{\epsilon})^T \mathbf{X}_2 \rceil_p$}
\State{$\left(\mathbf{K}_2^0, \mathbf{V}\right) \gets \textsf{Con}(\boldsymbol{\Sigma}_2, \textsf{params})$}
\State{$\mathbf{K}_2^1 \gets \mathbb{Z}_m^{l_A \times l_B}$}
\State{$b \gets \{0, 1\}$}
\State{$b' \gets \mathcal{A}(\mathbf{A}, \mathbf{Y}_1, \mathbf{Y}_2, \mathbf{K}_2^b, \mathbf{V})$}
\end{algorithmic}
\end{algorithm}
\end{definition}

Before starting to prove the security, we first recall some basic properties of the LWR assumption.
{The following lemma is derived by a hybrid argument,  similar to that of  LWE \cite{PVW08,bcd16}}.


\begin{lemma}[LWR problem in the matrix form]
    For positive integer parameters $(\lambda,n, q \ge 2, l, t)$, where $n,q,l,t$ all are polynomial in $\lambda$
    satisfying $p | q$,
    and a distribution $\chi$ over $\mathbb{Z}_q$,
    denote by $L_{\chi}^{(l, t)}$ the distribution over
    $\mathbb{Z}_q^{t \times n} \times \mathbb{Z}_p^{t \times l}$
    generated by taking
    $\mathbf{A} \gets \mathbb{Z}_q^{t \times n},
    \mathbf{S} \gets \chi^{n \times l}$ and outputting $(\mathbf{A}, \lfloor \mathbf{A} \mathbf{S} \rceil_p)$.
    Then, under the assumption on indistinguishability between
    $A_{q, \mathbf{s}, \chi}$ (with $\mathbf{s} \gets \chi^n$) and
    $\mathcal{U}(\mathbb{Z}_q^{n} \times \mathbb{Z}_p)$ within $t$ samples,
    no PPT distinguisher $\mathcal{D}$ can distinguish, with non-negligible probability,
    between the  distribution $L_{\chi}^{(l, t)}$ and
    $\mathcal{U}(\mathbb{Z}_q^{t \times n} \times \mathbb{Z}_p^{t \times l})$ for sufficiently large $\lambda$.
\end{lemma}

\begin{theorem}\label{LWR-security}
    If $(\textsf{params}, \textsf{Con}, \textsf{Rec})$ is a \emph{correct} and  \emph{secure} KC or AKC scheme, the key exchange protocol described in \figurename~\ref{kex:lwr} is \emph{secure} under the (matrix form of) LWR assumption.
\end{theorem}

\begin{proof} 
The proof is analogous  to   that in \cite{peikert2014lattice,bcd16}.
The general idea  is that we construct a sequence of games: $G_0$,
$G_1$ and $G_2$, where $G_0$ is the original game for defining security. In every move from game $G_i$
to $G_{i+1}$, $0\leq i\leq 1$,  we change a little.
All games $G_i$'s share the  same PPT adversary $\mathcal{A}$, whose goal is to distinguish between the  matrices chosen uniformly at random and the matrices generated in the actual  key exchange protocol.
Denote  by $T_i$, $0\leq i\leq 2$, the event that $b = b'$ in   Game $G_i$. Our goal is to prove that  $\Pr[T_0] < 1/2 + negl$, where $negl$ is a negligible function in $\lambda$.  For ease of readability, we re-produce game $G_0$ below. For presentation simplicity, in the subsequent analysis, we always assume the underlying KC or AKC is \emph{correct}.
The proof can be trivially extended to the case that correctness holds with \emph{overwhelming} probability (i.e., failure occurs with negligible probability).

\noindent
\begin{minipage}[H]{0.5\textwidth}
\null
\begin{algorithm}[H]
\caption{Game $G_0$}
\begin{algorithmic}[1]
\State{$\mathbf{A} \gets \mathbb{Z}^{n \times n}_q$}

\State{$\mathbf{X}_1 \gets \chi^{n \times l_A}$}
\State{$\mathbf{Y}_1 = \lfloor \mathbf{A} \mathbf{X}_1 \rceil_p$}

\State{$\mathbf{X}_2 \gets \chi^{n \times l_B}$}
\State{$\boldsymbol{\epsilon} \gets \{-q/2p \dots q/2p-1\}^{n \times l_A}$}
\State{$\mathbf{Y}_2 = \lfloor \mathbf{A}^T \mathbf{X}_2 \rceil_p$}
\State{$\boldsymbol{\Sigma}_2 = \lfloor (\frac{q}{p}\mathbf{Y}_1 + \boldsymbol{\epsilon})^T \mathbf{X}_2 \rceil_p$}
\State{$\left(\mathbf{K}_2^0, \mathbf{V}\right) \gets \textsf{Con}(\boldsymbol{\Sigma}_2, \textsf{params})$}
\State{$\mathbf{K}_2^1 \gets \mathbb{Z}_m^{l_A \times l_B}$}
\State{$b \gets \{0, 1\}$}
\State{$b' \gets \mathcal{A}(\mathbf{A}, \mathbf{Y}_1, \mathbf{Y}_2, \mathbf{K}_2^b, \mathbf{V})$}
\end{algorithmic}
\end{algorithm}
\end{minipage}
\begin{minipage}[H]{0.5\textwidth}
\null
\begin{algorithm}[H]
\caption{Game $G_1$}
\begin{algorithmic}[1]
\State{$\mathbf{A} \gets \mathbb{Z}^{n \times n}_q$}

\State{$\mathbf{X}_1 \gets \chi^{n \times l_A}$}
\State{{\color{red} $\mathbf{Y}_1 \gets \mathbb{Z}_p^{n \times l_A}$}}

\State{$\mathbf{X}_2 \gets \chi^{n \times l_B}$}
\State{$\boldsymbol{\epsilon} \gets \{-q/2p \dots q/2p-1\}^{n \times l_A}$}
\State{$\mathbf{Y}_2 = \lfloor \mathbf{A}^T \mathbf{X}_2 \rceil_p$}
\State{$\boldsymbol{\Sigma}_2 = \lfloor (\frac{q}{p}\mathbf{Y}_1 + \boldsymbol{\epsilon})^T \mathbf{X}_2 \rceil_p$}
\State{$\left(\mathbf{K}_2^0, \mathbf{V}\right) \gets \textsf{Con}(\boldsymbol{\Sigma}_2, \textsf{params})$}
\State{$\mathbf{K}_2^1 \gets \mathbb{Z}_m^{l_A \times l_B}$}
\State{$b \gets \{0, 1\}$}
\State{$b' \gets \mathcal{A}(\mathbf{A}, \mathbf{Y}_1, \mathbf{Y}_2, \mathbf{K}_2^b, \mathbf{V})$}
\end{algorithmic}
\end{algorithm}
\end{minipage}

\begin{lemma}\label{th-reduction-1}
       $|\Pr[T_0] - \Pr[T_1]| < negl$, under the indistinguishability between $L_{\chi}^{(l_A, n)}$ and
    {$\mathcal{U}(\mathbb{Z}_q^{n \times n} \times \mathbb{Z}_p^{n \times l_A})$}.
\end{lemma}
\begin{proof}
    Construct a distinguisher $\mathcal{D}$, in Algorithm \ref{D-1}, who tries to distinguish  $L_{\chi}^{(l_A, n)}$ from
    $\mathcal{U}(\mathbb{Z}_q^{n \times n} \times \mathbb{Z}_p^{n \times l_A})$.

\begin{algorithm}[H]
\caption{Distinguisher $\mathcal{D}$}\label{D-1}
\begin{algorithmic}[1]
    \Procedure{$\mathcal{D}$}{$\mathbf{A}, \mathbf{B}$}
\Comment{$\mathbf{A} \in \mathbb{Z}_q^{n \times n}, \mathbf{B} \in \mathbb{Z}_p^{n \times l_A}$}
    \State{$\mathbf{Y}_1 = \mathbf{B}$}
    \State{$\mathbf{X}_2 \gets \chi^{n \times l_B}$}
    \State{$\boldsymbol{\epsilon} \gets \{-q/2p \dots q/2p-1\}^{n \times l_A}$}
    \State{$\mathbf{Y}_2 = \lfloor \mathbf{A}^T \mathbf{X}_2 \rceil_p$}
    \State{$\boldsymbol{\Sigma}_2 = \lfloor (\frac{q}{p}\mathbf{Y}_1 + \boldsymbol{\epsilon})^T \mathbf{X}_2 \rceil_p$}
    \State{$\left(\mathbf{K}_2^0, \mathbf{V}\right) \gets \textsf{Con}(\boldsymbol{\Sigma}_2, \textsf{params})$}
    \State{$\mathbf{K}_2^1 \gets \mathbb{Z}_m^{l_A \times l_B}$}
    \State{$b \gets \{0, 1\}$}
    \State{$b' \gets \mathcal{A}(\mathbf{A}, \mathbf{Y}_1, \mathbf{Y}_2, \mathbf{K}_2^b, \mathbf{V})$}
    \If {$b' = b$}
        \Return $1$
    \Else
        \Return $0$
    \EndIf
\EndProcedure
\end{algorithmic}
\end{algorithm}

If $(\mathbf{A}, \mathbf{B})$ is subjected to $L_\chi^{(l_A, n)}$, then
$\mathcal{D}$ perfectly simulates $G_0$. Hence,
$\Pr\left[\mathcal{D}\left(L_\chi^{(l_A, n)}\right) = 1\right] = \Pr[T_0]$.
On the other hand, if $(\mathbf{A}, \mathbf{B})$ is chosen uniformly at random  from
$\mathbb{Z}_q^{n \times n} \times \mathbb{Z}_p^{n \times l_A}$, which is denoted  as $(\mathbf{A}^\mathcal{U},
\mathbf{B}^\mathcal{U})$,
then $\mathcal{D}$ perfectly  simulates $G_1$.
So $\Pr[\mathcal{D}(\mathbf{A}^\mathcal{U}, \mathbf{B}^\mathcal{U}) = 1] = \Pr[T_1]$.
Hence,
$\left|\Pr[T_0] - \Pr[T_1]\right| =
\left|\Pr[\mathcal{D}(L_\chi^{(l_A, n)}) = 1] -
    \Pr[\mathcal{D}(\mathbf{A}^\mathcal{U}, \mathbf{B}^\mathcal{U}) = 1]\right| < negl$.
  \end{proof}

\noindent
\begin{minipage}[H]{0.5\textwidth}
\null
\begin{algorithm}[H]
\caption{Game $G_1$}
\begin{algorithmic}[1]
\State{$\mathbf{A} \gets \mathbb{Z}^{n \times n}_q$}
\State{$\mathbf{X}_1, \mathbf{E}_1 \gets \chi^{n \times l_A}$}
\State{$\mathbf{Y}_1 \gets \mathbb{Z}_q^{n \times l_A}$}
\State{$\mathbf{X}_2 \gets \chi^{n \times l_B}$}
\State{$\boldsymbol{\epsilon} \gets \{-q/2p \dots q/2p-1\}^{n \times l_A}$}
\State{$\mathbf{Y}_2 = \lfloor \mathbf{A}^T \mathbf{X}_2 \rceil_p$}
\State{$\boldsymbol{\Sigma}_2 = \lfloor (\frac{q}{p}\mathbf{Y}_1 + \boldsymbol{\epsilon})^T \mathbf{X}_2 \rceil_p$}
\State{$\left(\mathbf{K}_2^0, \mathbf{V}\right) \gets \textsf{Con}(\boldsymbol{\Sigma}_2, \textsf{params})$}
\State{$\mathbf{K}_2^1 \gets \mathbb{Z}_m^{l_A \times l_B}$}
\State{$b \gets \{0, 1\}$}
\State{$b' \gets \mathcal{A}(\mathbf{A}, \mathbf{Y}_1, \mathbf{Y}_2, \mathbf{K}_2^b, \mathbf{V})$}
\end{algorithmic}
\end{algorithm}
\end{minipage}
\begin{minipage}[H]{0.5\textwidth}
\null
\begin{algorithm}[H]
\caption{Game $G_2$}
\begin{algorithmic}[1]
\State{$\mathbf{A} \gets \mathbb{Z}^{n \times n}_q$}

\State{$\mathbf{X}_1, \mathbf{E}_1 \gets \chi^{n \times l_A}$}
\State{$\mathbf{Y}_1 \gets \mathbb{Z}_q^{n \times l_A}$}

\State{$\mathbf{X}_2 \gets \chi^{n \times l_B}$}

\State{$\boldsymbol{\epsilon} \gets \{-q/2p \dots q/2p-1\}^{n \times l_A}$}
\State{{\color{red}$\mathbf{Y}_2 \gets \mathbb{Z}_p^{n \times l_B}$}}
\State{{\color{red}$\boldsymbol{\Sigma}_2 \gets \mathbb{Z}_p^{l_A \times l_B}$}}
\State{$\left(\mathbf{K}_2^0, \mathbf{V}\right) \gets \textsf{Con}(\boldsymbol{\Sigma}_2, \textsf{params})$}
\State{$\mathbf{K}_2^1 \gets \mathbb{Z}_m^{l_A \times l_B}$}
\State{$b \gets \{0, 1\}$}
\State{$b' \gets \mathcal{A}(\mathbf{A}, \mathbf{Y}_1, \mathbf{Y}_2, \mathbf{K}_2^b, \mathbf{V})$}
\end{algorithmic}
\end{algorithm}
\end{minipage}

\begin{lemma}\label{th-reduction-2}
$|\Pr[T_1] - \Pr[T_2]| < negl$, under the indistinguishability between $L_\chi^{(l_B, n + l_A)}$ and
$\mathcal{U}(\mathbb{Z}_q^{(n+l_A) \times n} \times \mathbb{Z}_p^{(n+l_A) \times l_B})$.
\end{lemma}
\begin{proof}
    As $\mathbf{Y}_1$ and $\boldsymbol{\epsilon}$ are subjected to uniform distribution in $G_1$,
    $\frac{p}{q}\mathbf{Y}_1 + \boldsymbol{\epsilon}$ is subjected to uniform distribution
    over $\mathbb{Z}_q^{n \times l_A}$.
    Based on this observation, we construct the following distinguisher $\mathcal{D}^\prime$.

First observe that 
    $\mathbf{Y}_1' = (\frac{q}{p} \mathbf{Y}_1 + \boldsymbol{\epsilon})
    \in \mathbb{Z}_q^{n\times l_A}$ follows the uniform distribution $\mathcal{U}(\mathbb{Z}_q^{n \times l_A})$, where $\mathbf{Y}_1 \gets \mathbb{Z}_q^{n \times l_A}$ and  $\boldsymbol{\epsilon} \gets [-q/2p, q/2p-1]^{n \times l_A}$.
{If $(\mathbf{A'}, \mathbf{B})$ is subject to $L_\chi^{(l_B, n + l_A)}$,
$\mathbf{A'} \gets \mathbb{Z}_q^{(n+l_A) \times n}$ 
corresponds to $\mathbf{A} \gets \mathbb{Z}_q^{n \times n}$ and $\mathbf{Y}_1' = \frac{q}{p} \mathbf{Y}_1 + \boldsymbol{\epsilon}$ 
in $G_1$;
And $\mathbf{S}\gets \chi^{n \times l_B}$ in generating
$(\mathbf{A'}, \mathbf{B})$
corresponds to $\mathbf{X}_2 \gets \chi^{n \times l_B}$ in $G_1$.
 In this case, we re-write}
\begin{align*}
\mathbf{B} &= \lfloor \mathbf{A'} \mathbf{S} \rceil_p =
\left\lfloor \left(
  \begin{array}{c}
    \mathbf{A}^T \\
    \mathbf{Y}_1'^T
  \end{array}
\right)
\mathbf{X}_2 \right\rceil_p \\
&=
\left(
  \begin{array}{c}
    \lfloor \mathbf{A}^T \mathbf{X}_2 \rceil_p \\
    \lfloor \mathbf{Y}_1'^T \mathbf{X}_2 \rceil_p
  \end{array}
\right)
=
\left(
  \begin{array}{c}
    \mathbf{Y}_2 \\
    \mathbf{\Sigma}_2
  \end{array}
\right)
\end{align*}
Hence $\Pr\left[\mathcal{D}^\prime\left(L_\chi^{(l_B, n + l_A)}\right) = 1\right] = \Pr[T_1]$.

{On the other hand, if $(\mathbf{A'}, \mathbf{B})$ is subject to uniform distribution $\mathcal{U}(\mathbb{Z}_q^{(n+l_A) \times n} \times \mathbb{Z}_p^{(n+l_A) \times l_B})$,
then $\mathbf{A}, \mathbf{Y}_1', \mathbf{Y}_2, \boldsymbol{\Sigma}_2$ 
all are also
uniformly random; So,  the view of  $\mathcal{D}^\prime$ in this case is the same as that in  game $G_2$. Hence,
$\Pr\left[\mathcal{D}^\prime\left(\mathbf{A'},
        \mathbf{B}\right) = 1\right] = \Pr[T_2]$ in this case.
%
Then,  $|\Pr[T_1] - \Pr[T_2]| =|\Pr[\mathcal{D}^\prime(L_\chi^{(l_B, n + l_A)}) = 1] -
\Pr[\mathcal{D}^\prime(\mathcal{U}(\mathbb{Z}_q^{(n+l_A) \times n} \times \mathbb{Z}_p^{(n+l_A) \times l_B})) = 1]| < negl$.}
  \end{proof}

\noindent
\begin{algorithm}[H]
    \caption{Distinguisher $\mathcal{D}^\prime$} \label{D-12}
\begin{algorithmic}[1]
    \Procedure{$\mathcal{D}^\prime$}{$\mathbf{A'}, \mathbf{B}$}
where {$\mathbf{A'} \in \mathbb{Z}_q^{(n + l_A) \times n}, \mathbf{B} \in \mathbb{Z}_p^{(n + l_A) \times l_B}$}
\State{Denote $\mathbf{A'} = \left(
      \begin{array}{c}
        \mathbf{A}^T \\
        \mathbf{Y}_1'^T
      \end{array}
    \right)
    $}
    \Comment{$\mathbf{A} \in \mathbb{Z}_q^{n \times n},
    \mathbf{Y}_1'^T = (\frac{q}{p} \mathbf{Y}_1 + \boldsymbol{\epsilon})^T
    \in \mathbb{Z}_q^{l_A \times n}$}
    \State{Denote $\mathbf{B} = \left(
      \begin{array}{c}
        \mathbf{Y}_2 \\
        \boldsymbol{\Sigma}_2
      \end{array}
        \right)$}
        \Comment{$\mathbf{Y}_2 \in \mathbb{Z}_p^{n \times l_B},
            \boldsymbol{\Sigma}_2 \in \mathbb{Z}_p^{l_A \times l_B}$}
        \State{$\left(\mathbf{K}_2^0, \mathbf{V}\right) \gets \textsf{Con}(\mathbf{\Sigma}_2, \textsf{params})$}
        \State{$\mathbf{K}_2^1 \gets \mathbb{Z}_m^{l_A \times l_B}$}
        \State{$b \gets \{0, 1\}$}
        \State{$b' \gets \mathcal{A}(\mathbf{A}, \lfloor \mathbf{Y}_1' \rceil_p, \mathbf{Y}_2, \mathbf{K}_2^b, \mathbf{V})$}
        \If {$b' = b$}
            \Return $1$
        \Else
            \Return $0$
        \EndIf
    \EndProcedure
\end{algorithmic}
\end{algorithm}

\begin{lemma}
    If the underlying KC or AKC is \emph{secure}, $\Pr[T2]=\frac{1}{2}$.
\end{lemma}

\begin{proof}
    Note that, in Game $G_2$,  for any
    $1 \le i \le l_A$ and $1 \le j \le l_B$, $\left(\mathbf{K}^0_2[i, j], \mathbf{V}[i, j]\right)$
    only depends on $\boldsymbol{\Sigma}_2[i, j]$, and $\boldsymbol{\Sigma}_2$
    is subject to uniform distribution.
    By the \emph{security} of KC, we have that, for each pair $(i,j)$,  $\mathbf{K}^0_2[i, j]$ and $\mathbf{V}[i, j]$ are independent, and
    $\mathbf{K}^0_2[i, j]$ is uniform distributed. Hence,  $\mathbf{K}^0_2$ and
    $\mathbf{V}$ are independent, and $\mathbf{K}^0_2$ is uniformly distributed, which implies that
    $\Pr[T_2] = 1/2$.
  \end{proof}

    This finishes the proof of Theorem \ref{LWE-security}.
  \end{proof}



\subsection{Parameter Selection and Evaluation} \label{param-select}

In this subsection, the concrete parameters for key exchange protocol from LWR are evaluated and chosen.


\subsubsection{Discrete Distribution}
It is suggested in \cite{newhope15,bcd16} that  rounded Gaussian distribution can be replaced
by discrete distribution that is very close to rounded Gaussian in the
sense of R{\'e}nyi divergence \cite{bai2015improved}.

\begin{definition}[\cite{bai2015improved}]
   For two discrete distributions $P, Q$ satisfying $\mathsf{Supp}(P) \subseteq \mathsf{Supp}(Q)$,   their $a$-order R{\'e}nyi divergence, for some $a > 1$, is
\begin{equation*}
R_a(P || Q) = \left( \sum_{x \in \mathsf{Supp}(P)}
\frac{P(x)^a}{Q(x)^{a-1}}\right)^{\frac{1}{a-1}}.
\end{equation*}
\end{definition}

\begin{lemma}[\cite{bai2015improved}]\label{lemma-bai15}
    Letting $a > 1$, $P$ and $Q$ are two discrete distributions satisfying $\mathsf{Supp}(P)
    \subseteq \mathsf{Supp}(Q)$, then we have
\begin{description}
\item [Multiplicativity:]
Let $P$ and $Q$  be two distributions of random variable $(Y_1, Y_2)$. For $i \in \{1, 2\}$,
let $P_i$ and $Q_i$ be the margin distribution of $Y_i$ over $P$ and $Q$ respectively.
If $Y_1$ and $Y_2$, under $P$ and $Q$ respectively, are independent, then
\begin{equation*}
R_a(P || Q) = R_a(P_1 || Q_1) \cdot R_a(P_2 || Q_2).
\end{equation*}

\item [Probability Preservation:]
Let $A \subseteq \mathsf{Supp}(Q)$ be an event, then
\begin{equation*}
Q(A) \ge P(A)^{\frac{a}{a-1}}/R_a(P || Q).
\end{equation*}
\end{description}
\end{lemma}

Note that, when the underlying key derivation function $KDF$ is modelled as a random oracle (as in \cite{bcd16,newhope15}),
 an attacker is considered to be successful only if it can recover the entire consensus bits.
Denote by $E$ the event that a PPT attacker can successfully and entirely recover the bits  of $\mathbf{K}_1=\mathbf{K}_2$. By Lemma \ref{lemma-bai15}, we have that
$\Pr_{\text{rounded Gaussian}}[E] > \Pr_{\text{discrete}}[E]^{a/(a-1)} /
R_a^{n\cdot (l_A + l_B) + l_A \cdot l_B}(\chi || \bar{\phi})$,  where $\bar{\phi}$ is the rounded Gaussian distribution, and $\chi$ is the discrete distribution.

\subsubsection{Proposed Parameters}

\noindent
\begin{table}[H]
\centering
\begin{tabular}{c|r|r|r r r r r r|c  c}
\hline
\multirow{2}{*}{dist.} & \multirow{2}{*}{bits} & \multirow{2}{*}{var.}
& \multicolumn{6}{c|}{probability of } & \multirow{2}{*}{order} & \multirow{2}{*}{divergence} \\
& &  & $0$ & $\pm{1}$ & $\pm{2}$ & $\pm{3}$ & $\pm{4}$ & $\pm{5}$ \\ \hline
$D_{R}$ & 16 & 1.70 & $19572$ & $14792$ & $6383$ & $1570$ & $220$ & $17$ &  500.0 & 1.0000396 \\
$D_{P}$ & 16 & 1.40 & $21456$ & $15326$ & $5580$ & $1033$ & $97$  & $4$  & 500.0 & 1.0000277 \\
\hline
\end{tabular}

\vspace{2mm}
\caption{Discrete distributions of every component in the LWR secret. We choose the standard variances  large enough to prevent  potential combinational attacks.}
\label{distribution-LWR}
\end{table}

\begin{table}[H]
\centering
\begin{tabular}{c rrr rrr c r rr}
\hline
            & $n$ & $q$ & $p$ & $l$ & $m$ & $g$ & distr. & bw.  & err. & $|\mathbf{K}|$ \\ \hline
Recommended & 680 & $2^{15}$ & $2^{12}$ & 8 & $2^4$ & $2^8$ & $D_{R}$& 16.39 & $2^{-35}$ & 256 \\ \hline
Paranoid    & 832 & $2^{15}$ & $2^{12}$ & 8 & $2^4$ & $2^8$ & $D_{P}$ & 20.03 & $2^{-34}$ & 256 \\ \hline
\end{tabular}
\vspace{2mm}
\caption{Parameters for LWR-Based key exchange.
``bw.'' refers to the bandwidth in kilo-bytes.
``err.'' refers to the overall error probability that is calculated by the algorithm developed in Section \ref{SEC-LWR-1}.
``$|\mathbf{K}|$'' refers to the length of consensus bits.}
\label{lwr-params}
\end{table}

\subsubsection{Security Estimation}

 Similar to  \cite{newhope15,bcd16,CKLS16}, we only consider  the  primal and dual  attacks \cite{Chen2011BKZ,Schnorr1994Lattice} adapted to the LWR problem.

The dual attack tries to distinguish the distribution of LWE samples and the uniform distribution.
Suppose $(\mathbf{A}, \mathbf{b} = \mathbf{A} \mathbf{s} + \mathbf{e}) \in \mathbb{Z}_q^{m \times n} \times \mathbb{Z}_q^m$ is a LWE sample,
where $\mathbf{s}$ and $\mathbf{e}$ are drawn from discrete Gaussian of variance $\sigma_s^2$ and $\sigma_e^2$ respectively.
Then we choose a positive real $c \in \mathbb{R}, 0 < c \le q$, and construct
$L_c(\mathbf{A}) = \left\{ (\mathbf{x}, \mathbf{y}/c) \in \mathbb{Z}^m \times (\mathbb{Z}/c)^n
\mid \mathbf{x}^T \mathbf{A} = \mathbf{y}^T \mod q \right\}$, which is a lattice with dimension $m + n$
and determinant $(q/c)^n$.
For a short vector $(\mathbf{x}, \mathbf{y}) \in L_c(\mathbf{A})$ found by the BKZ algorithm,
we have $\mathbf{x}^T \mathbf{b} = \mathbf{x}^T (\mathbf{A} \mathbf{s} + \mathbf{e}) = c \cdot \mathbf{y}^T \mathbf{s} + \mathbf{x}^T \mathbf{e} \mod q$.
If $(\mathbf{A}, \mathbf{b})$ is an LWE sample, the distribution of the right-hand side will be very close to a Gaussian of standard deviation
$\sqrt{c^2 \|\mathbf{y}\|^2 \sigma_s^2 + \|\mathbf{x}\|^2 \sigma_e^2}$, otherwise the distribution will be uniform.
$\|(\mathbf{x}, \mathbf{y})\|$ is about $\delta_0^{m + n} (q/c)^{\frac{n}{m + n}}$, where $\delta_0$ is the root Hermite factor.
We heuristically assume that $\|\mathbf{x}\| = \sqrt{\frac{m}{m + n}} \left\|(\mathbf{x}, \mathbf{y})\right\|$,
and $\|\mathbf{y}\| = \sqrt{\frac{n}{m + n}} \left\|(\mathbf{x}, \mathbf{y})\right\|$.
Then we can choose $c = \sigma_e/\sigma_s$ that minimizes the standard deviation of $\mathbf{x}^T \mathbf{b}$.
The advantage of distinguishing $\mathbf{x}^T \mathbf{b}$ from uniform distribution is $\epsilon = 4 \exp(-2\pi^2 \tau^2)$,
where $\tau = \sqrt{c^2 \|\mathbf{y}\|^2 \sigma_s^2 + \|\mathbf{x}\|^2 \sigma_e^2} / q$.
This attack must be  repeated $R = \max\{1, 1/(2^{0.2075b} \epsilon^2)\}$ times to be successful.

The primal attack reduces the LWE problem to the unique-SVP problem. Let $\Lambda_{w}(\mathbf{A}) = \{(\mathbf{x}, \mathbf{y}, z)
\in \mathbb{Z}^n \times (\mathbb{Z}^m / w) \times \mathbb{Z} \mid
\mathbf{A} \mathbf{x} + w \mathbf{y} = z \mathbf{b} \mod q \}$, and a vector $\mathbf{v}= (\mathbf{s}, \mathbf{e}/w, 1) \in \Lambda_{w}(\mathbf{A})$.
$\Lambda_{w}(\mathbf{A})$ is a lattice of  $d = m + n + 1$ dimensions, and its determinant is $(q/w)^m$.
From geometry series assumption, we can derive $\|\mathbf{b}_i^*\| \approx \delta_0^{d - 2i - 1} \det(\Lambda_w(\mathbf{A}))^{1/d}$.
We heuristically assume that the length of projection of $\mathbf{v}$ onto the vector space spanned by
the last $b$ Gram-Schmidt vectors is about $\sqrt{\frac{b}{d}} \left\|(\mathbf{s}, \mathbf{e}/w, 1)\right\|
\approx \sqrt{\frac{b}{d} \left(n \sigma_s^2 + m \sigma_e^2 / w^2 + 1\right)}$.
If this length is shorter than $\|\mathbf{b}_{d - b}^*\|$, this attack can be successful.
Hence, the successful condition is $\sqrt{\frac{b}{d} \left(n \sigma_s^2 + m \sigma_e^2 / w^2 + 1\right)} \le \delta_0^{2b - d - 1} \left(\frac{q}{w}\right)^{m/d}$.
We know that the optimal $w$ balancing the secret $\mathbf{s}$ and the noise $\mathbf{e}$ is about $\sigma_e / \sigma_s$.

   We aim at providing parameter sets for long term security, and  estimate the concrete security \emph{in a more conservative way} than \cite{concrete-hardness}
from the defender's point of view.
We first consider the attacks of LWE whose secret and noise have different variances.
Then,  we treat the LWR problem as a special LWE problem whose noise is uniformly distributed
over $[-q/2p, q/2p - 1]$.
In our security estimation, we simply ignore the difference between  the discrete distribution and
the rounded Gaussian, on the following grounds:  the  dual attack and the  primal attack only concern  about
the standard deviation, and the R{\'e}nyi divergence between the two distributions is very small.

\begin{table}[H]
\centering
\begin{tabular}{c|l|rrrrr}
\hline
Scheme & Attack & $m'$ & $b$ & C & Q & P  \\ \hline
\multirow{2}{*}{Recommended} & Primal & 667 & 461 & 143 & \bf{131} & 104 \\
& Dual & 631 & 458 & 142 & \bf{130} & 103 \\
\hline
\multirow{2}{*}{Paranoid} & Primal & 768 & 584 &  180 & 164 & \bf{130} \\
& Dual & 746 & 580 & 179 & 163 & \bf{129} \\
\hline
\end{tabular}
\vspace{2mm}
\caption{Security estimation of the parameters described in Table \ref{lwr-params}. ``C, Q, P'' stand for ``Classical, Quantum, Plausible'' respectively.
    Numbers under these columns are the binary logarithm of running
    time of the corresponding attacks.
    Numbers under ``$m', b$'' are the best parameters for the attacks.
}
\label{table:lwr-security}
\end{table}

\section{LWE-Based Key Exchange  from KC and AKC}\label{Sec-LWE}
In this section, following the protocol structure in \cite{peikert2014lattice,newhope15,bcd16}, we present the applications of OKCN and AKCN to key exchange  protocols based on LWE.  Denote by $(\lambda, n, q, \chi, KC, l_A, l_B, t)$ the underlying parameters, where $\lambda$ is the security parameter,  
  $q \ge 2$, $n$, $l_A$ and  $l_B$ are positive integers that  are  polynomial in $\lambda$ (for protocol symmetry, $l_A$ and $l_B$ are usually set to be equal and are actually small constant).
  To save bandwidth, we cut off $t$ least significant bits of $\mathbf{Y}_2$ before sending it to Alice.


Let $KC=(\textsf{params}, \textsf{Con}, \textsf{Rec})$ be a \emph{correct} and \emph{secure} KC scheme, where $\textsf{params}$ is set to be $(q,g,m,d)$.  The KC-based  key exchange protocol from LWE is depicted in \figurename~\ref{kex:lwe}, and the actual session-key is derived from $\mathbf{K}_1$ and $\mathbf{K}_2$ via some key derivation function $KDF$. There, for presentation simplicity, the 
 \textsf{Con} and \textsf{Rec} functions  are applied to
matrices, meaning they are applied to each of the coordinates separately.
Note that $2^t \mathbf{Y}_2' + 2^{t-1} \mathbf{1}$ is an approximation of $\mathbf{Y}_2$,
so we have
$\boldsymbol{\Sigma}_1 \approx \mathbf{X}_1^T \mathbf{Y}_2 =
\mathbf{X}_1^T \mathbf{A}^T \mathbf{X}_2 + \mathbf{X}_1^T \mathbf{E}_2$,
$\boldsymbol{\Sigma}_2 = \mathbf{Y}_1^T \mathbf{X}_2 + \mathbf{E}_{\sigma}
= \mathbf{X}_1^T \mathbf{A}^T \mathbf{X}_2 + \mathbf{E}_1^T \mathbf{X}_2 +
\mathbf{E}_{\sigma}$.
As we choose $\mathbf{X}_1, \mathbf{X}_2, \mathbf{E}_1, \mathbf{E}_2,
\mathbf{E}_{\sigma}$ according to a
small noise distribution $\chi$, the main part of $\boldsymbol{\Sigma}_1$ and that of
$\boldsymbol{\Sigma}_2$ are the same $\mathbf{X}_1^T \mathbf{A}^T \mathbf{X}_2$.
Hence, the corresponding coordinates of $\boldsymbol{\Sigma}_1$ and
$\boldsymbol{\Sigma}_2$ are close in the sense of $|\cdot|_q$, from which some key consensus can be reached. The failure probability depends upon the number of bits we cut $t$, the underlying distribution $\chi$ and the distance parameter $d$, which will be analyzed in detail in subsequent sections.   In the following security definition and analysis,  we  simply assume that the output of the PRG $\textsf{Gen}$ is truly random. For presentation simplicity, we have described the LWE-based key exchange protocol from a KC scheme. But it can be straightforwardly adapted to work on any correct and secure AKC scheme, as clarified in Section \ref{sec-lwrke}.




\begin{figure}[t]
\centering
\includegraphics{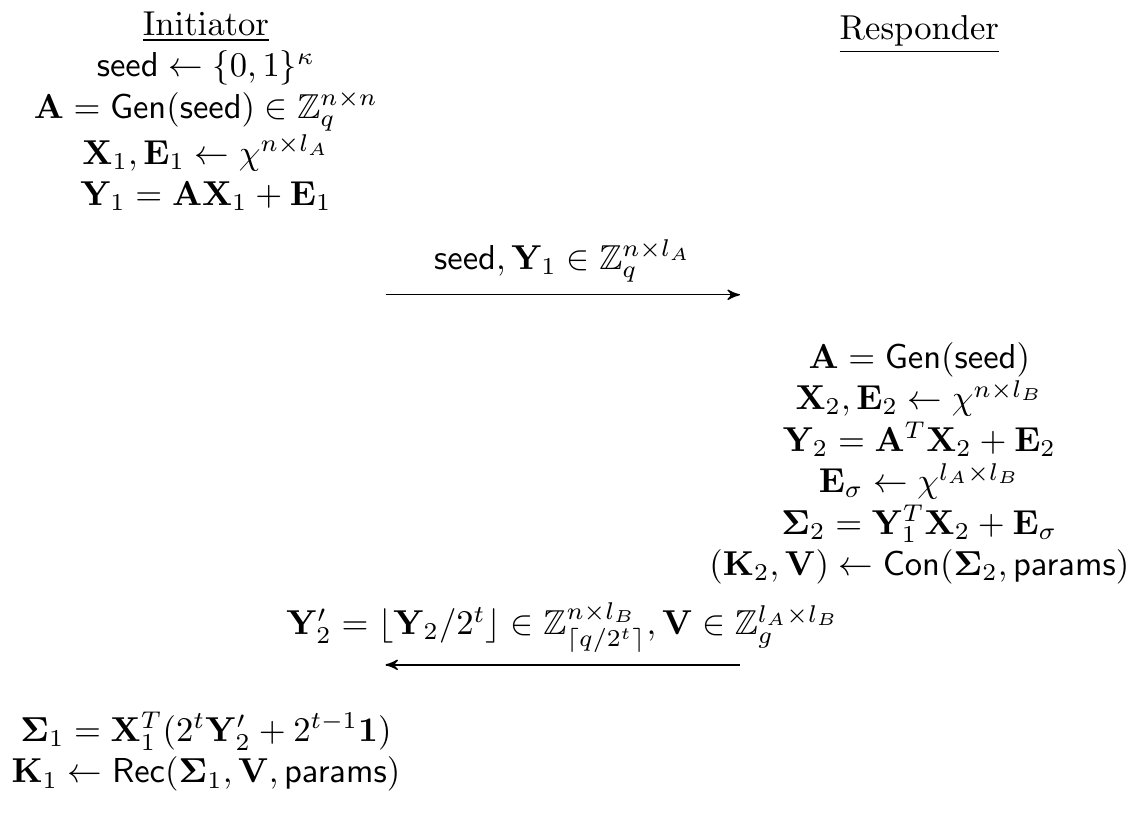}\label{ke-lwe}
\caption{LWE-based key exchange from KC and AKC, where $\mathbf{K}_1,\mathbf{K}_2\in \mathbb{Z}_m^{l_A\times l_B}$ and $\left|\mathbf{K}_1\right|=\left|\mathbf{K}_2\right|=l_A\l_B|m|$.
$\mathbf{1}$ refers to the matrix which every elements are 1.
}
\label{kex:lwe}
\end{figure}

By a straightforward adaption (actually simplification) of the security proof of  LWR-based key exchange protocol in Section \ref{Sec-LWRproof}, we have the following theorem. The detailed proof of Theorem \ref{Sec-LWE-security} is presented in Appendix \ref{App-LWE}.
\begin{theorem}\label{Sec-LWE-security}
    If $(\textsf{params}, \textsf{Con}, \textsf{Rec})$ is a \emph{correct} and  \emph{secure} KC or AKC scheme, the key exchange protocol described in \figurename~\ref{kex:lwe} is \emph{secure} under the (matrix form of) LWE assumption \cite{PVW08,bcd16}.
\end{theorem}

The correctness of the protocol depends upon the underlying error distributions, which are discussed  in the next subsection.

\subsection{Noise Distributions and Correctness} \label{subsection:failure-probability}

For a $\emph{correct}$ KC with parameter $d$, if the distance of corresponding elements of $\boldsymbol{\Sigma}_1$
and $\boldsymbol{\Sigma}_2$ is less than $d$ in the sense of $|\cdot|_q$,
then the scheme depicted in \figurename~\ref{kex:lwe} is correct.
Denote $\epsilon(\mathbf{Y}_2) = 2^t \lfloor \mathbf{Y}_2/2^t \rfloor + 2^{t-1} \mathbf{1} - \mathbf{Y}_2$.
Then
\begin{align*}
\boldsymbol{\Sigma}_1 - \boldsymbol{\Sigma}_2 &=
\mathbf{X}_1^T(2^t \mathbf{Y}_2' + 2^{t-1} \mathbf{1}) - \mathbf{Y}_1^T \mathbf{X}_2 - \mathbf{E}_{\sigma} \\
&= \mathbf{X}_1^T(\mathbf{Y}_2 + \epsilon(\mathbf{Y}_2)) - \mathbf{Y}_1^T \mathbf{X}_2 - \mathbf{E}_{\sigma} \\
&= \mathbf{X}_1^T(\mathbf{A}^T \mathbf{X}_2 + \mathbf{E}_2 + \epsilon(\mathbf{Y}_2)) -
(\mathbf{A} \mathbf{X}_1 + \mathbf{E}_1)^T \mathbf{X}_2 - \mathbf{E}_{\sigma} \\
&= \mathbf{X}_1^T (\mathbf{E}_2 + \epsilon(\mathbf{Y}_2)) - \mathbf{E}_1^T \mathbf{X}_2 - \mathbf{E}_{\sigma}
\end{align*}

We consider each pair of elements in matrix $\boldsymbol{\Sigma}_1, \boldsymbol{\Sigma}_2$ separately,
then derive the overall error rate by \emph{union bound}.
Now, we only need to consider the case $l_A = l_B = 1$.
In this case, $\mathbf{X}_i, \mathbf{E}_i, \mathbf{Y}_i, (i = 1, 2)$ are column vectors in $\mathbb{Z}_q^n$,
and $\mathbf{E}_{\sigma} \in \mathbb{Z}_q$.

If $\mathbf{Y}_2$ is independent of $(\mathbf{X}_2, \mathbf{E}_2)$,
then we can directly calculate the distribution of $\boldsymbol{\sigma}_1 - \boldsymbol{\sigma}_2$.
But now $\mathbf{Y}_2$ depends on $(\mathbf{X}_2, \mathbf{E}_2)$.
To overcome this difficulty, we show that $\mathbf{Y}_2$ is independent of $(\mathbf{X}_2, \mathbf{E}_2)$
under a condition of $\mathbf{X}_2$ that happens with very high probability.

\begin{theorem}\label{th-full-rank}
For any positive integer $q, n$, and a column vector $\mathbf{s} \in \mathbb{Z}_q^n$,
let $\phi_{\mathbf{s}}$ denote the map $\mathbb{Z}_q^n \rightarrow \mathbb{Z}_q:
\phi_{\mathbf{s}}(\mathbf{x}) = \mathbf{x}^T \mathbf{s}$.
If there exits a coordinate of $\mathbf{s}$ which is not \emph{zero divisor} in ring $\mathbb{Z}_q$,
then map $\phi_{\mathbf{s}}$ is surjective.
\end{theorem}

\begin{proof}
    Let us assume one coordinate of $\mathbf{s}$, say $s$, has no \emph{zero divisor} in ring $\mathbb{Z}_q$.
    Then the $\mathbb{Z}_q \rightarrow \mathbb{Z}_q$ map between the two $\mathbb{Z}_q$-modules deduced by $s$: $x \mapsto sx$, is injective,
    and thus surjective. Hence, $\phi_{\mathbf{s}}$ is surjective.
  \end{proof}

For a column vector $\mathbf{s}$ composed by random variables,
denote by $F(\mathbf{s})$ the event that $\phi_{\mathbf{s}}$ is surjective.
The following theorem gives a lower bound of probability of $F(\mathbf{s})$,
where $\mathbf{s} \gets \chi^{n}$. In our application, this lower bound is very close to $1$.

\begin{theorem} \label{sigularS}
    Let $p_0$ be the probability that $e$ is a \emph{zero divisor} in ring $\mathbb{Z}_q$,
    where $e$ is subject to $\chi$.
    Then $\Pr[\mathbf{s} \gets \chi^{n}: F(\mathbf{s})] \ge 1 - p_0^n$
\end{theorem}

\begin{proof}
From Theorem~\ref{th-full-rank}, if $\phi_{\mathbf{s}}$ is not surjective,
then all coordinates of $\mathbf{s}$ are \emph{zero divisors}.
Then $\Pr[\mathbf{s} \gets \chi^{n}: \neg F(\mathbf{s})] \le p_0^n$,
and the proof is finished.
  \end{proof}

\begin{theorem} \label{th-uniform}
    If $\mathbf{s}, \mathbf{e} \gets \chi^{n}, \mathbf{A} \gets \mathbb{Z}_q^{n \times n},
    \mathbf{y} = \mathbf{A} \mathbf{s} + \mathbf{e} \in \mathbb{Z}_q^n$,
    then under the condition $F(\mathbf{s})$, $\mathbf{y}$ is independent of $(\mathbf{s}, \mathbf{e})$,
    and is uniformly distributed over $\mathbb{Z}_q^{n}$.
\end{theorem}

\begin{proof}
    For all $\tilde{\mathbf{y}}, \tilde{\mathbf{s}}, \tilde{\mathbf{e}}$,
    $\Pr[\mathbf{y} = \tilde{\mathbf{y}} \mid \mathbf{s} = \tilde{\mathbf{s}}, \mathbf{e} = \tilde{\mathbf{e}}, F(\mathbf{s})]
    = \Pr[\mathbf{A}\tilde{\mathbf{s}} = \tilde{\mathbf{y}} - \tilde{\mathbf{e}}
    \mid \mathbf{s} = \tilde{\mathbf{s}}, \mathbf{e} = \tilde{\mathbf{e}}, F(\mathbf{s})]$.
    Let $\mathbf{A} = (\mathbf{a}_1, \mathbf{a}_2, \dots, \mathbf{a}_n)^T,
    \tilde{\mathbf{y}} - \tilde{\mathbf{e}} = (c_1, c_2, \dots, c_n)^T$,
    where $\mathbf{a}_i \in \mathbb{Z}_q^n$, and $c_i \in \mathbb{Z}_q$, for every $1 \le i \le n$.
    Since $\phi_{\mathbf{s}}$ is surjective, the number of possible choices of $\mathbf{a}_i$,  satisfying
    $\mathbf{a}_i^T \cdot \tilde{\mathbf{s}} = c_i$,  is $|\text{Ker} \phi_{\mathbf{s}}| = q^{n - 1}$.
    Hence, $\Pr[\mathbf{A}\tilde{\mathbf{s}} = \tilde{\mathbf{y}} - \tilde{\mathbf{e}}
    \mid \mathbf{s} = \tilde{\mathbf{s}}, \mathbf{e} = \tilde{\mathbf{e}}, F(\mathbf{s})] = (q^{n-1})^n/q^{n^2} = 1/q^n$.
    Since the right-hand side is the constant $1/q^{n}$, the distribution of $\mathbf{y}$ is uniform over $\mathbb{Z}_q^{n}$, and is irrelevant of $(\mathbf{s}, \mathbf{e})$.
  \end{proof}

We now begin to analyze the error probability of the scheme presented in \figurename~\ref{kex:lwe}.

Denote by $E$ the event $|\mathbf{X}_1^T (\mathbf{E}_2 + \epsilon(\mathbf{Y}_2)) - \mathbf{E}_1^T \mathbf{X}_2 - \mathbf{E}_{\sigma}|_q > d$.
Then $\Pr[E] = \Pr[E | F(\mathbf{S})] \Pr[F(\mathbf{S})] + \Pr[E | \neg  F(\mathbf{S})] \Pr[\neg  F(\mathbf{S})]$.
From Theorem~\ref{th-uniform}, we replace $\mathbf{Y}_2 = \mathbf{A}^T \mathbf{X}_2 + \mathbf{E}_2$ in the event $E | F(\mathbf{S})$ with
uniformly distributed $\mathbf{Y}_2$. Then,
\begin{align*}
    \Pr[E] &= \Pr[\mathbf{Y}_2 \gets \mathbb{Z}_q^{n}: E | F(\mathbf{S})] \Pr[F(\mathbf{S})] +
                \Pr[E | \neg  F(\mathbf{S})] \Pr[\neg  F(\mathbf{S})] \\
     &= \Pr[\mathbf{Y}_2 \gets \mathbb{Z}_q^{n}: E | F(\mathbf{S})] \Pr[F(\mathbf{S})] +
        \Pr[\mathbf{Y}_2 \gets \mathbb{Z}_q^{n}: E | \neg F(\mathbf{S})] \Pr[\neg  F(\mathbf{S})] \\
        &\quad +\Pr[E | \neg  F(\mathbf{S})] \Pr[\neg  F(\mathbf{S})]
                - \Pr[\mathbf{Y}_2 \gets \mathbb{Z}_q^{n}: E | \neg  F(\mathbf{S})] \Pr[\neg  F(\mathbf{S})] \\
                &= \Pr[\mathbf{Y}_2 \gets \mathbb{Z}_q^{n}: E] + \epsilon
\end{align*}
where $|\epsilon| \le \Pr[\neg  F(\mathbf{S})]$.
{In our application, $p_0$ is far from $1$, and $n$ is very large,
by Theorem~\ref{sigularS}, $\epsilon$ is very small, so we simply ignore $\epsilon$.}
If $\mathbf{Y}_2$ is uniformly distributed, then $\epsilon(\mathbf{Y}_2)$ is a centered uniform distribution.
Then, the  distribution of $\mathbf{X}_1^T (\mathbf{E}_2 + \epsilon(\mathbf{Y}_2)) - \mathbf{E}_1^T \mathbf{X}_2 - \mathbf{E}_{\sigma}$
can be directly computed by programs.

As noted in \cite{newhope15,bcd16}, sampling from rounded Gaussian distribution (i.e., sampling from a discrete Gaussian distribution to a high precision) constitutes one of major efficiency bottleneck. In this work, for LWE-based key exchange, we are mainly concerned with the following two kinds of efficiently sampleable distributions.

\subsubsection{Binary Distribution} 

Binary-LWE is a variation of LWE, where the noise distribution is set to be  $\chi =
\mathcal{U}(\{0, 1\})$. With respect to   $m=n \cdot (1 + \Omega(1/\log{n}))$ samples and large enough polynomial
$q \ge n^{O(1)}$, the hardness of binary-LWE is established in \cite{micciancio2013hardness}, with a reduction {from} some approximation lattice problem in dimension $\Theta(n/\log{n})$.
Concrete error probability can be calculated on the concrete parameters by the method in previous section.

For KC-based key exchange from binary-LWE, we have the following theorem,
which means that it is correct when the underlying parameter $d$  satisfies $d\ge n+1$. For LWE-based KE from OKCN, where $2md<q$, we get that it is correct when $q>2m(n+1)$.    Actually, this theorem has already been implied in  the above analysis.

\begin{theorem}
If $\chi = \mathcal{U}(\{0, 1\})$, and $(\textsf{params}, \textsf{Con}, \textsf{Rec})$ is a \emph{correct} KC or AKC
 scheme where $d \ge n + 1$, the key exchange protocol described in Algorithm \ref{kex:lwe} is correct. 
\end{theorem}
\begin{proof}
We prove that, for any  $(i,j)$,
$1 \le i \le l_A$ and $1 \le j \le l_B$,
$\left|\boldsymbol{\Sigma}_1[i, j]-\boldsymbol{\Sigma}_2[i, j]\right|_q\leq d$ holds true.
Denote $\mathbf{X}_1 = (\hat{\mathbf{X}}_1, \hat{\mathbf{X}}_2, \dots, \hat{\mathbf{X}}_{l_A})$,
$\mathbf{E}_1 = (\hat{\mathbf{E}}_1, \hat{\mathbf{E}}_2, \dots, \hat{\mathbf{E}}_{l_A})$,
and $\mathbf{X}_2 = (\hat{\mathbf{X}}_1', \hat{\mathbf{X}}_2', \dots, \hat{\mathbf{X}}_{l_B}')$,
$\mathbf{E}_2 = (\hat{\mathbf{E}}_1', \hat{\mathbf{E}}_2', \dots, \hat{\mathbf{E}}_{l_B}')$,
where $\hat{\mathbf{X}}_i, \hat{\mathbf{X}}_i', \hat{\mathbf{E}}_i,
\hat{\mathbf{E}}_i' \in \{0, 1\}^n$. Then
\begin{align*}
\left|\boldsymbol{\Sigma}_1[i, j]-\boldsymbol{\Sigma}_2[i, j]\right|_q =&
\left|\hat{\mathbf{X}}_i^T \hat{\mathbf{E}}_j' - \hat{\mathbf{E}}_i^T \hat{\mathbf{X}}_j' -
    \mathbf{E}_\sigma[i, j]\right|_q \\
\le& \left|\hat{\mathbf{X}}_i^T \hat{\mathbf{E}}_j' - \hat{\mathbf{E}}_i^T \hat{\mathbf{X}}_j'\right|_q +
    \left| \mathbf{E}_\sigma[i, j]\right|_q \\
\le& n + 1 \le d
\end{align*}
\end{proof}

However, cautions should be taken when deploying  key exchange protocols based upon binary-LWE.
By noting that any Binary-LWE sample satisfies a quadric equation,  if no less than $n^2$
samples can be used for an adversary, the secret and noise can be recovered easily.
The work \cite{arora2011new} proposes an algorithm for binary-LWE  with time
complexity $2^{\tilde{O}(\alpha q)^2}$. If $\alpha q = o(\sqrt{n})$, this
algorithm is subexponential, but it requires $2^{\tilde{O}((\alpha q)^2)}$ samples. When $m \log{q} / (n + m) = o(n / \log{n})$,
\cite{kirchner2015improved} proposes a distinguishing attack with time
complexity $2^{\frac{n/2 + o(n)}{\ln{(m\log{q}/(n+m) - 1)}}}$,
which results in a  subexponential-time algorithm if  $m$ grows linearly in  $n$.


\subsubsection{Discrete Distributions}
In this work, for LWE-based key exchange, we use the following two classes of discrete distributions, which are specified in Table \ref{table:discrete-distr} and Table \ref{table:frodo-distr} respectively,  where ``bits'' refers to the number of bits required to sample the distribution and
``var.'' means the standard variation of the Gaussian distribution approximated. We remark that the discrete distributions specified in Table \ref{table:frodo-distr} are just those specified and used in  \cite{bcd16} for the LWE-based  Frodo scheme.


\begin{table*}[h!]
\centering
\resizebox{0.7\textwidth}{!}{
\begin{tabular}{c|r|r|r r r r r r |c  c}
\hline
\multirow{2}{*}{dist.} & \multirow{2}{*}{bits} & \multirow{2}{*}{var.}
& \multicolumn{6}{c|}{probability of } & \multirow{2}{*}{order} & \multirow{2}{*}{divergence} \\
& &  & $0$ & $\pm{1}$ & $\pm{2}$ & $\pm{3}$ & $\pm{4}$ & $\pm{5}$ & & \\ \hline
$D_1$ & 8  & 1.10 &    94 &    62 &   17 &    2 &     &    & 15.0  & 1.0015832 \\
$D_2$ & 12 & 0.90 &  1646 &   992 &  216 &   17 &     &    & 75.0  & 1.0003146 \\
$D_3$ & 12 & 1.66 &  1238 &   929 &  393 &   94 &  12 &  1 & 30.0  & 1.0002034 \\
$D_4$ & 16 & 1.66 & 19794 & 14865 & 6292 & 1499 & 200 & 15 & 500.0 & 1.0000274 \\
$D_5$ & 16 & 1.30 & $22218$ & $15490$ & $5242$ & $858$ & $67$ & $2$ &  500.0 & 1.0000337 \\
\hline
\end{tabular}
}
\vspace{2mm}
\caption{Discrete distributions  proposed in this work, and their R{\'e}nyi divergences.}
\label{table:discrete-distr}
\end{table*}

\vspace{-1.5cm}

\begin{table*}[h!]
\centering
\resizebox{0.7\textwidth}{!}{
\begin{tabular}{c|r|r|r r r r r r r|c  c}
\hline
\multirow{2}{*}{dist.} & \multirow{2}{*}{bits} & \multirow{2}{*}{var.}
& \multicolumn{7}{c|}{probability of } & \multirow{2}{*}{order} & \multirow{2}{*}{divergence} \\
& &  & $0$ & $\pm{1}$ & $\pm{2}$ & $\pm{3}$ & $\pm{4}$ & $\pm{5}$ & $\pm{6}$ & & \\ \hline
$\bar{D}_1$ & 8  & 1.25 &    88 &    61 &   20 &    3 &     &    &   & 25.0  & 1.0021674 \\
$\bar{D}_2$ & 12 & 1.00 &  1570 &   990 &  248 &   24 &   1 &    &   & 40.0  & 1.0001925 \\
$\bar{D}_3$ & 12 & 1.75 &  1206 &   919 &  406 &  104 &  15 &  1 &   & 100.0 & 1.0003011 \\
$\bar{D}_4$ & 16 & 1.75 & 19304 & 14700 & 6490 & 1659 & 245 & 21 & 1 & 500.0 & 1.0000146 \\
\hline
\end{tabular}
}
\vspace{2mm}
\caption{Discrete distributions for Frodo \cite{bcd16}, and their R{\'e}nyi divergences}
    \label{table:frodo-distr}
\end{table*}

\begin{table*}[h!]
\centering
\resizebox{1\textwidth}{!}{
\begin{tabular}{c|rrrr|rr|rr|c|cc|r|c|c}
\hline
\multicolumn{1}{l|}{\multirow{2}{*}{}} & \multicolumn{1}{c}{\multirow{2}{*}{$q$}} & \multicolumn{1}{c}{\multirow{2}{*}{$n$}} & \multicolumn{1}{c}{\multirow{2}{*}{$l$}} & \multicolumn{1}{c|}{\multirow{2}{*}{$m$}} & \multicolumn{2}{c|}{$g$} & \multicolumn{2}{c|}{$d$} & \multirow{2}{*}{dist.} & \multicolumn{2}{c|}{error probability} & \multirow{2}{*}{bw. (kB)} & \multirow{2}{*}{$|A|$ (kB)} & \multirow{2}{*}{$|K|$} \\ \cline{6-9} \cline{11-12}
\multicolumn{1}{l|}{} & \multicolumn{1}{c}{} & \multicolumn{1}{c}{} & \multicolumn{1}{c}{} & \multicolumn{1}{c|}{}
& OKCN & Frodo & OKCN & Frodo & & OKCN & Frodo & & & \\ \hline
Challenge& $2^{10}$ & $334$ & $8$ & $2^1$ & $2^{9}$ & $2$ & $255$ & $127$ & $D_1$ & $2^{-47.9}$ & $2^{-14.9}$ & $ 6.75$ &  $139.45$ & $64$\\
Classical& $2^{11}$ & $554$ & $8$ & $2^2$ & $2^{9}$ & $2$ & $255$ & $127$ & $D_2$ & $2^{-39.4}$ & $2^{-11.5}$ & $12.26$ &  $422.01$ & $128$\\
Recommended  & $2^{14}$ & $718$ & $8$ & $2^4$ & $2^{10}$ & $2$ & $511$& $255$& $D_3$ & $2^{-37.9}$ & $2^{-10.2}$ & $20.18$ &  $902.17$ & $256$\\
Paranoid & $2^{14}$ & $818$ & $8$ & $2^4$ & $2^{10}$ & $2$ & $511$& $255$& $D_4$ & $2^{-32.6}$ & $2^{-8.6} $ & $22.98$ & $1170.97$ & $256$\\
Paranoid-512& $2^{12}$ & $700$ & $16$ & $2^2$ & $2^{10}$ & $2$ & $511$ & $255$ & $\bar{D}_4$ & $2^{-33.6}$ & $2^{-8.3}$ & $33.92$ &  $735.00$ & $512$\\
\hline
\end{tabular}
}

\vspace{0.2cm}

\caption{Parameters proposed for OKCN-LWE when $t=0$ (i.e., without cutting off least significant bits).  ``distr.'' refers to  the discrete distributions proposed in
    Table~\ref{table:discrete-distr} and Table \ref{table:frodo-distr}.
     ``bw.'' means bandwidth in kilo-bytes (kB).  ``$|\mathbf{A}|$" refers to the size of the matrix.
     $|\mathbf{K}|=l^2\log m$ denotes the length of consensus bits.}  \label{tabel:parameters-okcn}
\vspace{0.35cm}
\end{table*}

\begin{table*}[h!]
\centering
\resizebox{1\textwidth}{!}{
\begin{tabular}{c|rrrr|rr|rr|c|cc|rr|c|c}
\hline
\multicolumn{1}{l|}{\multirow{2}{*}{}} & \multicolumn{1}{c}{\multirow{2}{*}{$q$}} & \multicolumn{1}{c}{\multirow{2}{*}{$n$}} & \multicolumn{1}{c}{\multirow{2}{*}{$l$}} & \multicolumn{1}{c|}{\multirow{2}{*}{$m$}} & \multicolumn{2}{c|}{$g$} & \multicolumn{2}{c|}{$d$} & \multirow{2}{*}{dist.} & \multicolumn{2}{c|}{error probability} & \multicolumn{2}{c|}{bw. (kB)} & \multirow{2}{*}{$|A|$ (kB)} & \multirow{2}{*}{$|K|$} \\ \cline{6-9} \cline{11-14}
\multicolumn{1}{l|}{} & \multicolumn{1}{c}{} & \multicolumn{1}{c}{} & \multicolumn{1}{c}{} & \multicolumn{1}{c|}{}
& OKCN & Frodo & OKCN & Frodo & & OKCN & Frodo & OKCN & Frodo & & \\ \hline
Challenge& $2^{11}$ & $352$ & $8$ & $2^1$ & $2^{2}$ & $2$ & $383$ & $255$ & $\bar{D}_1$ & $2^{-80.1}$ & $2^{-41.8}$ & $ 7.76$ & $ 7.75$ & $170.37$ & $64$  \\
Classical& $2^{12}$ & $592$ & $8$ & $2^2$ & $2^{2}$ & $2$ & $383$ & $255$ & $\bar{D}_2$ & $2^{-70.3}$ & $2^{-36.2}$ & $14.22$ & $14.22$ & $525.70$ & $128$ \\
Recommended  & $2^{15}$ & $752$ & $8$ & $2^4$ & $2^{3}$ & $2$ & $895$ & $511$ & $\bar{D}_3$ & $2^{-105.9}$ & $2^{-38.9}$ & $22.58$ & $22.57$ & $1060.32$& $256$ \\
Paranoid & $2^{15}$ & $864$ & $8$ & $2^4$ & $2^{3}$ & $2$ & $895$ & $511$ & $\bar{D}_4$ & $2^{-91.9}$ & $2^{-33.8}$ & $25.94$ & $25.93 $& $1399.68$& $256$ \\ \hline
\end{tabular}
}
\vspace{0.3cm}
\caption{Parameters of Frodo, and comparison with OKCN-LWE when $t=0$.
Here, ``distr." refers to the discrete distributions specified in Table \ref{table:frodo-distr}. Note that, on the parameters of Frodo, OKCN-LWE achieves significantly lower error probability, which are negligible and are thus  sufficient for building  CPA-secure public-key encryption schemes. }
\label{tabel:parameters-frodo}
\vspace{0.3cm}
\end{table*}

\begin{table*}[h!]
\centering
\resizebox{0.8\textwidth}{!}{
\begin{tabular}{c|rrrrrrrrrcccrcc}
\hline
& $q$ & $n$ & $l$ & $m$ & $g$ & $t$ & $d$ & dist. & err. & bw. (kB) & $|A|$ (kB) & $|K|$ \\
\hline
OKCN-T2 & $2^{14}$ & $712$ & $8$ & $2^4$ & $2^{8}$ & $2$ & $509$ & $D_5$ & $2^{-39.0}$ & 18.58 & 887.15 & $256$ \\
OKCN-T1 & $2^{14}$ & $712$ & $8$ & $2^4$ & $2^{8}$ & $1$ & $509$ & $D_5$ & $2^{-52.3}$ & 19.29 & 887.15 & $256$ \\
\hline
\end{tabular}
}
\vspace{2mm}
\caption{Parameters proposed for OKCN-LWE with  $t$ least significant bits cut off.
     } \label{tabel:params-okcn-recommended}
     \end{table*}

\vspace{8mm}
\subsubsection{Instantiations, and Comparisons with Frodo}\label{sec-comparefrodo} The comparisons,  between the instantiations of our LWE-based KE protocol and Frodo, are summarized in the following tables \ref{tabel:parameters-okcn}, \ref{tabel:parameters-frodo} and \ref{tabel:params-okcn-recommended}.
 Note that, for presentation simplicity, we take $l_A=l_B=l$ for the sets of parameters under consideration. Also, for presentation simplicity, we use OKCN to denote OKCN-LWE in these tables.
  For both  ``OKCN simple'' proposed in Algorithm \ref{kcs:1-power2} and ``AKCN power 2'' proposed in Algorithm \ref{kcs:2-power2}, they achieve a tight parameter constraint, specifically, $2md<q$. In comparison, the parameter constraint achieved by Frodo is $4md<q$.
    As we shall see, such a difference is one source that allows us to achieve  better  trade-offs among error probability, security, (computational and bandwidth) efficiency,  and consensus range. 
    In particular, it allows us to use $q$ that is one bit shorter than that used in Frodo.
    Beyond saving bandwidth, employing a one-bit shorter $q$ also much improves the computational efficiency (as the matrix $\mathbf{A}$ becomes shorter, and consequently  the cost of generating $\mathbf{A}$ and the related matrix operations are  more efficient), and can render  stronger security levels simultaneously. {Note that, when being  transformed into PKE schemes,  smaller matrix $\mathbf{A}$ means shorter public-key size.}
{Put in other words, on the same parameters of $(q,m)$,  OKCN allows a larger error distance $d$ than Frodo, resulting in significantly lower error probability.}
{Here, we briefly highlight one  performance comparision:
 OKCN-T2  (resp., Frodo-recommended) has 18.58kB (resp., 22.57kB)  bandwidth, 887.15kB (resp., 1060.32kB) matrix $\mathbf{A}$, at least 134-bit (resp., 130-bit) quantum security,  and  error probability $2^{-39}$ (resp., $2^{-38.9})$.}

    The error probability for OKCN-LWE in these tables are derived by computing  $\Pr\left[\left|\boldsymbol{\Sigma}_1[i, j]-\boldsymbol{\Sigma}_2[i, j]\right|_q>d\right]$, for any $1\leq i \leq l_A$ and $1\leq j\leq l_B$,  and  then applying the union bound. The concrete failure probabilities are gotten by running the code  slightly adjusted, actually simplified,   from the open source code of Frodo. The simplified code are  available from  Github \url{http://github.com/OKCN}.




 \subsubsection{Security Estimation} For security evaluation, similar to \cite{newhope15,bcd16}, we only consider the resistance to two kinds of BKZ attacks, specifically  primal attack and dual attack \cite{Chen2011BKZ}\cite{Schnorr1994Lattice}, with respect to the core-SVP hardness. The reader is referred to \cite{newhope15,bcd16} for more details.  The concrete security levels are calculated by running the same code of Frodo.
 The evaluated security levels are summarized in the following tables.

%

\begin{table}[H]
\centering
\begin{tabular}{c|l|rrrrr|rrr}
\hline
\multirow{2}{*}{Scheme}    & \multirow{2}{*}{Attack} &
\multicolumn{5}{c|}{Rounded Gaussian} & \multicolumn{3}{c}{Post-reduction} \\
&        & $m'$   & $b$   & C   & Q   & P & C & Q & P  \\ \hline

\multirow{2}{*}{Challenge} & Primal & 327 & 275 &-- & -- & -- & -- & -- & -- \\
& Dual & 310 & 272 &-- & -- & -- & -- & -- & -- \\
\hline
\multirow{2}{*}{Classical} & Primal & 477 & 444 &138 & 126 & 100 & \bf{132} & 120 & 95 \\
& Dual & 502 & 439 &137 & 125 & 99 & \bf{131} & 119 & 94 \\
\hline
\multirow{2}{*}{Recommended} & Primal & 664 & 500 &155 & 141 & 112 & 146 & \bf{133} & 105 \\
& Dual & 661 & 496 &154 & 140 & 111 & 145 & \bf{132} & 104 \\
\hline
\multirow{2}{*}{Paranoid} & Primal & 765 & 586 &180 & 164 & 130 & 179 & 163 & \bf{130} \\
& Dual & 743 & 582 &179 & 163 & 129 & 178 & 162 & \bf{129} \\
\hline
\multirow{2}{*}{Paranoid-512} & Primal & 643 & 587 &180 & 164 & 131 & 180 & 164 & \bf{130} \\
& Dual & 681 & 581 &179 & 163 & 129 & 178 & 162 & \bf{129} \\
\hline
\multirow{2}{*}{OKCN-T2} & Primal & 638 & 480 &149 & 136 & 108 & 148 & \bf{135} & 107 \\
& Dual & 640 & 476 &148 & 135 & 107 & 147 & \bf{134} & 106 \\
\hline
\end{tabular}
\vspace{2mm}
\caption{Security estimation of the parameters described in Table \ref{tabel:parameters-okcn} and Table \ref{tabel:params-okcn-recommended}.
    ``C, Q, P'' stand for ``Classical, Quantum, Plausible'' respectively.
    Numbers under these columns are the binary logarithm of running
    time of the corresponding attacks.
    Numbers under ``$m', b$'' are the best parameters for the attacks. ``Rounded Gaussian''
    refers to the ideal case that noises and errors  follow the  rounded Gaussian
    distribution. ``Post-reduction'' refers to the case of using  discrete distributions as specified
    in Table~\ref{table:discrete-distr}.}
\label{table:okcn-security}
\end{table}

\begin{table}[H]
\centering
\begin{tabular}{c|l|rrrrr|rrr}
\hline
\multirow{2}{*}{Scheme}    & \multirow{2}{*}{Attack} &
\multicolumn{5}{c|}{Rounded Gaussian} & \multicolumn{3}{c}{Post-reduction} \\
&        & $m'$   & $b$   & C   & Q   & P & C & Q & P  \\ \hline

\multirow{2}{*}{Challenge} & Primal & 338 & 266 &-- & -- & -- & -- & -- & -- \\
& Dual & 331 & 263 &-- & -- & -- & -- & -- & -- \\
\hline
\multirow{2}{*}{Classical} & Primal & 549 & 442 &138 & 126 & 100 & \bf{132} & 120 & 95 \\
& Dual & 544 & 438 &136 & 124 & 99 & \bf{130} & 119 & 94 \\
\hline
\multirow{2}{*}{Recommended} & Primal & 716 & 489 &151 & 138 & 110 & 145 & \bf{132} & 104 \\
& Dual & 737 & 485 &150 & 137 & 109 & 144 & \bf{130} & 103 \\
\hline
\multirow{2}{*}{Paranoid} & Primal & 793 & 581 &179 & 163 & 129 & 178 & 162 & \bf{129} \\
& Dual & 833 & 576 &177 & 161 & 128 & 177 & 161 & \bf{128} \\
\hline
\end{tabular}
\vspace{2mm}
\caption{Security estimation of the parameters proposed for  Frodo in \cite{bcd16},  as specified in Table  \ref{tabel:parameters-frodo}.}
\label{table:frodo-security}
\end{table}


\section{Hybrid Construction of Key Exchange from LWE and LWR in the Public-Key Setting}
By composing a CPA-secure symmetric-key  encryption scheme, the LWE-based  key exchange protocols presented  Section \ref{Sec-LWE} can be used to construct public-key encryption (PKE) schemes, by treating $(\mathbf{A}, \mathbf{Y}_1)$ (resp., $\mathbf{X}_1$) as the static public key (resp., secret key).  Moreover, AKC-based key-exchange protocol  can be directly used as a CPA-secure PKE scheme. To further improve the efficiency of the resultant PKE scheme, the observation here is we can generate the ephemeral $\textbf{Y}_2$ in the ciphertext with  LWR samples. This results in the following hybrid construction of key exchange from LWE and LWR in the public-key setting. For applications to PKE, we focus on the AKC-based protocol construction.  Denote by $(n_A, n_B, l_A, l_B, q, p, KC, \chi)$ the system parameters,
where $p | q$, and we choose $p$ and $q$ to be power of $2$. The AKC-based protocol from LWE and LWR is  presented in Figure \ref{kex:lwe-lwr}.


\begin{figure}[t]
\centering
\includegraphics[width=0.8\textwidth]{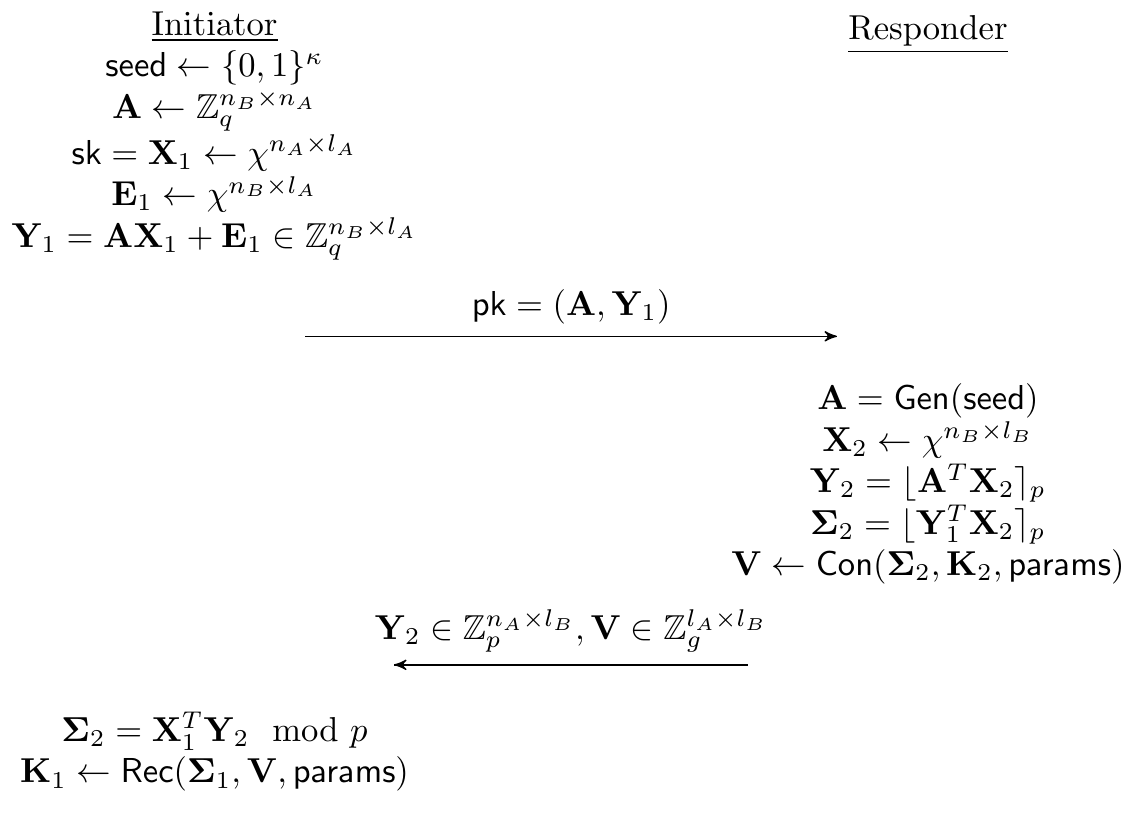}
\caption{AKC-based key exchange from LWE and LWR in the public-key setting,
where $\mathsf{pk} = (\mathbf{A}, \mathbf{Y}_1)$ is fixed once and for all,  $\mathbf{K}_1,\mathbf{K}_2\in \mathbb{Z}_m^{l_A\times l_B}$ and $\left|\mathbf{K}_1\right|=\left|\mathbf{K}_2\right|=l_A\l_B|m|$.}
\label{kex:lwe-lwr}
\end{figure}

The hybird construction of key exchange from LWE and LWR is similar to the underlying protocol in Lizard \cite{CKLS16}. The Lizard PKE scheme uses our AKCN as the underlying reconciliation mechnism (see more details in Appendix \ref{sec-lizard}), while our protocol is a general structure that can be implemented with either  KC or AKC. 
In order to improve efficiency, Lizard
\cite{CKLS16} is based on the variants, referred to as spLWE and spLWR,  of LWE and LWR with sparse secret. We aim at providing parameter sets for long term security, and  estimate the concrete security in a more conservative way than \cite{CKLS16}
from the defender's point of view.

\subsection{Security and Error Rate Analysis}
The security proof is very similar to LWE-basd and LWR-based key exchanges in previous sections, and is omitted here.

For the error probability, we have
\begin{align*}
\boldsymbol{\Sigma}_1 &= \mathbf{X}_1^T \mathbf{Y}_2
= \frac{p}{q}\mathbf{X}_1^T\left(\mathbf{A}^T \mathbf{X}_2 - \{\mathbf{A}^T \mathbf{X}_2\}_p\right)
= \frac{p}{q}\left(\mathbf{X}_1^T \mathbf{A}^T \mathbf{X}_2 - \mathbf{X}_1^T \{\mathbf{A}^T \mathbf{X}_2\}_p \right) \\
\boldsymbol{\Sigma}_2 &= \left\lfloor \mathbf{Y}_1^T \mathbf{X}_2 \right\rceil_p =
\frac{p}{q} \left( \mathbf{Y}_1^T \mathbf{X}_2 - \{\mathbf{Y}_1^T \mathbf{X}_2 \}_p\right)
= \frac{p}{q} (\mathbf{X}_1^T \mathbf{A}^T \mathbf{X}_2 + \mathbf{E}_1^T \mathbf{X}_2 - \{\mathbf{Y}_1^T \mathbf{X}_2\}_p) \\
\boldsymbol{\Sigma}_2 - \boldsymbol{\Sigma}_1 &= \frac{p}{q}
\left(\mathbf{E}_1^T \mathbf{X}_2 + \mathbf{X}_1^T \{\mathbf{A}^T \mathbf{X}_2\}_p -
\{ \mathbf{E}_1^T \mathbf{X}_2 + \mathbf{X}_1^T \mathbf{A}^T \mathbf{X}_2 \}_p \right)
= \lfloor \mathbf{E}_1^T \mathbf{X}_2 + \mathbf{X}_1^T \{\mathbf{A}^T \mathbf{X}_2\}_p \rceil_p
\end{align*}

We can see that the distribution of $\boldsymbol{\Sigma}_2 - \boldsymbol{\Sigma}_1$
can be derived from the distribution of $\mathbf{E}_1 \mathbf{X}_2 + \mathbf{X}_1^T \{\mathbf{A}^T \mathbf{X}_2\}_p$.
From Theorem~\ref{th-uniform}, we know that for almost  all (with \emph{overwhelm} probability)  given $\mathbf{X}_2$,
the distribution of $\{\mathbf{A}^T \mathbf{X}_2\}_p$ is the uniform distribution over $[-q/2p, q/2p)^{n_{A}}$.
The concrete error probability can then  be derived numerically by computer programs. The codes and scripts are available on Github \url{http://github.com/OKCN}.

\subsection{Parameter Selection}
For simplicity, we use the Gaussian distribution of the same variance (denote as $\sigma_s^2$) for the noise $\mathbf{E}_1$,
secrets $\mathbf{X}_1$ and $\mathbf{X}_2$.
We consider the weighted dual attack and weighted primal attack in Section \ref{param-select}.

\begin{table}[H]
\centering
\begin{tabular}{c r rrrr rrr c rr rr}
\hline
& $\sigma_s^2$ & $n_A$ & $n_B$ & $q$ & $p$ & $l$ & $m$ & $g$ & pk & cipher  & err. & $|\mathbf{K}|$ \\ \hline
Recommended & 2.0 & 712 & 704 & $2^{15}$ & $2^{12}$ & 8 & $2^4$ & $2^8$ & 10.56 &  8.61 & $2^{-63}$ & 256 \\ \hline
Paranoid    & 2.0 & 864 & 832 & $2^{15}$ & $2^{12}$ & 8 & $2^4$ & $2^8$ & 12.24 & 10.43 & $2^{-52}$ & 256 \\ \hline
\end{tabular}
\caption{Parameters for the hybrid construction of  key exchange from LWE and LWR.
``err.'' refers to the overall error probability.
``$|\mathbf{K}|$'' refers to the length of consensus bits.
``pk'' refers to the kilo-byte (kB) size of the public key $pk=(\textbf{A},\textbf{Y}_1)$.
``cipher'' refers to the kB size of $(\textbf{Y}_2,\textbf{V})$.
}
\label{lwe-lwr-params}
\end{table}

\begin{table}[H]
\centering
\resizebox{0.8\textwidth}{!}{
\begin{tabular}{c|l|rrrrr|rrrrr}
\hline
\multirow{2}{*}{Scheme} & \multirow{2}{*}{Attack} & \multicolumn{5}{c|}{LWE} & \multicolumn{5}{c}{LWR} \\ \cline{3-12}
 & & $m'$ & $b$ & C & Q & P & $m'$ & $b$ & C & Q & P \\ \hline
\multicolumn{1}{c|}{\multirow{2}{*}{Recommended}} & \multicolumn{1}{l|}{Primal} & 699 & 464
& 144 & \bf{131} & 105 & 664 & 487 & 151 & \bf{138} & 109 \\
\multicolumn{1}{c|}{}                             & \multicolumn{1}{l|}{Dual}   & 672 & 461
& 143 & \bf{131} & 104 & 665 & 483 & 150 & \bf{137} & 109 \\ \hline
\multicolumn{1}{c|}{\multirow{2}{*}{Paranoid}}    & \multicolumn{1}{l|}{Primal} & 808  & 590
& 181 & 165 & \bf{131} & 856 & 585 & 180 & 164 & \bf{130} \\
\multicolumn{1}{c|}{}                             & \multicolumn{1}{l|}{Dual}   & 789  & 583
& 179 & 163 & \bf{130} & 765 & 579 & 178 & 162 & \bf{129} \\ \hline
\end{tabular}}
\caption{Security estimation of the parameters described in Table \ref{lwe-lwr-params}.}
\label{table:lwr-security}
\end{table}

\section{RLWE-Based Key Exchange from KC and AKC}

\begin{figure}[H]
\centering
\includegraphics[scale=1]{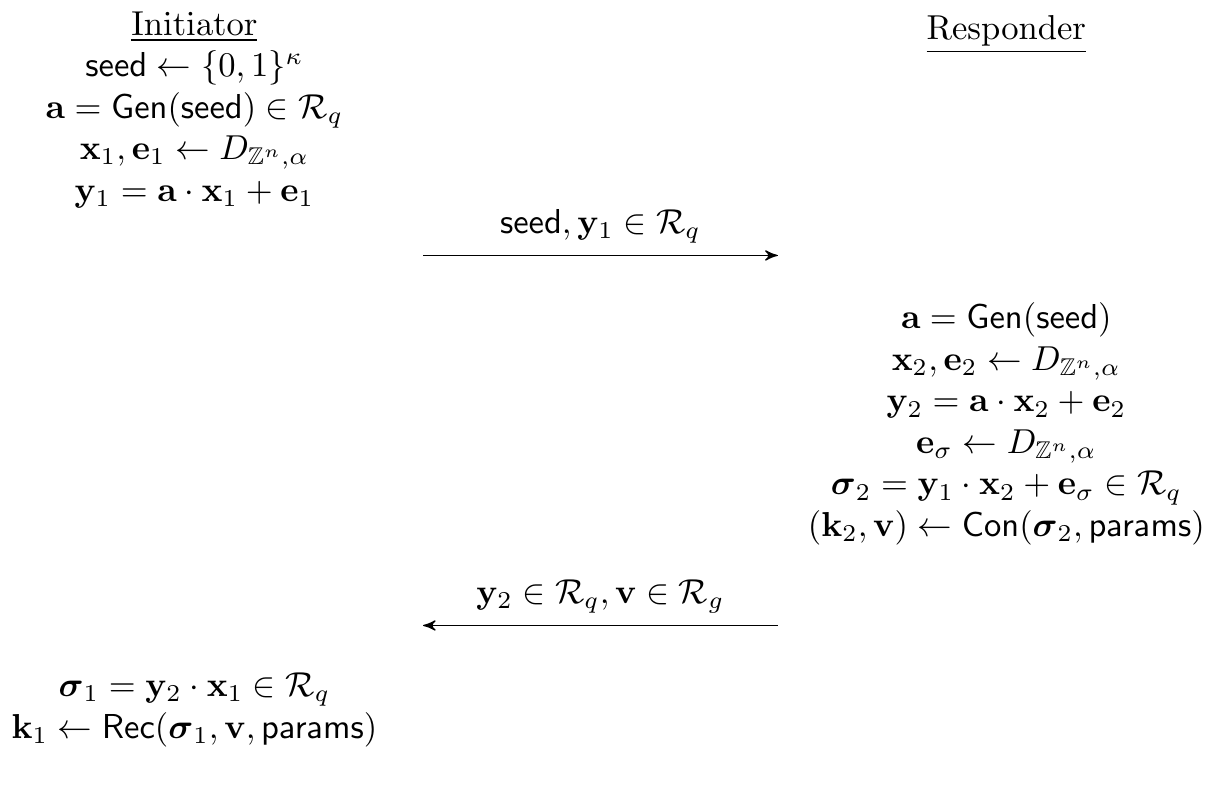}
\caption{RLWE-based key exchange from KC and AKC,
    where {$\mathbf{K}_1, \mathbf{K}_2 \in \mathcal{R}_q$}. The protocol instantiated with OKCN specified in Algorithm \ref{kcs:1} (resp., AKCN in Algorithm \ref{kcs:2}) is referred to as OKCN-RLWE (resp., AKCN-RLWE).}
    \label{rlwe-ke}
\end{figure}

Denote by $(\lambda, n, q, \alpha, KC)$ the system parameters, where $\lambda$ is the security parameter, $q \ge 2$ is a positive prime number, $\alpha$ parameterizes the discrete Gaussian distribution $D_{\mathbb{Z}^n, \alpha}$, $n$ denotes the degree of polynomials in $\mathcal{R}_q$, and  $\mathsf{Gen}$ a PRG generating $\mathbf{a}\in \mathbb{\mathcal{R}}_q$ from a small seed.

Let $KC=(\textsf{params}, \textsf{Con}, \textsf{Rec})$ be a correct and secure KC scheme, where $\mathsf{params}=(q,g,m,d)$.  The KC-based key exchange protocol from RLWE is depicted in \figurename~\ref{rlwe-ke}, where the actual session-key is derived from $\mathbf{K}_1$ and $\mathbf{K}_2$ via some key derivation function $KDF$. There, for presentation simplicity, the \textsf{Con} and \textsf{Rec} functions are applied to polynomials, meaning they are applied to each of the coefficients respectively.
Note that $\boldsymbol{\sigma}_1 = \mathbf{y}_2 \cdot \mathbf{x}_1 =
\mathbf{a} \cdot \mathbf{x}_2 \cdot \mathbf{x}_1 + \mathbf{e}_2 \cdot \mathbf{x}_1$,
$\boldsymbol{\sigma}_2 = \mathbf{y}_1 \cdot \mathbf{x}_2 + \mathbf{e}_{\sigma}
= \mathbf{a} \cdot \mathbf{x}_1 \cdot \mathbf{x}_2 + \mathbf{e}_1 \cdot \mathbf{x}_2 +
\mathbf{e}_{\sigma}$.
As we choose $\mathbf{x}_1, \mathbf{x}_2, \mathbf{e}_1, \mathbf{e}_2, \mathbf{e}_{\sigma}$ according to a
small noise distribution $D_{\mathbb{Z}^n, \alpha}$, the main part of $\boldsymbol{\sigma}_1$ and that of
$\boldsymbol{\sigma}_2$ are the same $\mathbf{a} \cdot \mathbf{x}_1 \cdot \mathbf{x}_2$.
Hence, the corresponding coordinates of $\boldsymbol{\sigma}_1$ and
$\boldsymbol{\sigma}_2$ are close in the sense of $|\cdot|_q$, from which some key consensus can be reached. The error probability depends upon the concrete value of $\alpha$ and the distance parameter $d$. As discussed in Section \ref{sec-lwrke}, a  KC-based key exchange protocol  can be trivially extended to work on any correct and secure AKC scheme. As the bandwidth of RLWE-based KE protocol has already been low, we do not apply the technique of cutting off some least significant bits of each element of the polynomial $\mathbf{y}_2$. 


\textsf{On security analysis.}
 The security definition and proof of the  RLWE-based key exchange protocol can be straightforwardly adapted from those for the KE protocol based on  LWE or LWR. Briefly speaking, from Game $G_0$ to Game $G_1$, we replace $\mathbf{y}_1$ with a
uniformly random polynomial in $\mathcal{R}_q$. From Game $G_1$ to Game $G_2$, we replace
$\mathbf{y}_2$ and $\boldsymbol{\sigma}_2$ with uniformly random polynomials.
Then, from the \emph{security} of the underlying KC or AKC, we get that for any $i$, $1\leq i\leq n$, $\mathbf{K}_2[i]$ and $\mathbf{v}[i]$ are independent, and so the protocol is secure.


\textsf{On  implementations of RLWE-based KE.}
The protocol described in \figurename~\ref{rlwe-ke} works on any hard instantiation of the RLWE problem. But if $n$ is power of $2$, and prime $q$ satisfies $q \bmod 2n = 1$, then  number-theoretic transform (NTT) can be used to speed up polynomial multiplication. 
The performance can be  further improved   by using the Montgomery arithmetic and AVX2 instruction set   \cite{newhope15}, and  by carefully
optimizing   performance-critical routines (in particular,  NTT) in ARM assembly  \cite{cryptoeprint:2016:758}.


For all the implementations considered in this work for RLWE-based KE, we use the same parameters and noise distributions proposed for NewHope in \cite{newhope15}, as described in Table \ref{table-nhcompare}. They  achieve about 281-bit (resp., 255-, 199-)  security against classic (resp., quantum, plausible) attacks. The reader is referred to \cite{newhope15} for details. In particular, the underlying noise distribution is the centered binomial distribution $\Psi_{16}$ (rather than rounded Gaussian distribution with the standard deviation $\sqrt{8}$), which can be rather trivially sampled in hardware and software with much better  protection against timing attacks.


%
%
%

\subsection{RLWE-Based Key Exchange with Negligible Error Rate}
When implemented with the same parameters proposed in \cite{newhope15} for NewHope, as shown in Table \ref{table-nhcompare}  OKCN-RLWE and  AKCN-RLWE  reach 1024 consensus bits, with  a failure probability around $2^{-40}$  that, we suggest, suffices  for most applications of key exchange. In order for reaching a negligible error probability,  particularly for achieving a CPA-secure PKE scheme, we need to further lower the error probability.

A straightforward approach to reducing the error probability is to use the technique of NewHope by  encoding and decoding  the four-dimensional lattice $\tilde{D}_4$. With such an approach, the error probability can be lowered  to about $2^{-61}$, but the shared-key size is reduced from 1024 to 256.  AKCN-RLWE equipped with this approach, referred to as AKCN-4:1, is presented and analyzed in Appendix \ref{app-4:1}. We note that, in comparison with NewHope-simple recently proposed in \cite{newhope-simple}, AKCN-4:1 still has (slight) performance advantage in bandwidth expansion; specifically expanding  256 bits by AKCN-4:1 vs. 1024 bits by NewHope-simple compared to that of NewHope.\footnote{The bandwidth expansion, for both AKCN-4:1 and NewHope-simple, can be further compressed but at the price of losing operation simplicity.}

Decoding the 24-dimensional Leech lattice is also recently considered in \cite{leech}. But  decoding the 4-dimensional  lattice $\tilde{D}_4$ has already been  relatively complicated and computationally less efficient.
Another approach is to employ error correction code (ECC). Unfortunately, in general, the ECC-based approach can be  more inefficient and overburdened  than NewHope's approach.

In this work, we make a key observation  on RLWE-based key exchange,
by proving that the errors in different positions in the shared-key are independent when $n$ is large.
Based upon this observation, we present  a super simple and fast code, referred to as \emph{single-error correction} (SEC) code, to correct at least one bit error.  By equipping OKCN/AKCN with  the SEC code, we achieve the simplest RLWE-based key exchange from both OKCN and AKCN, referred to as OKCN-SEC and AKCN-SEC respectively, with negligible error probability for much longer shared-key size.

\begin{table*}[tph]
\centering
\begin{tabular}{c|rrrr rcr rrr}
\hline
&     $g$    & $d$    & $|\mathbf{K}|$ & bw.(B) & per. & $n_H$ & err. \\ \hline
SKCN-RLWE &     $2^4$ & 2879 & 1024 & 4128 & $2^{-48}$ & -  & $2^{-38}$ \\
SKCN-RLWE &     $2^6$ & 3023 & 1024 & 4384 & $2^{-52}$ & -  & $2^{-42}$ \\
AKCN-RLWE &     $2^4$ & 2687 & 1024 & 4128 & $2^{-42}$ & -  & $2^{-32}$ \\
AKCN-RLWE &     $2^6$ & 2975 & 1024 & 4384 & $2^{-51}$ & -  & $2^{-41}$ \\ \hline
SKCN-SEC      & $2^2$ & 2303 & 765  & 3904 & $2^{-31}$ & 4 & $2^{-48}$ \\
SKCN-SEC      & $2^3$ & 2687 & 765  & 4032 & $2^{-42}$ & 4 & $2^{-70}$ \\
SKCN-SEC      & $2^3$ & 2687 & 837  & 4021 & $2^{-42}$ & 5 & $2^{-69}$ \\
AKCN-SEC      & $2^4$ & 2687 & 765  & 4128 & $2^{-42}$ & 4 & $2^{-70}$ \\
AKCN-SEC      & $2^4$ & 2687 & 837  & 4128 & $2^{-42}$ & 5 & $2^{-69}$ \\ \hline
NewHope       & $2^2$ & -    & 256  & 3872 & $2^{-69}$ & -  & $2^{-61}$ \\
AKCN-4:1-RLWE & $2^2$ & -    & 256  & 3904 & $2^{-69}$ & -  & $2^{-61}$ \\
\hline
\end{tabular}

\vspace{0.5cm}
\caption{Comparisons with NewHope.
    All schemes in the table use the same noise distribution $\Psi_{16}$ used in \cite{newhope15},
    i.e. the sum of $16$ independent centered binomial variables.
    And all schemes in the table share the parameter $(q = 12289, n = 1024, m = 2^1)$.
    $|\mathbf{K}|$ refers to the total binary length of consensus bits.
    bw. (B) refers to the bandwidth in bytes.
    err. refers to failure probability.
    ``$n_H$'' refers to the dimension of SEC code used.
    ``per'' refers to the per bit error rate before applying the SEC code.
    ``err.'' refers to overall error rate.
}\label{table-nhcompare}
\end{table*}

The comparisons among these RLWE-based protocols are summarized in Table \ref{table-nhcompare}.
We note that for RLWE-based KE with negligible error probability, in general, AKC-based protocols are relatively simpler than KC-based protocols.


\subsection{On the Independence of Errors in Different Positions}
Suppose $f(x), g(x)$ are two  polynomials of degree $n$,  whose coefficients are drawn independently from Gaussian.
Let $h(x) = f(x) \cdot g(x) \in \mathbb{R}[x] / (x^n + 1)$.
We show that for every two different integers $0 \le c_1, c_2 < n$,
the joint distribution of $(h[c_1], h[c_2])$ will approach to the two-dimensional Gaussian when $n$ tends to infinity.
Hence, for the basic  construction of RLWE-based key exchange from KC and AKC presented in Figure \ref{rlwe-ke},  it is reasonable to assume that the error probability of any two different positions are independent when $n$ is sufficiently large.

For representation simplicity, for any polynomial $f$,
let $f[i]$ denote the coefficient of $x^i$.
\begin{lemma}
Suppose $f(x), g(x) \in \mathbb{R}[x] / (x^n + 1)$
are two $n$-degree polynomials whose coefficients are drawn independently from $\mathcal{N}(0, \sigma^2)$.
Let $h(x) = f(x) \cdot g(x) \in \mathbb{R}[x] / (x^n + 1)$, where $h(x)$ is represented as an $n$-degree polynomial.
For any two different integers $0 \le c_1, c_2 < n$,
the characteristic function of the two-dimensional random vector $(h[c_1], h[c_2]) \in \mathbb{R}^2$ is
\begin{equation}
\phi_{c_1, c_2}(t_1, t_2) = \mathbb{E}\left[ e^{i\left(t_1 h[c_1] + t_2 h[c_2]\right)} \right] =
\prod_{k = 0}^{n-1}\left(1 + \sigma^4\left(t_1^2 + t_2^2 + 2 t_1 t_2 \cos\left(\pi(c_1 - c_2)\frac{2k+1}{n}\right)\right)\right)^{-\frac{1}{2}}
\label{cf-h}
\end{equation}
\end{lemma}

\begin{proof}
One can observe that
\begin{align*}
t_1 h[c_1] + t_2 h[c_2] &= t_1 \left( \sum_{i + j = c_1} f[i] g[j] - \sum_{i + j = c_1 + n} f[i] g[j] \right)
+ t_2 \left( \sum_{i + j = c_2} f[i]g[j] - \sum_{i + j = c_2 + n} f[i]g[j] \right) \\
&= t_1 \mathbf{f}^T \mathbf{A}_{c_1} \mathbf{g} + t_2 \mathbf{f}^T \mathbf{A}_{c_2} \mathbf{g}
= \mathbf{f}^T (t_1 \mathbf{A}_{c_1} + t_2 \mathbf{A}_{c_2})\mathbf{g}
\end{align*}
Where $\mathbf{f} = (f[0], f[1], \dots, f[n-1])^T$, $\mathbf{g} = (g[0], g[1], \dots, g[n-1])^T$, and
the notations $\mathbf{A}_{c_1}, \mathbf{A}_{c_2}$ are defined by
\begin{align*}
\mathbf{A}_c =
\begin{pmatrix}
& & 1 & & & \\
& \reflectbox{$\ddots$} \\
1 & & & & \\
& & & & & -1 \\
& & & & \reflectbox{$\ddots$} & \\
& & & -1 & &
\end{pmatrix}
\end{align*}
The value $1$ in the first row is in the  $c$-th column.

As $t_1\mathbf{A}_{c_1} + t_2\mathbf{A}_{c_2}$ is symmetric, it can be orthogonally diagonalize as
$\mathbf{P}^T \mathbf{\Lambda} \mathbf{P}$, where $\mathbf{P}$ is orthogonal,
and $\mathbf{\Lambda}$ is diagonal. Hence,
$\phi_{c_1, c_2}(t_1, t_2) = \mathbb{E}[\exp(i(\mathbf{P}\mathbf{f})^T \mathbf{\Lambda} (\mathbf{P} \mathbf{g}))]$.
Since $\mathbf{P}$ is orthogonal, it keeps the normal distribution unchanged.
Hence, $(\mathbf{P}\mathbf{f})^T \mathbf{\Lambda} (\mathbf{P} \mathbf{g})$
equals to the sum of $n$ scaled products of two independent one-dimensional Gaussian.

Suppose $\lambda_1, \lambda_2, \dots, \lambda_n$ are the eigenvalues of
$t_1\mathbf{A}_{c_1} + t_2\mathbf{A}_{c_2}$, and $\phi$ is the characteristic function of
the product of two independent one-dimensional standard Gaussian. Then we have
\begin{equation} \label{eq-final}
\phi_{c_1, c_2}(t_1, t_2) = \prod_{k = 0}^{n - 1} \phi(\sigma^2 \lambda_k)
\end{equation}

From \cite{handbook}, $\phi(t) = (1 + t^2)^{-1/2}$.
For $\lambda_{k}$, we further  observe that
\begin{align*}
(t_1 \mathbf{A}_{c_1} + t_2 \mathbf{A}_{c_2})^2 &=
(t_1^2 + t_2^2) \mathbf{I} + t_1 t_2 (\mathbf{A}_{c_1} \mathbf{A}_{c_2} + \mathbf{A}_{c_2} \mathbf{A}_{c_1}) \\
&= (t_1^2 + t_2^2) \mathbf{I} + t_1 t_2 (\mathbf{G}^{c_2 - c_1} + \mathbf{G}^{c_1 - c_2}),
\end{align*}
 where
\begin{equation*}
\mathbf{G} =
\begin{pmatrix}
& 1 & & & \\
& & 1 & & \\
& & & \ddots \\
& & & & & 1 \\
-1
\end{pmatrix}
\end{equation*}
The characteristic polynomial of $\mathbf{G}$ is $x^n + 1$.
Hence, $\lambda_k$ satisfies
\begin{equation*}
\lambda_k^2 = t_1^2 + t_2^2 + 2 t_1 t_2 \cos\left( \pi (c_1 - c_2)\frac{2k + 1}{n} \right)
\end{equation*}

By taking  this into Equation~\ref{eq-final},  we derive the Equation~\ref{cf-h}.
  \end{proof}

\begin{theorem}\label{independent-theorem}
For any fixed integers $0 \le c_1, c_2 < n$, $c_1 \neq c_2$,
when $n$ tends to infinity, the distribution of
$\left(\frac{h[c_1]}{\sigma^2 \sqrt{n}}, \frac{h[c_2]}{\sigma^2 \sqrt{n}}\right)$
\emph{converges (in distribution)}
to the two-dimensional normal distribution $\mathcal{N}(\mathbf{0}, \mathbf{I}_2)$.
\end{theorem}

\begin{proof}
Let $\phi(t_1, t_2)$ denote the characteristic function of the random vector
$\left(\frac{h[c_1]}{\sigma^2 \sqrt{n}}, \frac{h[c_2]}{\sigma^2 \sqrt{n}}\right)$.
Then, for fixed $t_1, t_2$,
\begin{align}
\ln(\phi(t_1, t_2)) &= -\frac{1}{2} \sum_{k = 0}^{n - 1} \ln\left(1 +
\frac{1}{n}\left(t_1^2 + t_2^2 + 2t_1t_2 \cos\left(\pi(c_1 - c_2)\frac{2k + 1}{n}\right)\right)\right) \\
&= -\frac{1}{2} \sum_{k = 0}^{n - 1} \left[ \frac{1}{n} \left(t_1^2 + t_2^2 + 2t_1t_2 \cos\left(\pi(c_1 - c_2)\frac{2k + 1}{n}\right)\right)
 + r_k \right] \\
&= -\frac{1}{2}\left(t_1^2 + t_2^2\right) - \frac{1}{2}\sum_{k = 0}^{n - 1} r_k, \label{converge}
\end{align}
 where $r_k$ is the Lagrange remainders. So,  $|r_k| \le \lambda_k^4 / 2n^2$.
Since $\lambda_k^2 \le (|t_1| + |t_2|)^2$, we have $|r_k| \le (|t_1| + |t_2|)^4/2n^2$.

When $n$ tends to infinity,
$\phi(t_1, t_2)$ converges pointwise to $\exp(-(t_1^2 + t_2^2)/2)$,
which is the characteristic function of
the two-dimensional normal distribution $\mathcal{N}(\mathbf{0}, \mathbf{I}_2)$.
From L{\'e}vy's convergence theorem,
we derive that the random vector $\left(\frac{h[c_1]}{\sigma^2 \sqrt{n}}, \frac{h[c_2]}{\sigma^2 \sqrt{n}}\right)$
\emph{converges in distribution} to the normal distribution $\mathcal{N}(\mathbf{0}, \mathbf{I}_2)$.
\end{proof}

\subsection{Reducing Error Rate   with Single-Error Correction  Code} \label{sec:err}

Note that, for  the basic protocol construction of RLWE-based key exchange from KC and AKC presented in Figure \ref{rlwe-ke}, it has already achieved per-bit error probability of about $2^{-42}$. The observation here is that, by Theorem \ref{independent-theorem} on the independence of error in different positions when $n$ is large, if we can correct one bit error the error probability will be  greatly lowered  to be negligible.
 Towards this goal, we present an variant of the Hamming code,
 referred to as \emph{single-error correction} (SEC) code, which can correct one-bit error in a very simple and fast way.

 \subsubsection{Single-Error Correction Code}
All the arithmetic operations in this section are over $\mathbb{Z}_2$.
For a positive integer $n_H$, denote $N_H = 2^{n_H}$,
and define the matrix $\mathbf{H}$ as following, {where for any $i$, $1\leq i\leq N_H-1$, the $i$-th column of $\mathbf{H}$ just  corresponds to the binary presentation of $i$.}
\setcounter{MaxMatrixCols}{20}
\begin{equation*}
\mathbf{H}_{n_H \times (N_H-1)}
=\begin{pmatrix}
1 & 0 & 1 & 0 & 1 & 0 & 1 & \cdots & 0 & 1 & 0 & 1 \\
0 & 1 & 1 & 0 & 0 & 1 & 1 & \cdots & 0 & 0 & 1 & 1 \\
0 & 0 & 0 & 1 & 1 & 1 & 1 & \cdots & 1 & 1 & 1 & 1 \\
  &   &   &   &   &   &   & \cdots \\
0 & 0 & 0 & 0 & 0 & 0 & 0 & \cdots & 1 & 1 & 1 & 1
\end{pmatrix}
\end{equation*}

For arbitrary  $\mathbf{x} = (x_1, \dots, x_{N_H -1}) \in \mathbb{Z}_2^{N_H-1}$,
let $\mathbf{p}^T = \mathbf{H} \mathbf{x}^T$. It is easy to check that the $j$-th element of $\mathbf{p}$ is the exclusive-or of all $x_i$'s
satisfying  the $j$-th least  significant bit of $i$ is $1$, where $1\leq j\leq n_H$ and $1\leq i\leq N_H-1$. Specifically, the first element of $\mathbf{p}$ is the exclusive-or of all $x_i$
that the least significant bit of $i$ is $1$,
and the second element of $\mathbf{p}$ is the exclusive-or of all $x_i$
that the second least significant bit of $i$ is $1$, and so on.
Denote $\mathbf{p} = (p_1, p_2, \dots, p_{n_H})$.
We can combine the bits in $\mathbf{p}$ into a binary number $\overline{\mathbf{p}}=2^0 p_1 + 2^1 p_2 + \dots 2^{n_H-1} p_{n_H}$.
The construction of $\mathbf{H}$ directly leads to the following proposition.

\begin{proposition}\label{prop-H}
If $\mathbf{p}^T = \mathbf{H} \mathbf{x}^T$, and the Hamming weight of $\mathbf{x}$ is $1$,
then $\overline{\mathbf{p}}$ is the subscript index  of the only $1$
in $\mathbf{x}$.
\end{proposition}

The single-error correction code $\mathcal{C}$ is defined by
\begin{equation*}
\mathcal{C} = \left\{(x_0, \mathbf{x}, \mathbf{p}) \in
\mathbb{Z}_2 \times \mathbb{Z}_2^{N_H-1}\times \mathbb{Z}_2^{n_H}
\mid x_0 = \oplus_{i = 1}^{N_H-1} x_i,
\mathbf{p}^T = \mathbf{H} \mathbf{x}^T \right\}
\end{equation*}

The encoding algorithm is straightforward and depicted in Algorithm~\ref{ehc-encode}.

We now show that $\mathcal{C}$ can correct one bit error.
Suppose $\mathbf{x}$ is encoded into $\mathbf{c} = (x_0, \mathbf{x}, \mathbf{p})$.
For some reasons, such as the  noise in communication channel, the message $\mathbf{c}$ may be changed into
$\mathbf{c}' = (x_0', \mathbf{x}', \mathbf{p}')$.
We only need to consider the case that at most one bit error occurs.
If $x_0'$ equals to the parity bit of $\mathbf{x}'$, then no error occurs in $x_0$ and $\mathbf{x}$.
Otherwise, there is one bit error in $x_0'$ or $\mathbf{x}'$, but $\mathbf{p}' = \mathbf{p}$ (as we assume there exists at most one bit error that has already occurred in $x_0'$ or $\mathbf{x}'$).
We calculate $\mathbf{p}'' = \mathbf{H} \mathbf{x}'^T \oplus \mathbf{p}'^T$.
In fact, $\mathbf{p}'' = \mathbf{H} \mathbf{x}'^T \oplus \mathbf{p}^T =
\mathbf{H} (\mathbf{x}'^T \oplus \mathbf{x}^T)$.
If the one-bit error occurs in $\mathbf{x}'$,
by Proposition~\ref{prop-H}, $\overline{\mathbf{p}''}$ is the subscript index of the error bit.
If the one-bit error occurs on $x_0'$,
then $\mathbf{x}' = \mathbf{x}$, and $\mathbf{p}'' = \mathbf{H} \mathbf{0} = \mathbf{0}$.
Hence, $\overline{\mathbf{p}''}$ always equals to the subscript index of the error bit.

The decoding algorithm is depicted in Algorithm~\ref{ehc-decode}.
Note that, according to the special form of  $\mathbf{H}$,  the matrix multiplication $\mathbf{H}\mathbf{x}^T$ in both encoding and decoding  can be done with simple bit operations like bit shifts and bitwise exclusive-or (such an implementation is given in Appendix \ref{app-Hxcode}). Moreover,  for AKCN-SEC and OKCN-SEC, the calculations in Lines 2-4 in Algorithm~\ref{ehc-decode} are executed only with probability around $2^{-40}$, so the decoding is extremely fast.

\noindent
\begin{minipage}[H]{0.5\textwidth}
\null
\begin{algorithm}[H]
\caption{$\mathsf{Encode}_{\mathcal{C}}(\mathbf{x} = (x_1,\dots,x_{N_H-1})) $} \label{ehc-encode}
\begin{algorithmic}[1]
\State{$x_0 = \oplus_{i = 1}^{N_H-1} x_i$}
\State{$\mathbf{p}^T = \mathbf{H} \mathbf{x}^T$}
\State{$\mathbf{c} = (x_0, \mathbf{x}, \mathbf{p})$}
\Return{$\mathbf{c}$}
\end{algorithmic}
\end{algorithm}
\end{minipage}
\begin{minipage}[H]{0.5\textwidth}
\null
\begin{algorithm}[H]
\caption{$\mathsf{Decode}_{\mathcal{C}}(x_0, \mathbf{x} = (x_1, \dots, x_{N_H-1}), \mathbf{p})$} \label{ehc-decode}
\begin{algorithmic}[1]
\State{$p = \oplus_{i = 0}^{N_H-1} x_i$}
\If{$p = 1$}
    \State{$i = \overline{\mathbf{H} \mathbf{x}^T} \oplus \overline{\mathbf{p}}$}
    \Comment{bitwise exclusive-or}
    \State{$x_{i} = x_{i} \oplus 1$}
\EndIf
\Return{$\mathbf{x}$}
\end{algorithmic}
\end{algorithm}
\end{minipage}
\vspace{0.3cm}

\subsubsection{AKC and KC with SEC code}

We divide the $n$-bit string $\mathbf{k}_1$ into $\lfloor n/(N_H + n_H) \rfloor$ blocks,
then apply our SEC code in each block.
Note that this approach can also correct more than one bit errors, if  at most one bit error occurs in each block.
\begin{figure}[H]
\centering
\includegraphics{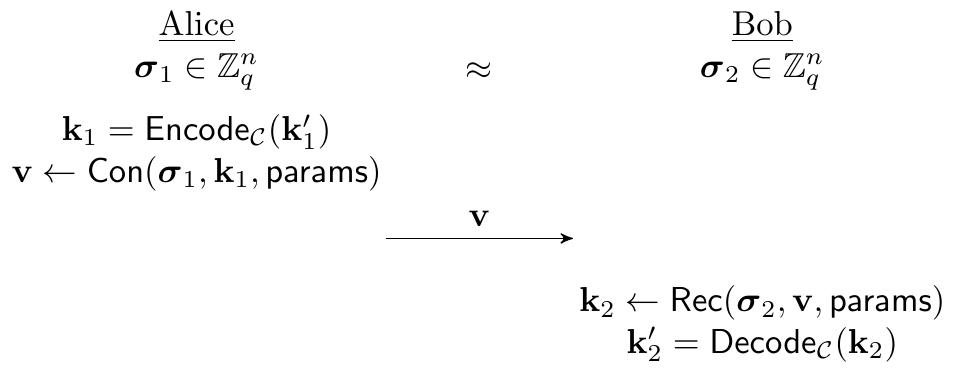}
\caption{Depiction of AKC with SEC code,
where $\mathbf{k}_1, \mathbf{k}_2 \in \mathbb{Z}_2^{N_H + n_H}$, $|\mathbf{k}_1^\prime|=|\mathbf{k}_2^\prime| = N_H-1$.
If the Hamming distance between  $\mathbf{k}_1$ and $\mathbf{k}_2$ is at most  $1$,
then $\mathbf{k}_1' = \mathbf{k}_2'$.}
\label{figure:multi-akc}
\end{figure}

\figurename~\ref{figure:multi-akc} depicts the AKC scheme equipped with the SEC code.
Note that $\mathsf{Encode}_{\mathcal{C}}$ can be calculated off-line.
Suppose the per bit error probability of $\mathbf{k}_1$ and $\mathbf{k}_2$ is $p$,
then under the  assumption that the errors in different positions are independent,
we can estimate that the overall heuristic error probability of $\mathbf{k}_1'$ and $\mathbf{k}_2'$
is no larger than $\lfloor \frac{n}{N_H + n_H} \rfloor C_{N_H + n_H}^{2} p^2$.

\begin{figure}[H]
\centering
\includegraphics{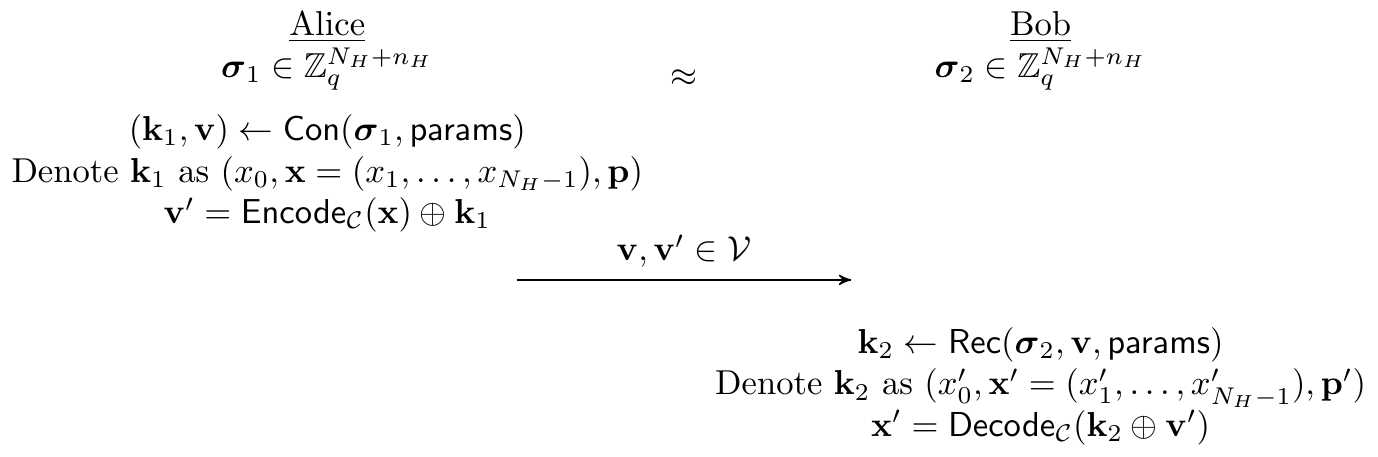}
\caption{Depiction of application of SEC code to KC,
where $\mathbf{k}_1, \mathbf{k}_2 \in \mathbb{Z}_2^{N_H + n_H}$.
If $\mathbf{k}_1$ and $\mathbf{k}_2$ have at most one different bit,
then $\mathbf{x} = \mathbf{x}'$.}
\label{figure:multi-kc-hamming}
\end{figure}

For KC equipped with the SEC code, we propose the algorithm depicted in \figurename~\ref{figure:multi-kc-hamming}.
Note that Alice only needs to send $n_H + 1$ bits of $\mathbf{v}'$, as the second to $N_H$th elements of $\mathbf{v}'$ are all zeros.
Bob calculates $\mathbf{x}' = \mathsf{Decode_{\mathcal{C}}}(\mathbf{k}_2 \oplus \mathbf{v}')$.
In fact, $\mathbf{k}_2 \oplus \mathbf{v}' = \mathsf{Encode}_{\mathcal{C}}(\mathbf{x}) \oplus (\mathbf{k}_1 \oplus \mathbf{k}_2)$. Hence, if the Hamming distance between $\mathbf{k}_1$ and $\mathbf{k}_2$ is $1$,
then $\mathbf{x}' = \mathbf{x}$.
To prove security of the algorithm in \figurename~\ref{figure:multi-kc-hamming},
we need the following theorem.
\begin{theorem} \label{th:decompose}
Let $\mathcal{V} = \mathbb{Z}_2 \times \{\mathbf{0} \in \mathbb{Z}_2^{N_H-1}\} \times \mathbb{Z}_2^{n_H}$,
then $\mathbb{Z}_2^{N_H + n_H} = \mathcal{C} \bigoplus \mathcal{V}$, where $\bigoplus$ denotes direct sum.
\end{theorem}

\begin{proof}
For any $\mathbf{k}_1 = (x_0, \mathbf{x} = (x_1, \dots, x_{N_H-1}), \mathbf{p}) \in \mathbb{Z}_2^{N_H + n_H}$,
let $\mathbf{c} = \mathsf{Encode}_{\mathcal{C}}(\mathbf{x})$
and $\mathbf{v}' = \mathsf{c} \oplus \mathbf{k}_1$.
We have the decomposition $\mathbf{k}_1 = \mathbf{c} \oplus \mathbf{v}'$,
where $\mathbf{c} \in \mathcal{C}$ and $\mathbf{v}' \in \mathcal{V}$.

Next, we prove $\mathcal{V} \cap \mathcal{C} = \mathbf{0}$.
If $\mathbf{k} = (x_0, \mathbf{x}, \mathbf{p}) \in \mathcal{V} \cap \mathcal{C}$,
then $\mathbf{x} = 0$, which implies
$x_0 = 0$ and $\mathbf{p}^T = \mathbf{H} \mathbf{0} = \mathbf{0}$. Hence, $\mathbf{k} = \mathbf{0}$.
\end{proof}

When $\mathbf{k}_1$ is subjected to uniform distribution, then by Theorem~\ref{th:decompose},
after the decomposition $\mathbf{k}_1 = \mathbf{c} \oplus \mathbf{v}'$ where $\mathbf{c} \in \mathcal{C}$
and $\mathbf{v}' \in \mathcal{V}$,
$\mathbf{c}$ and $\mathbf{v}'$ are subjected to uniform distribution in $\mathcal{C}$ and $\mathcal{V}$ respectively. And $\mathbf{c}$ and $\mathbf{v}'$ are independent.
As both $\mathbb{Z}_2^{N_H-1} \rightarrow \mathcal{C}$ and $\mathbf{x} \mapsto \mathsf{Encode}_{\mathcal{C}}(\mathbf{x})$
are  one-to-one correspondence, we derive that $\mathbf{x}$ and $\mathbf{v}'$ are independent, and $\mathbf{x}$ is uniformly distributed.

%

\textbf{On the desirability of OKCN-SEC and AKCN-SEC.} OKCN-SEC and AKCN-SEC are desirable on the following grounds:
\begin{itemize}
\item To our knowledge, OKCN-SEC and AKCN-SEC are the simplest RLWE-based KE protocols \emph{with negligible error probability}, which are better suitable for hardware or software implementations than encoding and decoding the four-dimensional lattice $\tilde{D}_4$. Note that SEC can be implemented with simple bit  operations. Moreover, with probability about $1-2^{-40}$, the decoding  only involves the XOR operations in Line-1 of Algorithm \ref{ehc-encode}, which is extremely simple and fast.

\item AKCN-SEC can be
 directly transformed into a CPA-secure PKE scheme  for encrypting $837$-bit messages with error probability about $2^{-69}$, while AKCN4:1-RLWE and NewHope-simple are for encrypting $256$-bit messages  with error probability about $2^{-61}$.

 \item It is  desirable to have KE protocols that directly share or transport keys of larger size, on the following grounds. Firstly, it is commonly expected that, in the post-quantum era, symmetric-key cryptographic primitives like AES need larger key sizes, in view of the quadratic speedup by  Grover's search algorithm and the possibility of more sophisticated quantum attacks \cite{KM10,KLL15} against symmetric-key cryptography.
     {Secondly, in some more critical application areas  than public  commercial usage, larger key size actually has already been mandated nowadays. Thirdly, when KE being used within TLS, actually a pair of keys $(K_{CS},K_{SC})$, and even more, is needed, where $K_{CS}$ (resp., $K_{SC}$) is for generating messages from client (resp., server) to server (resp., client). $(K_{CS},K_{SC})$ is usually derived via a key-derivation function (KDF)  from the session-key output by the KE protocol. With a KE protocol of larger shared-key size, the KDF can be waived, which simplifies the software or hardware implementations (particularly at the client side) and strengthens security simultaneously.}


  \item  As clarified,   the SEC approach fails only when  there are more than one bit errors in some block, and is versatile in the sense:  the smaller (resp., larger) is  the block size $n_H$, the lower the error probability (resp., bandwidth expansion) will be. The performances for some instantiations of OKCN-SEC and AKCN-SEC are summarized   in Table~\ref{table-nhcompare}.


\end{itemize}

\section{Performance Comparison: Benchmark}
The work \cite{liboqs} introduces the Open Quantum Safe Project.
liboqs is one part of this project.
liboqs provides the interface for adding new key exchange schemes,
benchmark, and an easy way to integrate to OpenSSL.

We fork the liboqs on Github and add our OKCN-LWR-Recommended and OKCN-LWE-Recommended.
Most of the source codes are modified from Frodo-Recommended provided in liboqs.
\noindent
\begin{table}[H]
\centering
\begin{tabular}{lrrrrr}
\hline
& time(us) & stdev & cycle & stdev & bw. (B)\\
\hline
\multicolumn{5}{l}{RLWE BCNS15} \\
\hline
Alice 0   & 1006.497 & 10.909 & 2309859 & 24996 & 4096 \\
Bob       & 1623.353 & 14.045 & 3725552 & 32220 & 4244 \\
Alice 1   & 201.335  & 1.268  & 461976  & 2805 & - \\
\hline
\multicolumn{5}{l}{RLWE NewHope} \\
\hline
Alice 0   & 81.430   & 26.947& 186797  & 61827 & 1824 \\
Bob       & 123.714  & 4.504 & 283841  & 10279 & 2048 \\
Alice 1   & 26.631   & 1.878 & 61048   & 4173  & - \\
\hline
\multicolumn{5}{l}{LWE Frodo recommended} \\
\hline
Alice 0   & 1443.915 & 10.990 & 3313704 & 25236 & 11280 \\
Bob       & 1940.616 & 12.809 & 4453734 & 29439 & 11288 \\
Alice 1   & 170.109  & 3.655  & 390331  & 8317  & - \\
\hline
\multicolumn{5}{l}{LWR OKCN recommended} \\
\hline
Alice 0   & 1161.154 & 11.839 & 2664789 & 27129 & 9968 \\
Bob       & 1722.525 & 12.401 & 3953182 & 28400 & 8224 \\
Alice 1   & 133.984  & 3.980  & 307404  & 9065  & - \\
\hline
\multicolumn{5}{l}{LWE OKCN recommended} \\
\hline
Alice 0   & 1335.453 & 13.460& 3064789 & 30871 & 9968 \\
Bob       & 1753.240 & 14.293& 4023632 & 32851 & 8608 \\
Alice 1   & 146.162  & 3.528 & 335380  & 8035  & - \\
\hline
\end{tabular}
\vspace{0.3cm}
\caption{Benchmark of liboqs integrated with OKCN-LWE-Recommended.
    ``time(us)" refers to mean time that spent on each iteration.
    ``cycle" refers to mean number of cpu cycles.
    ``stdev" refers to population standard deviation of time or cpu cycles.
    ``bw. (B)" refers to bandwidth, counted in bytes.
}
\label{benchmark}
\end{table}

We run benchmark of liboqs on Ubuntu Linux 16.04, GCC 5.4.0,
Intel Core i7-4712MQ 2.30GHz, with hyperthreading and TurboBoost disabled,
and the CPU frequency fixed to 2.30GHz (by following the instructions on \url{http://bench.cr.yp.to/supercop.html}).
The benchmark result (Table~\ref{benchmark}) shows that OKCN-LWR-Recommended and OKCN-LWE-Recommended are  faster than Frodo, and use smaller bandwidth.

We modify the source code of NewHope.
We only change the reconciliation part of NewHope to our OKCN/AKCN,
and keep other parts (e.g. NTT) unchanged.
The benchmark of our RLWE-AKCN-SEC and RLWE-OKCN-SEC are shown in Table~\ref{benchmark-SEC}.
The environment is the same as above.
We can see that the simplicity of our reconciliation mechanisms result in significant improvement of this part in the key exchange scheme.
\noindent
\begin{table}[H]
\centering
\begin{tabular}{c|l|rrrr}
\hline
                      &         & AKCN-SEC & OKCN-SEC & NewHope \\ \hline
\multirow{2}{*}{con}  & median  & 8032 & 12020 & 15276 \\
                      & average & 8169 & 12093 & 15448 \\ \hline
\multirow{2}{*}{rec}  & median  & 7420 & 7356  & 10548 \\
                      & average & 7625 & 7369  & 10498 \\ \hline
\end{tabular}
\vspace{0.3cm}
\caption{Benchmark of AKCN/OKCN-SEC and NewHope. Numbers in the table are median and average of cpu cycles.}
\label{benchmark-SEC}
\end{table}

\section{More Applications and Future Works}

The composition of a secure KE protocol (\emph{with negligible error probability}) and a CPA-secure symmetric-key encryption scheme (SKE) yields a CPA-secure PKE scheme. And any CPA-secure PKE can be transformed into a CCA-secure one via the FO-transformation \cite{FO99b,FO99a,peikert2014lattice,TU15} in the quantum  random oracle model. If we view $2^{-60}$ to be negligible, then  OKCN-LWE and AKCN-LWE (on the same parameters of Frodo),  AKCN-4:1-RLWE (on the same parameters of NewHope), and OKCN-SEC and AKCN-SEC can be used to build CPA-secure PKE schemes. Moreover, AKCN-LWE,  AKCN-4:1-RLWE, and AKCN-SEC can be used directly  for CPA-secure PKE scheme (without composing SKE).

One particularly important application of public-key cryptography is key transport (i.e., public-key encryption of a random symmetric key), which is in particular demanded by the Tor project \cite{Tor16} and NIST \cite{NIST16}. We note that our AKC-based KE protocols can just be used for key transport.

 Any secure KE protocol can be transformed, in a black-box way, into an authenticated key exchange (AKE) protocol by additionally using a secure signature scheme via the SIGMA paradigm  \cite{K03}. SIGMA   is just the basis of authentication mechanism for   the secure transport protocol  TLS in the client/server setting. Recently, the next generation of TLS, specifically TLS1.3, is now under development \cite{TLS}; And developing post-quantum secure TLS  protocol is now receiving more and more  efforts or attention both from research community and from standardization agencies. Compared to the currently deployed TLS1.2 standard, one salient change (among others) made in TLS1.3  is that the server now plays the role of  the  responder.
 The heavy workload in the server, particularly at peak time,   is one of the major sources that causes slower server responsiveness or causes the server an  easy target of more and more vicious DDoS attacks.   We suggest that the predicament faced by the server can be  mitigated with  AKC-based KE protocols like AKCN-LWE, AKCN-4:1-RLWE,  and AKCN-SEC, by taking advantage of  the session-key predetermination and online/offline parallel computability enjoyed by them. {The integration of the developed protocols into TLS is left to future works.}

 Finally, as a fundamental building tool for lattice-based cryptography, we suggest OKCN, AKCN, and the various KE protocols based on them, are of independent value. They   may possibly  find more applications in  more advanced cryptographic primitives (e.g., oblivious transfer,  IBE, ABE, FHE) from LWR, LWE and RLWE, by simplifying their design and analysis with versatile performance optimization or balance.

 \section{ZarZar: RLWE-Based Key Exchange from KC and AKC}\label{zarzar}
Previously, \cite{newhope15} proposes practical RLWE key exchange NewHope,
whose reconciliation mechanism uses lattice $D_2$.
\cite{newhope-simple} proposes a variant NewHope-Simple,
which simplifies the reconciliation mechanism of NewHope.
However, the slowest part of NewHope and NewHope-simple is the $1024$-dimension Number Theoretic Transform (NTT).
Hence, the simplification of the reconciliation mechanism does not bring much improvement of the overall efficiency.
\cite{newhope15} also provides JarJar, which is a key exchange based on $512$-dimension RLWE,
and its reconciliation mechanism is based on $D_2$.
However, JarJar can only provide $118$-bit security on best known quantum attack from \cite{newhope15}.

In this section, we propose a lightweight RLWE key exchange named ZarZar.
ZarZar uses lattice $E_8$, which is in higher dimension and thus denser than $D_2$.
This allows us to choose smaller dimension of RLWE with larger noise.
NTT of 512-dimension in ZarZar costs nearly half of that of NewHope.
Although the dimension of $E_8$ is higher than $D_2$,
based on the structure of $E_8$, we propose fast encoding and decoding algorithms.
To improve the efficiency of sampling noise of large variance,
we propose a fast sample method.

In our parameter set, $n = 512$.
We divide the $512$ coefficients of the polynomial $\boldsymbol{\sigma}_1$ and $\boldsymbol{\sigma}_2$ into $64$ groups,
where each group is composed by $8$ coefficients.
In specific, denote $R = \mathbb{Z}[x] / (x^8 + 1), R_q = R/qR, K = \mathbb{Q}[x] / (x^8 + 1)$ and
$K_{\mathbb{R}} = K \otimes \mathbb{R} \simeq \mathbb{R}[x] / (x^8 + 1)$.
Then the polynomial $\boldsymbol{\sigma}_1$ can be represented as
$\boldsymbol{\sigma}_1(x) = \sigma_0(x^{64}) + \sigma_1(x^{64}) x + \dots + \sigma_{63}(x^{64}) x^{63}$,
where $\sigma_i(x) \in R_q$ for $i = 0, 1, \dots 63$. $\boldsymbol{\sigma}_2$ can be divided in the same way.
Then we only need to construct the reconciliation mechanism for each $\sigma_i(x)$, and finally combine the keys together.
To do this, we first need to introduce the lattice $E_8$ and its encoding and decoding.

\subsection{$E_8$ and its Encoding and Decoding}
We construct lattice $E_8$ from the Extended Hamming Code in dimension $8$, $H_8$.
$H_8$ refers to the 4-dimension linear subspace of 8-dimension linear space $\mathbb{Z}_2^8$.
\begin{equation*}
H_8 = \{ \mathbf{c} \in \mathbb{Z}_2^8 \mid \mathbf{c} = \mathbf{z} \mathbf{H} \bmod 2, \mathbf{z} \in \mathbb{Z}^4\}
\end{equation*}
where
\begin{equation*}
    \mathbf{H} = \begin{pmatrix}
        1 & 1 & 1 & 1 & 0 & 0 & 0 & 0 \\
        0 & 0 & 1 & 1 & 1 & 1 & 0 & 0 \\
        0 & 0 & 0 & 0 & 1 & 1 & 1 & 1 \\
        0 & 1 & 0 & 1 & 0 & 1 & 0 & 1
    \end{pmatrix}
\end{equation*}

The encoding algorithm is straightforward:
given a 4-bit string $\mathbf{k}_1$,
calculate $\mathbf{k}_1 \mathbf{H}$.
This operation can be done efficiently by bitwise operations.
We combine this encoding with AKCN. The complete algorithm is shown in Algorithm~\ref{zarzar-con}.

\noindent
\begin{algorithm}[H]
\caption{Con with encoding in $E_8$}\label{zarzar-con}
\begin{algorithmic}[1]
\Procedure {$\mathsf{Con}$}{$\boldsymbol{\sigma}_1 \in \mathbb{Z}_q^8, \mathbf{k}_1 \in \mathbb{Z}_2^4, \textsf{params}$}
\State{$\mathbf{v} = \left\lfloor \frac{g}{q}\left(\boldsymbol{\sigma}_1 + \frac{q-1}{2}(\mathbf{k}_1 \mathbf{H} \bmod 2)\right) \right\rceil \bmod g$}
    \Return{$\mathbf{v}$}
\EndProcedure
\end{algorithmic}
\end{algorithm}

The decoding algorithm finds the solution of the closest vector problem (CVP) for the lattice $E_8$.
For any given $\mathbf{x} \in \mathbb{R}^8$, CVP asks which lattice point in $E_8$ is closest to $\mathbf{x}$.
Based on the structure of $E_8$, we propose an efficient decoding algorithm.
Let $C = \{(x_1, x_1, x_2, x_2, x_3, x_3, x_4, x_4) \in \mathbb{Z}_2^8 \mid x_1 + x_2 + x_3 + x_4 = 0 \bmod 2\}$.
In fact, $C$ is spanned by the up most three rows of $\mathbf{H}$.
Hence, $E_8 = C \cup (C + \mathbf{c})$, where $\mathbf{c} = (0, 1, 0, 1, 0, 1, 0, 1)$ is the last row of $\mathbf{H}$.
For a given $\mathbf{x} \in \mathbb{R}^8$,
to solve CVP of $\mathbf{x}$ in $E_8$, we solve CVP of $\mathbf{x}$ and $\mathbf{x} - \mathbf{c}$ in $C$,
and then choose the one that has smaller distance.

\noindent
\begin{algorithm}[H]
\caption{Rec with decoding in $E_8$}\label{zarzar-rec}
\begin{algorithmic}[1]
\Procedure {$\mathsf{Rec}$}{$\boldsymbol{\sigma}_2 \in \mathbb{Z}_q^8, \mathbf{v} \in \mathbb{Z}_g^8, \textsf{params}$}
    \State{$\mathbf{k}_2 = \mathsf{Decode_{E_8}}\left(\left\lfloor \frac{q}{g} \mathbf{v} \right\rceil - \boldsymbol{\sigma}_2\right)$}
    \Return{$\mathbf{k}_2$}
\EndProcedure
\end{algorithmic}
\end{algorithm}

\noindent
\begin{algorithm}
\caption{Decoding in $E_8$ and $C$}\label{zarzar-decode}
\begin{algorithmic}[1]
\Procedure {$\mathsf{Decode}_{E_8}$}{$\mathbf{x} \in \mathbb{Z}_q^8$}
    \For{$i = 0 \dots 7$}
    \State{$\mathsf{cost}_{i, 0} = |x_i|_q^2$}
    \State{$\mathsf{cost}_{i, 1} = |x_i + \frac{q-1}{2}|_q^2$}
    \EndFor
    \State{$(\mathbf{k}^{00}, \mathsf{TotalCost}^{00})\gets \mathsf{Decode}_C^{00}(\mathsf{cost}_{i \in 0\dots7, b\in\{0, 1\}})$}
    \State{$(\mathbf{k}^{01}, \mathsf{TotalCost}^{01})\gets \mathsf{Decode}_C^{01}(\mathsf{cost}_{i \in 0\dots7, b\in\{0, 1\}})$}
    \If {$\mathsf{TotalCost}^{00} < \mathsf{TotalCost}^{01}$}
        \State{$b = 0$}
    \Else
        \State{$b = 1$}
    \EndIf
    \State{$(k_0, k_1, k_2, k_3) \gets \mathbf{k}^{0b}$}
    \State{$\mathbf{k}_2 = (k_0, k_1 \oplus k_0, k_3, b)$}
    \Return{$\mathbf{k}_2$}
\EndProcedure
\Procedure {$\mathsf{Decode}_C^{b_0 b_1}$}{$\mathsf{cost}_{i \in 0\dots7, b\in\{0, 1\}} \in \mathbb{Z}^{8 \times 2}$}
    \State{$min_d = +\infty$}
    \State{$min_i = 0$}
    \State{$\mathsf{TotalCost} = 0$}
    \For{$j = 0 \dots 3$}
        \State{$c_0 \gets \mathsf{cost}_{2j, b_0} + \mathsf{cost}_{2j+1, b_1}$}
        \State{$c_1 \gets \mathsf{cost}_{2j, 1-b_0} + \mathsf{cost}_{2j+1, 1-b_1}$}
        \If {$c_0 < c_1$}
            \State{$k_i \gets 0$}
        \Else
            \State{$k_i \gets 1$}
        \EndIf
        \State{$\mathsf{TotalCost} \gets \mathsf{TotalCost} + c_{k_i}$}
        \If {$c_{1-k_i} - c_{k_i} < min_d$}
            \State{$min_d \gets c_{1-k_i} - c_{k_i}$}
            \State{$min_i \gets i$}
        \EndIf
    \EndFor
    \If {$k_0 + k_1 + k_2 + k_3 \bmod 2 = 1$}
        \State{$k_{min_i} \gets 1 - k_{min_i}$}
        \State{$\mathsf{TotalCost} \gets \mathsf{TotalCost} + min_d$}
    \EndIf
    \State{$\mathbf{k} = (k_0, k_1, k_2, k_3)$}
    \Return {$(\mathbf{k}, \mathsf{TotalCost})$}
\EndProcedure
\end{algorithmic}
\end{algorithm}

Then we consider how to solve CVP in $C$.
For an $\mathbf{x} \in \mathbb{R}^8$, we choose $(x_1, x_2, x_3, x_4) \in \mathbb{Z}_2^4$,
such that $(x_1, x_1, x_2, x_2, x_3, x_3, x_4, x_4)$ is closest to $\mathbf{x}$.
However, $x_1 + x_2 + x_3 + x_4 \bmod 2$ may equal to $1$.
In such cases, we choose the 4-bit string $(x_1', x_2', x_3', x_4')$ such that
$(x_1', x_1', x_2', x_2', x_3', x_3', x_4', x_4')$ is secondly closest to $\mathbf{x}$.
Note that $(x_1', x_2', x_3', x_4')$ has at most one-bit difference from $(x_1, x_2, x_3, x_4)$.
The detailed algorithm is depicted in Algorithm~\ref{zarzar-decode}. Considering potential timing attack,
all the ``if'' conditional statements can be implemented by constant time bitwise operations.
In practice, $\mathsf{Decode}_C^{00}$ and $\mathsf{Decode}_C^{01}$ are implemented as two subroutines.

For algorithm~\ref{zarzar-decode},
in $\mathsf{Decode}_{E_8}$, we calculate $\mathsf{cost}_{i, b}$, where $i = 0, 1, \dots, 7, b \in \{0, 1\}$,
which refer to the contribution to the total 2-norm when $x_i = b$.
$\mathsf{Decode}_C^{00}$ solves the CVP in lattice $C$, and $\mathsf{Decode}_C^{01}$ solves the CVP in lattice $C + \mathbf{c}$.
Then we choose the one that has smaller distance.
$\mathsf{Decode}_C^{b_0 b_1}$ calculates the $k_i, i = 0, 1, 2, 3$ such that
$\frac{q-1}{2}(k_0 \oplus b_0, k_0 \oplus b_1, k_1 \oplus b_0, k_1 \oplus b_1,
k_2 \oplus b_0, k_2 \oplus b_1, k_3 \oplus b_0, k_3 \oplus b_1)$ is closest to $\mathbf{x}$.
We use $min_d$ and $min_i$ to find the second closest vector.
Finally, we check the parity to decide which one should be returned.

The following theorem gives a condition of success of the encoding and decoding algorithm
in Algorithm~\ref{zarzar-con} and Algorithm~\ref{zarzar-rec}.
For simplicity, for any $\boldsymbol{\sigma} = (x_0, x_1, \dots, x_7) \in \mathbb{Z}_q^8$,
we define $\|\boldsymbol{\sigma}\|_{q, 2}^2 = \sum_{i = 0}^{7} |x_i|_q^2$.

\begin{theorem}\label{th-correct}
   If $\|\boldsymbol{\sigma}_1 - \boldsymbol{\sigma}_2\|_{q, 2} \le (q - 1)/2 - \sqrt{2}\left(\frac{q}{g} + 1\right)$,
   then $\mathbf{k}_1$ and $\mathbf{k}_2$ calculated by $\mathsf{Con}$
   and $\mathsf{Rec}$ are equal.
\end{theorem}

\begin{proof}

The minimal Hamming distance of the Extended Hamming code $H_8$ is $4$.
Hence, the minimal distance in the lattice we used is $\frac{1}{2}\sqrt{\left(\frac{q-1}{2}\right)^2 \times 4} = (q-1)/2$.

We can find $\boldsymbol{\epsilon}, \boldsymbol{\epsilon}_1 \in [-1/2, 1/2]^8, \boldsymbol{\theta} \in \mathbb{Z}^8$ such that
\begin{align*}
\left\lfloor \frac{q}{g} \mathbf{v} \right\rceil - \boldsymbol{\sigma}_2 &= \frac{q}{g} \mathbf{v} + \boldsymbol{\epsilon} - \boldsymbol{\sigma}_2
= \frac{q}{g} \left( \frac{g}{q}\left(\boldsymbol{\sigma}_1 + \frac{q-1}{2} \mathbf{k}_1 \mathbf{H}\right) + \boldsymbol{\epsilon} + \boldsymbol{\theta} g \right) +
\boldsymbol{\epsilon}_1 - \boldsymbol{\sigma}_2 \\
&= (\boldsymbol{\sigma}_1 - \boldsymbol{\sigma}_2) + \frac{q-1}{2} \mathbf{k}_1 \mathbf{H} + \frac{q}{g} \boldsymbol{\epsilon} +
\boldsymbol{\epsilon}_1 + \boldsymbol{\theta} q
\end{align*}
Hence, the bias from $\frac{q-1}{2} \mathbf{k}_1 \mathbf{H}$ is no larger than
$\|\boldsymbol{\sigma}_1 - \boldsymbol{\sigma}_2\|_{q, 2} + \frac{q}{g} \|\boldsymbol{\epsilon}\| + \sqrt{2}
\le \|\boldsymbol{\sigma}_1 - \boldsymbol{\sigma}_2\|_{q, 2} + \sqrt{2}\left(\frac{q}{g} + 1\right)$.
If this value is less than the minimal distance $(q-1)/2$, the decoding will be correct,
which implies $\mathbf{k}_1 = \mathbf{k}_2$.

\end{proof}

\subsection{Parameters and Implementation}

We shall explain why our key exchange based on $E_8$ outperforms JarJar,
which uses $D_2$ in \cite{newhope15}.
$E_8$ is densitest in 8-dimension, and it is denser than $D_2$.
Specifically, when $E_8$ is embeded into $\mathbb{Z}_q^8$,
the minimal distance of $\frac{q}{2} E_8$ is $\frac{q}{2}$,
while the minimal distance of $q D_2$ is $\frac{\sqrt{2}}{4} q$.
Both $E_8$ and $D_2$ can extract $256$-bit key from polynomial whose degree is $512$.
As we shall see in \figurename~\ref{fig:error-clt},
the logarithm of probability is almost propotional to the square of the distance
between $\boldsymbol{\sigma}_1$ and $\boldsymbol{\sigma}_2$.
Hence, the smaller minimal distance allows us to choose Gaussian of larger deviation in secrets and noises,
which garantees at least $128$-bit post-quantum safty of the protocol.

Encoding and decoding in $E_8$ may be more slower than $D_2$,
but we half the time of numerical theoretic transform (NTT) compared to NewHope by setting $n = 512$.
Note that the efficiency of NTT dominates the overall efficiency.
Hence, the trade off is worthwhile.

\subsubsection{Parameters}

We choose the parameter $(q = 12289, n = 512, \sigma^2 = 22)$.
See Table~\ref{table:params-zarzar} for the security estimation.
The attack time estimation of NewHope seems much larger than our ZarZar.
However, as NewHope only reach 256-bit shared key.
The post-quantum attacker may attack its symmetric encryption in 128-bit time.
In addition, the security level is estimated in a very conservative way.
The number $129$ in column $Q$ does not mean best known attack can achieve $2^{129}$ time.
\begin{table}[H]
\centering
\begin{tabular}{rrrrrr}
\hline
          & $m$ & $b$ & C   & Q   & P \\ \hline
Primal    & 646 & 491 & 143 & 130 & 101 \\
Dual      & 663 & 489 & 143 & 129 & 101 \\ \hline
\end{tabular}
\caption{Security estimation for ZarZar.} \label{table:params-zarzar}
\end{table}

\subsubsection{Noise distribution}
NewHope uses the centered binary distribution $\Psi_{16}$ as the secrete and noise distribution.
However, the deviation in our parameter set $(\sigma = \sqrt{22})$ is much larger than NewHope,
this method requires too many random bits.
Note that the generation of random bits costs a lot of time.
Frodo uses a table to generate a discrete distribution that is very close to the rounded Gaussian.
However, in our parameter set, the table will be too large to sample efficiently.
Hence, we propose the distribution $B^{a, b}$, where $a$ and $b$ are two integers.

\begin{algorithm}
    \caption{Sample $r$ from $B^{a, b}$} \label{algo-sample}
\begin{algorithmic}[1]
    \State{$r \gets \sum_{i = 1}^{a}\mathsf{getOneRandomBit}() + 2*\sum_{i = 1}^{b}\mathsf{getOneRandomBit}() - \left(\frac{a}{2} + b\right)$}
\end{algorithmic}
\end{algorithm}

The variation of $r$ in Algorithm~\ref{algo-sample} is $\frac{a}{4} + b$, and the expect value of $r$ is $0$.
By the \emph{central limit theorem}, the distribution of $r$ is close to a discrete Gaussian.
In our implementation, we choose $a = 24, b = 16$, and the summation of the random bits are calculated by fast bit counting.

\begin{table}[H]
\centering
\begin{tabular}{rrrrrrr}
\hline
        & $n$  & $a$ & $P$ & $\sigma^2$ & $R_a(P || Q)$ & $R_a(P || Q)^{5n}$ \\ \hline
ZarZar  & 512  & 30  & $B^{24,16}$ & 22 & 1.0016 & 65.71 \\
NewHope & 1024 & 9   & $\Psi_{16}$ & 8  & 1.00063 & 26 \\ \hline
\end{tabular}
\caption{Comparison of Renyi divergence. $Q$ refers to the rounded Gaussian of variance $\sigma^2$.} \label{table:b24-16}
\end{table}

The closeness of $B^{24, 16}$ and the rounded Gaussian of variance $22$ is measured in Table~\ref{table:b24-16}.
Recall that the Renyi divergence increase as $a$ increase.
Hence, $B^{24, 16}$ and rounded Gaussian of variance $22$ are more close compared to $\Psi_{16}$
and rounded Gaussian of variance $8$.
We use a larger $a$ than NewHope so that the potential security decline can be smaller,
although no attacks known make use of the information of different noise distributions.

\subsubsection{Benchmark}

\noindent
\begin{table}[H]
\centering
\begin{tabular}{c|r|rrr}
\hline
                      &         & ZarZar & NewHope \\ \hline
\multirow{2}{*}{NTT}  & median  & 26864 & 56056 \\
                      & average & 27449 & 56255 \\ \hline
\multirow{2}{*}{NTT$^{-1}$}   & median  & 28368 & 59356 \\
                            & average & 28447 & 59413 \\ \hline
\multirow{2}{*}{sample noise}      & median  & 20876 & 33100 \\
                            & average & 20909 & 33112 \\ \hline
\multirow{2}{*}{con/HelpRec}      & median  & 5760 & 15180 \\
                            & average & 5786 & 15165 \\ \hline
\multirow{2}{*}{rec}        & median  & 10920 & 10516 \\
                            & average & 10990 & 10515 \\ \hline
\multirow{2}{*}{Alice0}    & median  & 133646 & 249476 \\
                            & average & 135550 & 250264 \\
bandwidth (B)               & -       & 928 & 1824 \\ \hline
\multirow{2}{*}{Bob0}    & median  & 196478 & 386248 \\
                            & average & 197840 & 387104 \\
bandwidth (B)               & -       & 1280 & 2048 \\ \hline
\multirow{2}{*}{Alice1}    & median  & 48248 & 84880 \\
                            & average & 48332 & 84940 \\ \hline
\end{tabular}
\vspace{0.3cm}
\caption{Benchmark of ZarZar and NewHope. Numbers in the table are median and average of cpu cycles.}
\label{benchmark-zarzar}
\end{table}
We run benchmark of ZarZar and NewHope on Ubuntu Linux 16.04, GCC 5.4.0, Intel Core i7-4712MQ
2.30GHz, with hyperthreading and TurboBoost disabled. The benchmark
result (Table~\ref{benchmark-zarzar}) shows that ZarZar essentially  outperforms NewHope.
%
%

\subsection{Error Rate Analysis}

In this section, we propose a delicate analysis of the error probability.
From the depiction of the protocol, we know that $\boldsymbol{\sigma}_1 - \boldsymbol{\sigma}_2 =
\mathbf{y}_2 \mathbf{x}_1 - (\mathbf{y}_1 \mathbf{x}_2 + \mathbf{e}_{\sigma}) =
\mathbf{e}_2 \mathbf{x}_1 - \mathbf{e}_1 \mathbf{x}_2 - \mathbf{e}_{\sigma}$.

\begin{claim}\label{claim-bnd}
If $e_{\epsilon}$ is a real random variable subjected to Gaussian of variance $\sigma^2$, then
$\Pr[|e_{\epsilon}| > 10\sigma] < 2^{-70}$.
\end{claim}
We can bound the last term $\mathbf{e}_{\sigma}$ by the Claim~\ref{claim-bnd}.
For other terms, we divide the $512$ coefficients of the polynomial $\boldsymbol{\sigma}_1$ and $\boldsymbol{\sigma}_2$ into $64$ groups by the same way in Section ~\ref{zarzar}.
From Theorem~\ref{th-correct}, we only need to calculate the distribution of $\|(\mathbf{e}_2 \mathbf{x}_1 - \mathbf{e}_1 \mathbf{x}_2)_k\|_{q, 2}$
for every $k$.

We shall make use of the symmetry of multidimensional continual Gaussian to derive a simple form of
the distribution of $\|(\mathbf{e}_2 \mathbf{x}_1 - \mathbf{e}_1 \mathbf{x}_2)_k\|_{q, 2}^2$, and then  calculate the distribution numerically by computer programs.
One may assume $n$ is large enough to use the \emph{central limit theorem}.
The comparison of our numerical method result and the error probability obatined by
\emph{central limit theorem} is shown in \figurename~\ref{fig:error-clt}.

\begin{figure}[H]
\centering
\includegraphics[scale=0.7]{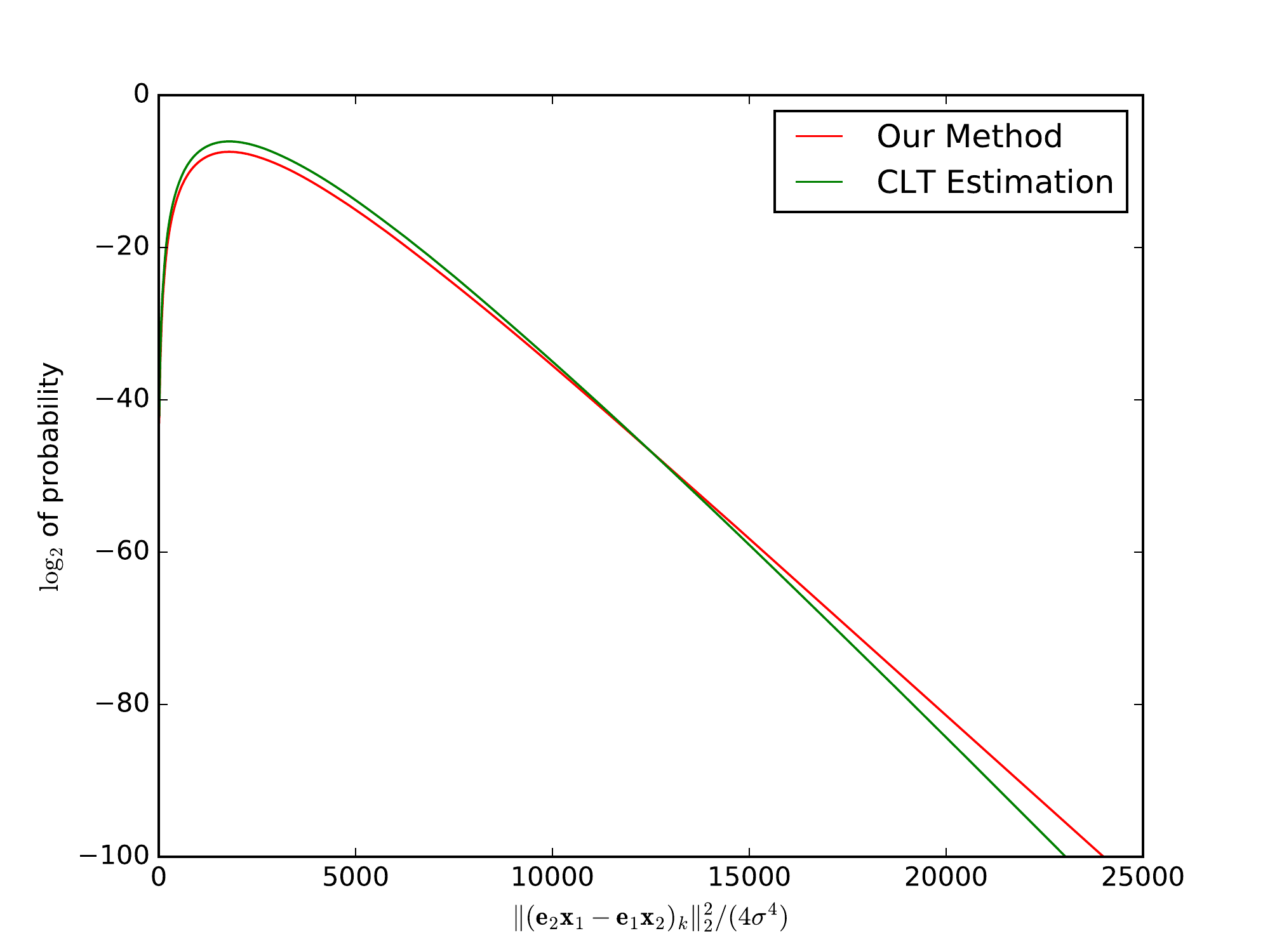}
\caption{Comparison of the probability estimated by our method and
the error probability obtained by \emph{central limit theorem}.} \label{fig:error-clt}
\end{figure}

\begin{theorem}\label{th-partition}
For two polynomials $f(x), g(x) \in \mathbb{Z}_q[x] / (x^{512} + 1)$,
we can represent $f(x)$ as $f(x) = f_0(x^{64}) + f_1(x^{64}) x + \dots + f_{63}(x^{64}) x^{63}$,
where $f_i(x) (i = 0, 1, \dots, 63)$ are in $R$. $g(x)$ can be represented in the same way as
$g(x) = g_0(x^{64}) + g_1(x^{64}) x + \dots + g_{63}(x^{64}) x^{63}$. Let $h(x) = f(x) g(x)$,
and $h(x)$ is represented as $h(x) = h_0(x^{64}) + h_1(x^{64}) x + \dots + h_{63}(x^{64}) x^{63}$.
Then we have
\begin{equation}
    h_k(x) = \sum_{i + j = k} f_i(x) g_j(x) + \sum_{i + j = k + 64} f_i(x) g_j(x) x
\end{equation}
\end{theorem}

\begin{proof}
    We have $f(x) g(x) = \sum_{i, j} f_i(x^{64}) g_j(x^{64}) x^{i + j}$.
    For each $(i, j)$ such that $i + j < 64$,
    the term $f_i(x^{64}) g_j(x^{64})$ is added into $h_{i + j}(x^{64})$.
    For each $(i, j)$ such that $i + j \ge 64$, we have $i + j - 64 < 64$,
    and $f_i(x^{64}) g_j(x^{64}) x^{i + j} = f_i(x^{64}) g_j(x^{64}) x^{i + j - 64} x^{64}$.
    Hence, the term $f_i(x^{64}) g_j(x^{64})x^{64}$ is added into $h_{i + j - 64}(x^{64})$.
    Now we derive $h_k(x^{64}) = \sum_{i + j = k} f_i(x^{64}) g_j(x^{64}) + \sum_{i + j = k + 64} f_i(x^{64}) g_j(x^{64}) x^{64}$.
\end{proof}

By applying Theorem~\ref{th-partition} to the polynomial $\mathbf{e}_2 \mathbf{x}_1 - \mathbf{x}_2\mathbf{e}_1$,
we know that for every fixed $k$, there are $a_i(x), b_i(x) \in R, i = 1, 2, \dots, 128$ whose coefficients are subjected to rounded Gaussian,
such that $(\mathbf{e}_2 \mathbf{x}_1 - \mathbf{x}_2\mathbf{e}_1)_k = \sum_{i = 1}^{128} a_i(x) b_i(x)$.

\begin{theorem}\label{th-simple}
If $a_i(x), b_i(x), i = 1, 2, \dots, 128$ are polynomials in $K_{\mathbb{R}}$
whose coefficients are drawn from continues Gaussian of variance $\sigma^2$ independently,
then the following two distributions are identical.
\begin{itemize}
\item $\left\|\sum_{i = 1}^{128} a_i(x) b_i(x)\right\|_2^2$,
    where the 2-norm here refers to treat the polynomials in $K_{\mathbb{R}}$ as a vector in $\mathbb{R}^8$,
    and take the 2-norm in $\mathbb{R}^8$.

\item $4\sigma^4\sum_{i = 1}^{4} x_i y_i$,
    where $x_i$ are drawn from chi-square distribution $\chi^2(256)$ independently,
    and $y_i$ are drawn from $\chi^2(2)$ independently.
\end{itemize}
\end{theorem}

To prove this theorem, we introduce the \emph{canonical embedding},
which is a commonly used tool in ideal lattice analysis.
The \emph{canonical embedding} $\sigma$ is a map from $K_{\mathbb{R}}$ to
$\mathbb{C}^8$, defined as $\sigma(f) = (\sigma_i(f))_{i \in \mathbb{Z}_{16}^*}$,
where $\sigma_i(f) = f(\omega_{16}^i)$, and $\omega_{16} \in \mathbb{C}$ is the root of equation $x^{16} - 1 = 0$.
Note that $\sigma(f)$ is essentially the Fourier transformation of the coefficients of the polynomial $f$.
Let $\odot$ denote the component-wise multiplication over $\mathbb{C}^8$.
Then $\sigma$ is a ring homomorphism. In addition, one half of $\sigma_i$ is the complex conjugation of the other half, i.e.
$\sigma_i(f) = \overline{\sigma_{16 - i}(f)}$. If we constrain $\sigma$ on the first half, i.e.
$(\sigma_i)_{i\in\mathbb{Z}_{16}^*, i\le8} \in \mathbb{C}^4$, and regard $\mathbb{C}^4$ as a real space $\mathbb{R}^8$,
then $\sigma$ is a scaled rotation.

\begin{proof}[Proof of Theorem~\ref{th-simple}]
From the scaling property and the homomorphism property of $\sigma$,
we have
\begin{align}
\left\|\sum_{i = 1}^{128} a_i(x) b_i(x)\right\|^2 &=
\frac{1}{8} \left\|\sigma\left(\sum_{i = 1}^{128} a_i(x) b_i(x)\right)\right\|^2 =
\frac{1}{8}\left\|\sum_{i = 1}^{128} \sigma(a_i) \odot \sigma(b_i)\right\|^2 \\
&=\frac{1}{4} \sum_{j\in\mathbb{Z}_{16}^{*}, i \le 8} \left|\sum_{i = 1}^{128} \sigma_{j}(a_i) \odot \sigma_{j}(b_i)\right|^2
\end{align}
The last equality comes from the complex conjugation property of $\sigma_j$.
As $\sigma$ is a rotation, it transforms a Gaussian to a Gaussian.
Hence, $\sigma_j(a_i), {j \in \mathbb{Z}_{16}^*}$ are subjected to Gaussian of deviation $2\sigma$,
and are indepedent.

Then, conditioned on $b_i, i = 1...128$,
$\sum_{i = 1}^{128} \sigma_j(a_i) \odot \sigma_j(b_i)$ is subjected to a 2-dimensional Gaussian of variance $(2\sigma)^2 \sum_{i = 1}^{128} |\sigma_j(b_i)|^2$.
This deviation is subjected to $(2\sigma)^4\chi^2(256)$.
Note that the Gaussian of variance $(2\sigma)^4\chi^2(256)$ can be generated by multiplying the standard deviation to the standard Gaussian.
Hence, $\left|\sum_{i = 1}^{128} \sigma_{j}(a_i) \odot \sigma_{j}(b_i)\right|^2$ is subjected to
$(2\sigma)^4 x_j y_j$, where $x_j$ are subjected to $\chi^2(256)$ independently,
and $y_j$ are the squared norm of the Gaussian, which are subjected to $\chi^2(2)$.
\end{proof}

By Theorem~\ref{th-simple}, we only need to numerically calculate the distribution of $\sum_{i = 1}^4 x_i y_i$,
where $x_i$'s are subjected to $\chi^2(256)$ and $y_i$'s are taken from $\chi^2(2)$ independently.

We design a numerical algorithm to calculate the distribution.
In detail, we use $\mathsf{discretization}$ to transform the $\chi^2$-distribution into discrete distribution.
Then we represent the discrete distributions as dictionaries, which are denote as $distr_1, distr_2$ and $distr$ in the following.

\begin{itemize}
    \item $\mathsf{discretization}(m, s)$ outputs a dictionary that represent the distribution of $s\lfloor \chi^2(m) / s\rceil$.
    \item $\mathsf{multiply}(distr_1, distr_2)$ outputs a dictionary that represent the distribution of $x \cdot y$,
        where $x$ is subjected to distribution $distr_1$ and $y$ is subjected to distribution $distr_2$.
    \item $\mathsf{add}(distr_1, distr_2)$ outputs a dictionary that represent the distribution of $x + y$,
        where $x$ is subjected to distribution $distr_1$ and $y$ is subjected to distribution $distr_2$.
    \item $\mathsf{merge}(distr, s)$ outputs a dictionary that represent the distribution of $s\lfloor x / s\rceil$,
        where $x$ is subjected to the discrete distribution $distr$. This subroutine is used to reduce the number of entries
        in the dictionary $distr$, and thus speeds up the calculation.
\end{itemize}
This four subroutines are straightforwardly implemented.
$\mathsf{discretization}$ need to uses the cumulative function of the distribution of $\chi^2$.
This function can be numerically calculated preciously and is provided by many open libraries.

\noindent
\begin{algorithm}
    \caption{Calculate the distribution of $\sum_{i = 1}^4 x_i y_i$} \label{algo-distr}
\begin{algorithmic}[1]
    \State{$distr_{\chi^2(256)} \gets \mathsf{discretization}(256, 0.1)$}
    \State{$distr_{\chi^2(2)} \gets \mathsf{discretization}(2, 0.02)$}
    \State{$distr \gets \mathsf{multiply}(distr_{\chi^2(2)}, distr_{\chi^2(256)})$}
    \State{$distr \gets \mathsf{merge}(distr, 4)$}
    \State{$distr \gets \mathsf{add}(distr, distr)$}
    \State{$distr \gets \mathsf{add}(distr, distr)$}
\end{algorithmic}
\end{algorithm}

\begin{theorem}\label{th-error}
    Let $distr$ be the distribution outputed by Algorithm~\ref{algo-distr}.
    For every positive real number $T$,
    $\Pr[\sum_{i = 1}^4 x_i y_i > T] < \Pr[distr > T - 64]$.
\end{theorem}

\begin{proof}
    We analyze the errors of each of the four operations. $\mathsf{multiply}$ and $\mathsf{add}$ have no truncation errors.
    For $\mathsf{discretization}$ and $\mathsf{merge}$, we have $|x - s\lfloor x / s\rceil| = s|x/s - \lfloor x/s\rceil| < s/2$.
    Let $e^d_{x_i}, e^d_{y_i}$ be the errors produced by $\mathsf{discretization}$,
    and $e^m_{k}$ be the errors produced by $\mathsf{merge}$.
    Then we have
    \begin{equation*}
        \left|\sum_{i=0}^3 ((x_i + e^d_{x_i})(y_i + e^d_{y_i}) + e^m_{k}) - \sum_{i=0}^3 x_i y_i\right|
        \le \sum_{i = 0}^3 |e^d_{x_i}| y_i + |e^d_{y_i}| x_i + |e^d_{x_i} e^d_{y_i}| + |e^m_{k}|
    \end{equation*}
    In our Algorithm~\ref{algo-distr}, we have $|e^d_{y_i}| < 0.01$ and $|e^d_{x_i}| < 0.1$.
    In addition, $\Pr[x_i > 600] < 2^{-180}$ and $\Pr[y_i > 150] < 2^{-190}$.
    So we can simply assume $x_i < 600$ and $y_i < 150$.
    Hence, the right hand side is less than $4\cdot(600\cdot0.01 + 150\cdot0.05 + 0.01\cdot0.05 + 2) < 64$.
\end{proof}

\begin{claim}
    For $n = 512, \sigma^2 = 22, g = 2^6$, the \emph{over all} error probability of the scheme is less than $2^{-58}$.
\end{claim}

\begin{proof}
    From Claim~\ref{claim-bnd} and Theorem~\ref{th-correct}, we need to calculate the probability of
    $\|(\mathbf{e}_2 \mathbf{x}_1 - \mathbf{e}_1 \mathbf{x}_2)_k\|_{q, 2} > (q - 1)/2 - \sqrt{2}(q/g - 1) - 10\sigma > 5824$.
    From Theorem~\ref{th-simple}, this equals to the probability of $\sum_{i = 0}^3 x_i y_i > 5824^2/(4\cdot22^2) = 17520$.
    From Theorem~\ref{th-error}, this probability is less than $\Pr[distr > 17520 - 64 = 17456]$.
    The computer program numerically calculates $\Pr[distr > 17456] < 2^{-64.6}$.
    We use the \emph{union bound} to derive the \emph{over all} error probability is less than $64 \cdot 2^{-64.6} < 2^{-58}$.
\end{proof}

\newcommand{\etalchar}[1]{$^{#1}$}

\appendix

\section{Consensus Mechanism of Frodo}\label{app:bcd}
Let the modulo $q$  be power of $2$, which can be generalized to arbitrary modulo using the techniques in \cite{peikert2014lattice}.
Let integer $B$ be a power of $2$.
$B < (\log{q}) - 1, \bar{B} = (\log{q}) - B$
(note that $m = 2^{B}$ in our notations).
The underlying KC mechanism implicitly in Frodo is presented in Figure \ref{alg:con-BCD-power2}. 
\begin{algorithm}[H]
    \caption{Key consensus scheme in Frodo}\label{alg:con-BCD-power2}
\begin{algorithmic}[1]
\Procedure{Con}{$\sigma_1, \mathsf{params}$}\Comment{$\sigma_1\in [0, q)$}
    \State{$v = \left\lfloor 2^{-\bar{B}+1}\sigma_1 \right\rfloor \bmod 2$}
    \State{$k_1 = \left\lfloor 2^{-\bar{B}}\sigma_1 \right\rceil \bmod 2^{B}$}
    \Return{$(k_1, v)$}
\EndProcedure

\Procedure{Rec}{$\sigma_2, v, \mathsf{params}$}\Comment{$\sigma_2\in [0, q)$}
    \State{find $x \in \mathbb{Z}_q$ closest to $\sigma_2$ s.t.
          $\left\lfloor 2^{-\bar{B} + 1} x\right\rfloor \bmod 2 = v$}
    \State{$k_2 = \left\lfloor 2^{-\bar{B}}x \right\rceil \bmod 2^{B}$}
    \Return{$k_2$}
\EndProcedure
\end{algorithmic}
\end{algorithm}

\begin{claim}[\cite{bcd16}, Claim 3.2]
    If $|\sigma_1 - \sigma_2|_q < 2^{\bar{B} - 2}$, then $\textsf{Rec}(\sigma_2, v) =
    k_1$. i.e. the scheme in Algorithm  \ref{alg:con-BCD-power2} is correct.
\end{claim}
This claim is equivalence to require $4md < q$.

\section{Consensus Mechanism of NewHope}\label{sec-nh}
Note that, for the consensus mechanism of NewHope, the $rec$ procedure is  run both in $\textsf{Con}$ and in $\textsf{Rec}$, and a random bit $b$ is used in $\textsf{Con}$  corresponding to the dbl trick in \cite{peikert2014lattice}.

\begin{algorithm}[H]
\caption{NewHope  Consensus Mechanism}\label{kc-nh}
\begin{algorithmic}[1]

\Procedure{Decode}{$\mathbf{x} \in \mathbb{R}^4/\mathbb{Z}^4$}
\Comment{Return a bit $k$ such that $k\mathbf{g}$ is closest to $\mathbf{x} + \mathbb{Z}^4$}
    \State{$\mathbf{v} = \mathbf{x} - \lfloor \mathbf{x} \rceil$}
    \Return{$k=0$ if $\|\mathbf{v}\|_1 \le 1$, and $1$ otherwise}
\EndProcedure\\

\State{$\textsf{HelpRec}(\mathbf{x}, b) =
\mathsf{CVP}_{\tilde{D}_4}\left(\frac{2^r}{q}(\mathbf{x} + b\mathbf{g})\right) \bmod 2^r$}
\Comment{$b$ corresponds to the dbl trick \cite{peikert2014lattice}}

\State{$rec\left(\mathbf{x} \in \mathbb{Z}_q^4, \mathbf{v}\in\mathbb{Z}_{2^r}^4\right) =
\mathsf{Decode}\left(\frac{1}{q}\mathbf{x} - \frac{1}{2^r}\mathbf{B}\mathbf{v}\right)$}
\\

\Procedure{\textsf{Con}}{$\boldsymbol{\sigma}_1 \in \mathbb{Z}_q^4, \textsf{params}$}
    \State{$b \gets \{0, 1\} $}
    \State{$\mathbf{v} \gets \textsf{HelpRec}(\boldsymbol{\sigma}_1, b) $}
    \State{$k_1 \gets rec(\boldsymbol{\sigma}_1, \mathbf{v}) $} 
    \Return{$(k_1, \mathbf{v})$}
\EndProcedure\\


\Procedure{\textsf{Rec}}{$\boldsymbol{\sigma}_2 \in \mathbb{Z}_q^4, \mathbf{v}\in\mathbb{Z}_{2^r}^4, \textsf{params}$}
    \State{$k_2 \gets rec(\boldsymbol{\sigma}_2, \mathbf{v}) $}
\EndProcedure\\


\end{algorithmic}
\end{algorithm}

\section{Security Analysis of LWE-Based Key Exchange}\label{App-LWE}

\begin{definition} A KC or AKC based key exchange protocol from LWE  is \emph{secure},
    if for any sufficiently large security parameter $\lambda$ and any PPT adversary $\mathcal{A}$,
    $\left|\Pr[b^\prime=b]-\frac{1}{2}\right|$ is negligible,  as defined w.r.t. game  $G_0$ specified in Algorithm \ref{G0}.

\begin{algorithm}[H]
\caption{Game $G_0$}\label{G0}
\begin{algorithmic}[1]
\State{$\mathbf{A} \gets \mathbb{Z}^{n \times n}_q$}

\State{$\mathbf{X}_1, \mathbf{E}_1 \gets \chi^{n \times l_A}$}
\State{$\mathbf{Y}_1 = \mathbf{A} \mathbf{X}_1 + \mathbf{E}_1$}

\State{$\mathbf{X}_2, \mathbf{E}_2 \gets \chi^{n \times l_B}$}
\State{$\mathbf{Y}_2 = \mathbf{A}^T \mathbf{X}_2 + \mathbf{E}_2$}
\State{$\mathbf{E}_\sigma \gets \chi^{l_A \times l_B}$}
\State{$\boldsymbol{\Sigma}_2 = \mathbf{Y}_1^T \mathbf{X}_2 + \mathbf{E}_\sigma$}

\State{$\left(\mathbf{K}_2^0, \mathbf{V}\right) \gets \textsf{Con}(\boldsymbol{\Sigma}_2, \textsf{params})$}
\State{$\mathbf{K}_2^1 \gets \mathbb{Z}_m^{l_A \times l_B}$}
\State{$b \gets \{0, 1\}$}
\State{$b' \gets \mathcal{A}(\mathbf{A}, \mathbf{Y}_1, \lfloor \mathbf{Y}_2/2^t \rfloor, \mathbf{K}_2^b, \mathbf{V})$}
\end{algorithmic}
\end{algorithm}
\end{definition}

Before starting to prove the security, we first recall some basic properties of the LWE assumption.  The following lemma is derived by a direct hybrid argument \cite{PVW08,bcd16}.


\begin{lemma}[LWE in the matrix form]
    For positive integer parameters $(\lambda,n, q \ge 2, l, t)$, where $n,q,l,t$ all are polynomial in $\lambda$,
    and a distribution $\chi$ over $\mathbb{Z}_q$,
    denote by $L_{\chi}^{(l, t)}$ the distribution over
    $\mathbb{Z}_q^{t \times n} \times \mathbb{Z}_q^{t \times l}$
    generated by taking
    $\mathbf{A} \gets \mathbb{Z}_q^{t \times n},
    \mathbf{S} \gets \chi^{n \times l},
    \mathbf{E} \gets \chi^{t \times l}$ and outputting $(\mathbf{A}, \mathbf{A}\mathbf{S} + \mathbf{E})$.
    Then, under the standard LWE assumption on indistinguishability between $A_{q, \mathbf{s}, \chi}$ (with $\mathbf{s} \gets \chi^n$) and
    $\mathcal{U}(\mathbb{Z}_q^{n} \times \mathbb{Z}_q)$,  no PPT distinguisher $\mathcal{D}$ can distinguish, with non-negligible probability,
    between the  distribution $L_{\chi}^{(l, t)}$ and
    $\mathcal{U}(\mathbb{Z}_q^{t \times n} \times \mathbb{Z}_q^{t \times l})$ for sufficiently large $\lambda$.
\end{lemma}

\begin{theorem}\label{LWE-security}
    If $(\textsf{params}, \textsf{Con}, \textsf{Rec})$ is a \emph{correct} and  \emph{secure} KC or AKC scheme, the key exchange protocol described in \figurename~\ref{kex:lwe} is \emph{secure} under the (matrix form of) LWE assumption.
\end{theorem}

\begin{proof} 
The proof is similar to, but actually simpler than,  that in \cite{peikert2014lattice,bcd16}.
The general idea  is that we construct a sequence of games: $G_0$,
$G_1$ and $G_2$, where $G_0$ is the original game for defining security. In every move from game $G_i$
to $G_{i+1}$, $0\leq i\leq 1$,  we change a little.
All games $G_i$'s share the  same PPT adversary $\mathcal{A}$, whose goal is to distinguish between the  matrices chosen uniformly at random and the matrices generated in the actual  key exchange protocol.
Denote  by $T_i$, $0\leq i\leq 2$, the event that $b = b'$ in   Game $G_i$. Our goal is to prove that  $\Pr[T_0] < 1/2 + negl$, where $negl$ is a negligible function in $\lambda$.  For ease of readability, we re-produce game $G_0$ below. For presentation simplicity, in the subsequent analysis, we always assume the underlying KC or AKC is correct. The proof can be trivially extended to the case that correctness holds with overwhelming probability (i.e., failure occurs with negligible probability).

\scalebox{0.9} {\Large \centering \hspace{-0.35cm}
\begin{minipage}[H]{0.5\textwidth}
\null
\begin{algorithm}[H]
\caption{Game $G_0$}
\begin{algorithmic}[1]
\State{$\mathbf{A} \gets \mathbb{Z}^{n \times n}_q$}

\State{$\mathbf{X}_1, \mathbf{E}_1 \gets \chi^{n \times l_A}$}
\State{$\mathbf{Y}_1 = \mathbf{A} \mathbf{X}_1 + \mathbf{E}_1$}

\State{$\mathbf{X}_2, \mathbf{E}_2 \gets \chi^{n \times l_B}$}
\State{$\mathbf{Y}_2 = \mathbf{A}^T \mathbf{X}_2 + \mathbf{E}_2$}

\State{$\mathbf{E}_\sigma \gets \chi^{l_A \times l_B}$}
\State{$\boldsymbol{\Sigma}_2 = \mathbf{Y}_1^T \mathbf{X}_2 + \mathbf{E}_\sigma$}
\State{$\left(\mathbf{K}_2^0, \mathbf{V}\right) \gets \textsf{Con}(\boldsymbol{\Sigma}_2, \textsf{params})$}
\State{$\mathbf{K}_2^1 \gets \mathbb{Z}_m^{l_A \times l_B}$}
\State{$b \gets \{0, 1\}$}
\State{$b' \gets \mathcal{A}(\mathbf{A}, \mathbf{Y}_1, \lfloor \mathbf{Y}_2/2^t \rfloor, \mathbf{K}_2^b, \mathbf{V})$}
\end{algorithmic}
\end{algorithm}
\end{minipage}
\begin{minipage}[H]{0.5\textwidth}
\null
\begin{algorithm}[H]
\caption{Game $G_1$}
\begin{algorithmic}[1]
\State{$\mathbf{A} \gets \mathbb{Z}^{n \times n}_q$}

\State{$\mathbf{X}_1, \mathbf{E}_1 \gets \chi^{n \times l_A}$}
\State{{\color{red} $\mathbf{Y}_1 \gets \mathbb{Z}_q^{n \times l_A}$}}

\State{$\mathbf{X}_2, \mathbf{E}_2 \gets \chi^{n \times l_B}$}
\State{$\mathbf{Y}_2 = \mathbf{A}^T \mathbf{X}_2 + \mathbf{E}_2$}

\State{$\mathbf{E}_\sigma \gets \chi^{l_A \times l_B}$}
\State{$\boldsymbol{\Sigma}_2 = \mathbf{Y}_1^T \mathbf{X}_2 + \mathbf{E}_\sigma$}
\State{$\left(\mathbf{K}_2^0, \mathbf{V}\right) \gets \textsf{Con}(\boldsymbol{\Sigma}_2, \textsf{params})$}
\State{$\mathbf{K}_2^1 \gets \mathbb{Z}_m^{l_A \times l_B}$}
\State{$b \gets \{0, 1\}$}
\State{$b' \gets \mathcal{A}(\mathbf{A}, \mathbf{Y}_1, \lfloor \mathbf{Y}_2 / 2^t \rfloor, \mathbf{K}_2^b, \mathbf{V})$}
\end{algorithmic}
\end{algorithm}
\end{minipage}
}

\begin{lemma}\label{th-reduction-1}
       $|\Pr[T_0] - \Pr[T_1]| < negl$, under the indistinguishability between $L_{\chi}^{(l_A, n)}$ and
    {$\mathcal{U}(\mathbb{Z}_q^{n \times n} \times \mathbb{Z}_q^{n \times l_A})$}.
\end{lemma}
\begin{proof}
    Construct a distinguisher $\mathcal{D}$, in Algorithm \ref{D-1}, who tries to distinguish  $L_{\chi}^{(l_A, n)}$ from
    $\mathcal{U}(\mathbb{Z}_q^{n \times n} \times \mathbb{Z}_q^{n \times l_A})$.
\begin{algorithm}[H]
\caption{Distinguisher $\mathcal{D}$}\label{D-1}
\begin{algorithmic}[1]
    \Procedure{$\mathcal{D}$}{$\mathbf{A}, \mathbf{B}$}
\Comment{$\mathbf{A} \in \mathbb{Z}_q^{n \times n}, \mathbf{B} \in \mathbb{Z}_q^{n \times l_A}$}
    \State{$\mathbf{Y}_1 = \mathbf{B}$}

    \State{$\mathbf{X}_2, \mathbf{E}_2 \gets \chi^{n \times l_B}$}
    \State{$\mathbf{Y}_2 = \mathbf{A}^T \mathbf{X}_2 + \mathbf{E}_2$}
    \State{$\mathbf{E}_\sigma \gets \chi^{l_A \times l_B}$}
    \State{$\boldsymbol{\Sigma}_2 = \mathbf{Y}_1^T \mathbf{X}_2 + \mathbf{E}_\sigma$}
    \State{$\left(\mathbf{K}_2^0, \mathbf{V}\right) \gets \textsf{Con}(\boldsymbol{\Sigma}_2, \textsf{params})$}
    \State{$\mathbf{K}_2^1 \gets \mathbb{Z}_m^{l_A \times l_B}$}
    \State{$b \gets \{0, 1\}$}
    \State{$b' \gets \mathcal{A}(\mathbf{A}, \mathbf{Y}_1, \lfloor \mathbf{Y}_2/2^t \rfloor, \mathbf{K}_2^b, \mathbf{V})$}
    \If {$b' = b$}
        \Return $1$
    \Else
        \Return $0$
    \EndIf
\EndProcedure
\end{algorithmic}
\end{algorithm}

If $(\mathbf{A}, \mathbf{B})$ is subject to $L_\chi^{(l_A, n)}$, then
$\mathcal{D}$ perfectly simulates $G_0$. Hence,
$\Pr\left[\mathcal{D}\left(L_\chi^{(l_A, n)}\right) = 1\right] = \Pr[T_0]$.
On the other hand, if $(\mathbf{A}, \mathbf{B})$ is chosen uniformly at random  from
$\mathbb{Z}_q^{n \times n} \times \mathbb{Z}_q^{n \times l_A}$, which are denoted  as $(\mathbf{A}^\mathcal{U},
\mathbf{B}^\mathcal{U})$,
then $\mathcal{D}$ perfectly  simulates $G_1$. So, $\Pr[\mathcal{D}(\mathbf{A}^\mathcal{U}, \mathbf{B}^\mathcal{U}) = 1] = \Pr[T_1]$.
Hence,
$\left|\Pr[T_0] - \Pr[T_1]\right| =
\left|\Pr[\mathcal{D}(L_\chi^{(l_A, n)}) = 1] -
    \Pr[\mathcal{D}(\mathbf{A}^\mathcal{U}, \mathbf{B}^\mathcal{U}) = 1]\right| < negl$.
  \end{proof}

\noindent
\begin{minipage}[H]{0.5\textwidth}
\null
\begin{algorithm}[H]
\caption{Game $G_1$}
\begin{algorithmic}[1]
\State{$\mathbf{A} \gets \mathbb{Z}^{n \times n}_q$}
\State{$\mathbf{X}_1, \mathbf{E}_1 \gets \chi^{n \times l_A}$}
\State{$\mathbf{Y}_1 \gets \mathbb{Z}_q^{n \times l_A}$}
\State{$\mathbf{X}_2, \mathbf{E}_2 \gets \chi^{n \times l_B}$}
\State{$\mathbf{Y}_2 = \mathbf{A}^T \mathbf{X}_2 + \mathbf{E}_2$}
\State{$\mathbf{E}_\sigma \gets \chi^{l_A \times l_B}$}
\State{$\boldsymbol{\Sigma}_2 = \mathbf{Y}_1^T \mathbf{X}_2 + \mathbf{E}_\sigma$}
\State{$\left(\mathbf{K}_2^0, \mathbf{V}\right) \gets \textsf{Con}(\boldsymbol{\Sigma}_2, \textsf{params})$}
\State{$\mathbf{K}_2^1 \gets \mathbb{Z}_m^{l_A \times l_B}$}
\State{$b \gets \{0, 1\}$}
\State{$b' \gets \mathcal{A}(\mathbf{A}, \mathbf{Y}_1, \lfloor \mathbf{Y}_2/2^t \rfloor, \mathbf{K}_2^b, \mathbf{V})$}
\end{algorithmic}
\end{algorithm}
\end{minipage}
\begin{minipage}[H]{0.5\textwidth}
\null
\begin{algorithm}[H]
\caption{Game $G_2$}
\begin{algorithmic}[1]
\State{$\mathbf{A} \gets \mathbb{Z}^{n \times n}_q$}
\State{$\mathbf{X}_1, \mathbf{E}_1 \gets \chi^{n \times l_A}$}
\State{$\mathbf{Y}_1 \gets \mathbb{Z}_q^{n \times l_A}$}
\State{$\mathbf{X}_2, \mathbf{E}_2 \gets \chi^{n \times l_B}$}
\State{{\color{red}$\mathbf{Y}_2 \gets \mathbb{Z}_q^{n \times l_B}$}}
\State{$\mathbf{E}_\sigma \gets \chi^{l_A \times l_B}$}
\State{{\color{red}$\boldsymbol{\Sigma}_2 \gets \mathbb{Z}_q^{l_A \times l_B}$}}
\State{$\left(\mathbf{K}_2^0, \mathbf{V}\right) \gets \textsf{Con}(\boldsymbol{\Sigma}_2, \textsf{params})$}
\State{$\mathbf{K}_2^1 \gets \mathbb{Z}_m^{l_A \times l_B}$}
\State{$b \gets \{0, 1\}$}
\State{$b' \gets \mathcal{A}(\mathbf{A}, \mathbf{Y}_1, \lfloor \mathbf{Y}_2/2^t \rfloor, \mathbf{K}_2^b, \mathbf{V})$}
\end{algorithmic}
\end{algorithm}
\end{minipage}

\begin{lemma}\label{th-reduction-2}
$|\Pr[T_1] - \Pr[T_2]| < negl$, under the indistinguishability between $L_\chi^{(l_B, n + l_A)}$ and
$\mathcal{U}(\mathbb{Z}_q^{(n+l_A) \times n} \times \mathbb{Z}_q^{(n+l_A) \times l_B})$.
\end{lemma}
\begin{proof}
    As $\mathbf{Y}_1$ is subject to uniform distribution in $G_1$,
    $(\mathbf{Y}_1^T, \boldsymbol{\Sigma}_2)$ can be regarded as an $L_{\chi}^{(l_B, l_A)}$
    sample of secret $\mathbf{X}_2$ and noise $\mathbf{E}_{\sigma}$.
    Based on this observation, we construct the following distinguisher $\mathcal{D}^\prime$.
\noindent
\begin{algorithm}[H]
    \caption{Distinguisher $\mathcal{D}^\prime$} \label{D-12}
\begin{algorithmic}[1]
    \Procedure{$\mathcal{D}^\prime$}{$\mathbf{A'}, \mathbf{B}$}
where {$\mathbf{A'} \in \mathbb{Z}_q^{(n + l_A) \times n}, \mathbf{B} \in \mathbb{Z}_q^{(n + l_A) \times l_B}$}
\State{Denote $\mathbf{A'} = \left(
      \begin{array}{c}
        \mathbf{A}^T \\
        \mathbf{Y}_1^T
      \end{array}
    \right)
    $}
    \Comment{$\mathbf{A} \in \mathbb{Z}_q^{n \times n}, \mathbf{Y}_1^T \in \mathbb{Z}_q^{l_A \times n}$}
    \State{Denote $\mathbf{B} = \left(
      \begin{array}{c}
        \mathbf{Y}_2 \\
        \boldsymbol{\Sigma}_2
      \end{array}
        \right)$}
        \Comment{$\mathbf{Y}_2 \in \mathbb{Z}_q^{n \times l_B},
            \boldsymbol{\Sigma}_2 \in \mathbb{Z}_q^{l_A \times l_B}$}
        \State{$\left(\mathbf{K}_2^0, \mathbf{V}\right) \gets \textsf{Con}(\Sigma_2, \textsf{params})$}
        \State{$\mathbf{K}_2^1 \gets \mathbb{Z}_m^{l_A \times l_B}$}
        \State{$b \gets \{0, 1\}$}
        \State{$b' \gets \mathcal{A}(\mathbf{A}, \mathbf{Y}_1, \lfloor\mathbf{Y}_2/2^t\rfloor, \mathbf{K}_2^b, \mathbf{V})$}
        \If {$b' = b$}
            \Return $1$
        \Else
            \Return $0$
        \EndIf
    \EndProcedure
\end{algorithmic}
\end{algorithm}

{If $(\mathbf{A'}, \mathbf{B})$ is subject to $L_\chi^{(l_B, n + l_A)}$,
$\mathbf{A'} \gets \mathbb{Z}_q^{(n+l_A) \times n}$ 
corresponds to $\mathbf{A} \gets \mathbb{Z}_q^{n \times n}$ and $\mathbf{Y}_1 \gets \mathbb{Z}_q^{n \times l_A}$
in $G_1$;  and $\mathbf{S}\gets \chi^{n \times l_B}$ (resp., $\mathbf{E} \gets \chi^{(n+l_A) \times l_B}$)   in generating $(\mathbf{A'}, \mathbf{B})$ 
corresponds to $\mathbf{X}_2 \gets \chi^{n \times l_B}$ (resp.,  $\mathbf{E}_2 \gets \chi^{n \times l_B}$ and $\mathbf{E}_\sigma \gets \chi^{l_A \times l_B}$) in $G_1$.
 In this case, we have}
\begin{align*}
\mathbf{B} &= \mathbf{A'} \mathbf{S} + \mathbf{E} =
\left(
  \begin{array}{c}
    \mathbf{A}^T \\
    \mathbf{Y}_1^T
  \end{array}
\right)
\mathbf{X}_2 +
\left(
  \begin{array}{c}
    \mathbf{E}_2 \\
    \mathbf{E}_\sigma
  \end{array}
\right) \\
&=
\left(
  \begin{array}{c}
    \mathbf{A}^T \mathbf{X}_2 + \mathbf{E}_2 \\
    \mathbf{Y}_1^T \mathbf{X}_2 + \mathbf{E}_\sigma
  \end{array}
\right)
=
\left(
  \begin{array}{c}
    \mathbf{Y}_2 \\
    \boldsymbol{\Sigma}_2
  \end{array}
\right)
\end{align*}
Hence $\Pr\left[\mathcal{D}^\prime\left(L_\chi^{(l_B, n + l_A)}\right) = 1\right] = \Pr[T_1]$.

{On the other hand, if $(\mathbf{A'}, \mathbf{B})$ is subject to uniform distribution $\mathcal{U}(\mathbb{Z}_q^{(n+l_A) \times n} \times \mathbb{Z}_q^{(n+l_A) \times l_B})$,
then $\mathbf{A}, \mathbf{Y}_1, \mathbf{Y}_2, \boldsymbol{\Sigma}_2$ 
all are also
uniformly random; So,  the view of  $\mathcal{D}^\prime$ in this case is the same as that in  game $G_2$. Hence,
$\Pr\left[\mathcal{D}^\prime\left(\mathbf{A'},
        \mathbf{B}\right) = 1\right] = \Pr[T_2]$ in this case.
%
Then $|\Pr[T_1] - \Pr[T_2]| =|\Pr[\mathcal{D}^\prime(L_\chi^{(l_B, n + l_A)}) = 1] -
\Pr[\mathcal{D}^\prime(\mathcal{U}(\mathbb{Z}_q^{(n+l_A) \times n} \times \mathbb{Z}_q^{(n+l_A) \times l_B})) = 1]| < negl$.}
  \end{proof}

\begin{lemma}
    If the underlying KC or AKC is \emph{secure}, $\Pr[T2]=\frac{1}{2}$.
\end{lemma}

\begin{proof}
    Note that, in Game $G_2$,  for any
    $1 \le i \le l_A$ and $1 \le j \le l_B$, $\left(\mathbf{K}^0_2[i, j], \mathbf{V}[i, j]\right)$
    only depends on $\boldsymbol{\Sigma}_2[i, j]$, and $\boldsymbol{\Sigma}_2$
    is subject to uniform distribution.
    By the \emph{security} of KC, we have that, for each pair $(i,j)$,  $\mathbf{K}^0_2[i, j]$ and $\mathbf{V}[i, j]$ are independent, and
    $\mathbf{K}^0_2[i, j]$ is uniform distributed. Hence,  $\mathbf{K}^0_2$ and
    $\mathbf{V}$ are independent, and $\mathbf{K}^0_2$ is uniformly distributed, which implies that
    $\Pr[T_2] = 1/2$.
  \end{proof}

    This finishes the proof of Theorem \ref{LWE-security}.
  \end{proof}

\section{Construction and Analysis of AKCN-4:1}\label{app-4:1}

\subsection{Overview of NewHope}
By extending the technique of \cite{poppelmann2013towards}, in NewHope the coefficients of {$\boldsymbol{\sigma}_1$} (i.e., the polynomial of degree $n$) are divided into $n/4$ groups, where each group contains four coordinates. On the input of four coordinates, only one bit (rather than four bits) consensus is reached, which reduces the error probability to about $2^{-61}$ which is viewed to be negligible in practice.

Specifically, suppose Alice and Bob have $\boldsymbol{\sigma}_1$ and $\boldsymbol{\sigma}_2$
in $\mathbb{Z}_q^4$ respectively, and they are close to each other.
One can regard the two vectors as elements in $\mathbb{R}^4/\mathbb{Z}^4$,
by treating them as $\frac{1}{q}\boldsymbol{\sigma}_1$ and $\frac{1}{q}\boldsymbol{\sigma}_2$.
Consider the matrix $\mathbf{B} = (\mathbf{u}_0, \mathbf{u}_1, \mathbf{u}_2, \mathbf{g})\in {\mathbb{R}^{4\times 4}}$,
where $\mathbf{u}_i$, $0\leq i\leq 2$, is the canonical unit vector whose $i$-th coordinate is $1$, and $\mathbf{g} = (1/2, 1/2, 1/2, 1/2)^T$.
Denote by $\tilde{D}_4$ the lattice generated by $\mathbf{B}$. 
Note that $\mathbb{Z}^4 \subset \tilde{D}_4 \subset \mathbb{R}^4$.
{ Denote by $\mathcal{V}$ the close Voronoi  cell of the origin in $\tilde{D}_4$.
In fact, $\mathcal{V}$ is the intersection of the unit ball in norm 1 and the unit ball in infinity norm (the reader is referred to  NewHope~\cite[Appendix C]{newhope15} for details).
The following procedure  $\mathsf{CVP}_{\tilde{D}_4}(\mathbf{x})$ returns
the vector $\mathbf{v}$ such that $\mathbf{B}\mathbf{v}$ is closest to $\mathbf{x}$,  i.e., $\mathbf{x} \in \mathbf{B} \mathbf{v} + \mathcal{V}$,
where the distance is measured in the Euclidean norm.

\begin{algorithm}[H]
\caption{CVP$_{\tilde{D}_4}$ in NewHope \cite{newhope15}}\label{alg-cvp}
\begin{algorithmic}[1]
\Procedure{CVP$_{\tilde{D}_4}$}{$\mathbf{x} \in \mathbb{R}^4$}
    \State{$\mathbf{v}_0 = \lfloor \mathbf{x} \rceil$}
    \State{$\mathbf{v}_1 = \lfloor \mathbf{x} - \mathbf{g} \rceil$}
    \State{$k = 0$ if $\|\mathbf{x} - \mathbf{v}_0\|_1 < 1$ and $1$ otherwise}
    \State{$(v_0, v_1, v_2, v_3)^T = \mathbf{v}_k$}
    \Return{$\textbf{v}=(v_0, v_1, v_2, k)^T + v_3 \cdot (-1, -1, -1, 2)^T$}
\EndProcedure
\end{algorithmic}
\end{algorithm}

If $\boldsymbol{\sigma}_1$ is in the Voronoi cell of $\mathbf{g}$, then the consensus bit is set to be $1$, and $0$ otherwise.
Hence, 
Alice finds the closest lattice vector of
$\boldsymbol{\sigma}_1$ by running the $\mathsf{CVP}_{\tilde{D}_4}$ procedure described in Algorithm \ref{alg-cvp}, and calculates their difference which is set to be the hint signal $\mathbf{v}$.  Upon receiving $\mathbf{v}$, Bob  subtracts the difference from $\boldsymbol{\sigma}_2$. Since $\boldsymbol{\sigma}_1$ and $\boldsymbol{\sigma}_2$ are very close,
the subtraction moves $\frac{1}{q}\boldsymbol{\sigma}_2$ towards a lattice point in $\tilde{D}_4$.
Then Bob checks whether or not the point after the move is in the Voronoi cell of $\boldsymbol{g}$, and so the  consensus is reached.
Furthermore, to save bandwidth,  NewHope chooses an integer $r$, and discretizes the Voronoi cell of $\boldsymbol{g}$ to $2^{4r}$ blocks, so that only $4r$ bits are needed to transfer the hint information.
To make the distribution of consensus bit uniform, NewHope adds a small noise to $\boldsymbol{\sigma}_1$, similar to the dbl trick used in \cite{peikert2014lattice}.
The \textsf{Con} and $\textsf{Rec}$ procedures, distilled from NewHope, are presented in Algorithm \ref{kc-nh} in Appendix \ref{sec-nh}.

\subsection{Construction and Analysis of AKCN-4:1}

For any integer $q$ and vector $\mathbf{x} = (x_0,x_1, x_2, x_3)^T \in \mathbb{Z}_q^4$, denote by $\|\mathbf{x}\|_{q, 1}$ the sum $|x_0|_q+|x_1|_q + |x_2|_q + |x_3|_q$.
For two vectors $\mathbf{a} = (a_0, a_1, a_2, a_3)^T, \mathbf{b} = (b_0, b_1, b_2, b_3)^T \in \mathbb{Z}^4$,
let $\mathbf{a} \bmod \mathbf{b}$ denote the vector $(a_0 \bmod b_0, a_1 \bmod b_1, a_2 \bmod b_2, a_3 \bmod b_3)^T \in \mathbb{Z}^4$.
The scheme of AKCN-4:1 is presented in Algorithm \ref{kcs:4D}.

Compared with the consensus mechanism of NewHope presented in Appendix \ref{sec-nh}, AKCN-4:1 can be  simpler and computationally more efficient.
    In specific, the uniformly random bit $b$ used in NewHope (corresponding the dbl trick in \cite{peikert2014lattice}) is eliminated with AKCN-4:1, which saves 256 (resp., 1024) random bits in total  when reaching 256 (resp., 1024) consensus bits. In addition, as $k_1$, as well as $k_1(q+1)\mathbf{g}$,  can be offline computed and used (e.g., for encryption, in parallel with the protocol run),  AKCN-4:1  enjoys online/offline speeding-up and parallel computing.


\begin{theorem}\label{th:correct-AKCN41}
    If $\|\boldsymbol{\sigma}_1 - \boldsymbol{\sigma}_2\|_{q, 1} < q\left(1 - \frac{1}{g}\right) - 2$,
    then the AKCN-4:1 scheme depicted in Algorithm~\ref{kcs:4D} is \emph{correct}.
\end{theorem}

\begin{proof}
    Suppose $\mathbf{v}' = \mathsf{CVP}_{\tilde{D}_4}(g(\boldsymbol{\sigma}_1 + k_1 (q+1) \mathbf{g})/q)$.
   Then, $\mathbf{v} = \mathbf{v}' \bmod (g, g, g, 2g)$,
    and so there exits $\boldsymbol{\theta} = (\theta_0, \theta_1, \theta_2, \theta_3) \in \mathbb{Z}^4$
    such that $\mathbf{v} = \mathbf{v}' + g(\theta_0, \theta_1, \theta_2, 2\theta_3)^T.$
    From the formula calculating $\mathbf{v}'$, we know there exits $\boldsymbol{\epsilon} \in \mathcal{V}$,
    such that $g(\boldsymbol{\sigma}_1 + k_1 (q+1)\mathbf{g})/q = \boldsymbol{\epsilon} + \mathbf{B} \mathbf{v}'$.
    Hence, $\mathbf{B} \mathbf{v}' = g(\boldsymbol{\sigma}_1 + k_1 (q+1)\mathbf{g})/q - \boldsymbol{\epsilon}$.

    From the formula computing $\mathbf{x}$ in $\mathsf{Rec}$,
    we have $\mathbf{x} = \mathbf{B}\mathbf{v} / g - \boldsymbol{\sigma}_2 / q
    = \mathbf{B} \mathbf{v}' / g - \boldsymbol{\sigma}_2 / q + \mathbf{B}(\theta_0, \theta_1, \theta_2, 2\theta_3)^T
    = k_1 \mathbf{g} + k_1 \mathbf{g} / q - \boldsymbol{\epsilon} / g + (\boldsymbol{\sigma}_1 - \boldsymbol{\sigma}_2) / q
    + \mathbf{B}(\theta_0, \theta_1, \theta_2, 2\theta_3)^T$.
    Note that the last term $\mathbf{B}(\theta_0, \theta_1, \theta_2, 2\theta_3)^T \in \mathbb{Z}^4$,
    and in line~\ref{AKCN41:x} of Algorithm~\ref{kcs:4D} we subtract $\lfloor \mathbf{x} \rceil \in \mathbb{Z}^4$ from $\mathbf{x}$,
    so the difference between $\mathbf{x} - \lfloor \mathbf{x} \rceil $ and $k_1\mathbf{g}$ in norm 1 is no more than
    $2/q + 1 / g + \|\boldsymbol\sigma_1 - \boldsymbol\sigma_2\|_{q, 1} / q < 1$.
    Hence, $k_2 = k_1$.
  \end{proof}

\begin{algorithm}[H]
    \caption{AKCN-4:1}
    \label{kcs:4D}
\begin{algorithmic}[1]
    \Procedure{Con}{$\boldsymbol{\sigma}_1 \in \mathbb{Z}_q^4, k_1 \in \{0, 1\}, \textsf{params}$}
    \State{$\mathbf{v} = \mathsf{CVP}_{\tilde{D}_4}(g(\boldsymbol{\sigma}_1 + k_1 (q+1) \mathbf{g})/q) \bmod (g, g, g, 2g)^T$}
    \Return{$\mathbf{v}$}
    \EndProcedure
    \Procedure{Rec}{$\boldsymbol{\sigma}_2 \in \mathbb{Z}_q^4, \boldsymbol{v} \in
        \mathbb{Z}_g^3 \times \mathbb{Z}_{2g}, \textsf{params}$}
    \State{$\mathbf{x} = \mathbf{B}\mathbf{v}/g - \boldsymbol{\sigma}_2/q$}
    \Return{$k_2 = 0$ if $\|\mathbf{x} - \lfloor \mathbf{x} \rceil\|_1 < 1$, 1 otherwise.} \label{AKCN41:x}
    \EndProcedure
\end{algorithmic}
\end{algorithm}

\begin{theorem}
    AKCN-4:1  depicted in Algorithm~\ref{kcs:4D} is \emph{secure}. Specifically,
    if $\boldsymbol{\sigma}_1$ is subject to uniform distribution over
    $\mathbb{Z}_q^4$, then $\mathbf{v}$ and $k_1$ are independent.
\end{theorem}

\begin{proof}
    Let $\mathbf{y} = (\boldsymbol{\sigma}_1 + k_1 (q + 1) \mathbf{g}) \bmod q \in \mathbb{Z}_q^4$.
    First we prove that $\mathbf{y}$ is independent of $k_1$, when $\boldsymbol{\sigma}_1 \gets \mathbb{Z}_q^4$.
    Specifically, for arbitrary $\tilde{\mathbf{y}} \in \mathbb{Z}_q^4 $ and arbitrary $\tilde{k}_1 \in \{0, 1\}$,
    we want to prove that $\Pr[\mathbf{y} = \tilde{\mathbf{y}} \mid k_1 = \tilde{k}_1]
    = \Pr[\boldsymbol{\sigma}_1 = (\tilde{\mathbf{y}} - k_1(q + 1)\mathbf{g}) \bmod q \mid k_1 = \tilde{k}_1]
    = 1/q^4$. Hence, $\mathbf{y}$ and $k_1$ are independent.

    For simplicity, denote by $\mathbf{G}$ the vector $(g, g, g, 2g)$.
    Map $\phi: \mathbb{Z}^4 \rightarrow \mathbb{Z}_g^3 \times \mathbb{Z}_{2g}$ is defined by
    $\phi(\mathbf{w}) = \mathsf{CVP}_{\tilde{D}_4}(g \mathbf{w} / q) \bmod \mathbf{G}$.
    We shall prove that, for any $\boldsymbol{\theta} \in \mathbb{Z}^4$,
    $\phi(\mathbf{w} + q\boldsymbol{\theta}) = \phi(\mathbf{w})$.
    By definition of $\phi$,
    $\phi(\mathbf{w} + q\boldsymbol{\theta}) = \mathsf{CVP}_{\tilde{D}_4}(g\mathbf{w}/q + g\boldsymbol{\theta}) \bmod \mathbf{G}$.
    Taking $\mathbf{x} = g\mathbf{w}/q + g\boldsymbol{\theta}$ into Algorithm~\ref{alg-cvp}, we have $\mathsf{CVP}_{\tilde{D}_4}(g\mathbf{w}/q + g\boldsymbol{\theta})
    =\mathsf{CVP}_{\tilde{D}_4}(g\mathbf{w}/q) + \mathbf{B}^{-1}(g\boldsymbol{\theta})$.
    It is easy to check that the last term $\mathbf{B}^{-1}(g\boldsymbol{\theta})$
    always satisfies $\mathbf{B}^{-1}(g\boldsymbol{\theta}) \bmod \mathbf{G} = 0$.

    From the above property of $\phi$, we have $\phi(\mathbf{y}) = \phi((\boldsymbol{\sigma}_1 + k_1 (q + 1) \mathbf{g}) \bmod q)
    = \phi(\boldsymbol{\sigma}_1 + k_1 (q + 1) \mathbf{g}) = \mathbf{v}$.
    As $k_1$ is independent of $\mathbf{y}$, and $\mathbf{v}$ only depends on $\mathbf{y}$, $k_1$ and $\mathbf{v}$ are independent.
  \end{proof}

\section{Implementing $\mathbf{H} \mathbf{x}^T$ in SEC with Simple Bit Operations}\label{app-Hxcode}
\lstset{basicstyle=\ttfamily,breaklines=true}
\begin{lstlisting}[language=C, caption=An implementation of $\mathbf{H} \mathbf{x}^T$ with C language]
uint16_t getCode(uint16_t x)
{
    uint16_t c, p;
    c = (x >> 4) ^ x;
    c = (c >> 2) ^ c;
    p = ((c >> 1) ^ c) & 1;
    x = (x >> 8) ^ x;
    c = (x >> 2) ^ x;
    p = (((c >> 1) ^ c) & 1) | (p << 1);
    x = (x >> 4) ^ x;
    p = (((x >> 1) ^ x) & 1) | (p << 1);
    x = (x >> 2) ^ x;
    p = (x & 1) | (p << 1);

    return p;
}
\end{lstlisting}

\section{A Note on  Lizard}\label{sec-lizard}
We note that the CPA-secure PKE scheme, Lizard, proposed in \cite{CKLS16} is actually  instantiated from our AKCN scheme presented in Algorithm \ref{kcs:2}, where the two close values are derived from  generating and exchanging   spLWE  and spLWR samples in an asymmetric way. Specifically, the public key is generated with spLWE samples, while ciphertext is generated with spLWR samples. However, the underlying  AKC mechanism in the spLWE/spLWR based PKE scheme analyzed in \cite{CKLS16} is actually an instantiation of our  AKCN scheme for the special case of $m|g|q$, where  $g$  (resp., $m$) in AKCN corresponds to $p$ (resp., $t$) in \cite{CKLS16}.

Let $X_{n, \rho, \theta}$ denote the set containing all
$n$ dimension vectors that have exactly $\theta$ non-zero components,
and each non-zero component is in $\{\pm{1}, \pm{2}, \dots, \pm{\rho}\}$.
spLWE problem is the LWE problem whose secret vector is drawn uniformly randomly from the set $X_{n, \rho, \theta}$.
\end{document}